\documentclass{amsart}

\newif\ifsubmission

\usepackage{amsthm,amsmath,amsfonts,amssymb}
\usepackage[numbers,sort&compress]{natbib}
\usepackage[colorlinks,citecolor=blue,urlcolor=blue]{hyperref}
\usepackage[]{graphicx}

\ifsubmission
\usepackage{xr}
\else
\fi

\usepackage{comment}
\usepackage{booktabs}
\usepackage{longtable}
\usepackage{array,bbm,multirow} 

\theoremstyle{plain}
\newtheorem{theorem}{Theorem}[section]
\newtheorem{lemma}{Lemma}[section]
\newtheorem{corollary}[theorem]{Corollary}
\newtheorem{proposition}[theorem]{Proposition}
\newtheorem{claim}[theorem]{Claim}

\theoremstyle{definition}
\newtheorem{definition}{Definition}
\newtheorem{assumption}{Assumption}

\theoremstyle{remark}
\newtheorem{remark}{Remark}

\begin{document}

\ifsubmission

\else

\title[]{Functional Estimation of Manifold-Valued Diffusion Processes}

\author{Jacob McErlean}

\address{Department of Mathematics, Duke University, Durham, 27708, NC, USA} 

\author{Hau-Tieng Wu}

\address{Courant Institute of Mathematical Sciences, New York University, New York, 10012, USA} 

\begin{abstract}
Nonstationary high-dimensional time series are increasingly encountered in biomedical research as measurement technologies advance. Owing to the homeostatic nature of physiological systems, such datasets are often located on, or can be well approximated by, a low-dimensional manifold. Modeling such datasets by manifold-valued It\^o diffusion processes has been shown to provide valuable insights and to guide the design of algorithms for clinical applications. 
In this paper, we propose Nadaraya-Watson type nonparametric estimators for the drift vector field and diffusion matrix of the process from one trajectory. Assuming a time-homogeneous stochastic differential equation on a smooth complete manifold without boundary, we show that as the sampling interval and kernel bandwidth vanish with increasing trajectory length, recurrence of the process yields asymptotic consistency and normality of the drift and diffusion estimators, as well as the associated occupation density. Analysis of the diffusion estimator further produces a tangent space estimator for dependent data, which has its own interest and is essential for drift estimation. 
Numerical experiments across a range of manifold configurations support the theoretical results.
\end{abstract}

\maketitle

\fi

\section{Introduction}  
\label{intro}

High-dimensional time series data have become increasingly pervasive across a range of quantitative disciplines, including but not exclusively economics, finance, geoscience, and medicine \cite{box2015time}. In such settings, multiple interrelated variables are recorded simultaneously over time, often at high sampling frequencies and over long horizons, resulting in complex, high-dimensional temporal datasets. A central inferential objective is to characterize and quantify the underlying system dynamics, either for scientific understanding, real-time monitoring and control, or improved prediction. This requires the development of statistical methods capable of adequately capturing both temporal dependence and cross-sectional interactions in the data. 

A motivating example from biomedicine is the repurposing of underutilized intraoperative signals, such as arterial blood pressure (ABP) \cite{wang2023arterial,Chang2025.10.12.25337819} and photoplethysmogram (PPG) \cite{ho2024variability}, to predict clinical outcomes. By leveraging physiological knowledge, these signals can be transformed into high-dimensional time series and interpreted as sensors probing the dynamics of the cardiovascular system. In particular, each cardiac cycle in an ABP signal can be viewed as the system's response to a heartbeat stroke, analogous to the sound produced when a drum is struck. This perspective evokes the classical inverse problem of whether one can ``hear the shape of a drum''. \cite{kac1966can}. However, the cardiovascular system is far more complex: it is adaptive rather than static, and its state evolves continuously over time. Consequently, variability across cardiac cycles \cite{malik1990heart,bruce1996temporal}, rather than any single waveform, encodes critical information about physiological status. This observation motivates representing ABP as a sequence of functions, each corresponding to a cardiac cycle, thereby forming a high-dimensional time series that captures temporal evolution. Despite this variability, homeostatic regulation \cite{modell2015physiologist} ensures that these dynamics are not arbitrary but constrained by some underlying physiological principles.

These considerations naturally motivate a geometric perspective for modeling such data. High-dimensional biomedical time series are typically nonstationary and arise from deterministic dynamics modulated by stochastic perturbations. The underlying dynamics often evolve on a low-dimensional, nonlinear manifold that is not directly observable. Empirically, these trajectories exhibit diffusion-like behavior over time, echoing reactiondiffusion models  \cite{lobanova2004running} used in cardiac electrophysiology and stochastic descriptions of heart rate variability \cite{von2019physiology}. While the underlying physiological mechanisms are complex, only partially accessible, and usually qualitative, we adopt a {\em phenomenological} approach and model the data as evolving on a manifold endowed with a diffusion process. This approach provides a principled framework for capturing intrinsic structure and enables the development of algorithms for quantifying the evolving physiological dynamics.

To capture the above mentioned characteristics, we model a high-dimensional time series $\mathcal{X}:=\{x_k\}_{k=1}^n\subset \mathbb{R}^p$, where $x_k$ is sampled at time $t_k$, using a stochastic differential equation (SDE). Let $(M,g)$ be a $d$-dimensional Riemannian manifold isometrically embedded in $\mathbb{R}^p$  through $\iota$ with $d \leq p$. Adapting Einstein summation convention, consider an It\^o semimartingale $X_t$ satisfying the {\em time-homogeneous} SDE in the It\^o differential form:
\begin{equation}\label{EQ: master equation latent space model}
dX_t=\nu(X_t)dt+\sigma_\alpha(X_t)\circ dW^\alpha_t=\mu(X_t)dt+\sigma_\alpha(X_t) dW^\alpha_t\in M\,,
\end{equation}
where $\nu\in \Gamma^\infty(TM)$ and $\mu\in \Gamma^\infty(TM)$ denote the {\em drift} in the Stratonovich and It\^o form, respectively, $\sigma_\alpha\in \Gamma^\infty(TM)$, $\alpha=1,\ldots,r$, $r\in \mathbb{N}$, denote the {\em diffusion} vector fields, $\circ$ indicates the Stratonovich stochastic integral, and $W_t$ is standard $r$-dim Brownian motion.   
We model $\mathcal{X}$ as discrete samples $x_k=\iota(X_{t_k})\in\mathbb{R}^p$ at $t_k=k\Delta$.  
Here, $M$ reflects the geometric structure that constrains $X_t$.

Although the model provides a structured representation of the underlying system, in practice neither the parameters of the SDE nor the parametrization of the manifold $M$ are known. Instead, we only have access to the observed time series data $\mathcal{X}$. 
This challenge has been widely studied in the literature, though often considering $M=\mathbb{R}^p$; that is, without incorporating a nonlinear manifold in the model. The univariate case, $M=\mathbb{R}^1$, is by now relatively well-understood; we refer the reader to \cite{ait2009handbook,bosq2012nonparametric} and references therein for a broad, though not exhaustive, overview. In the multivariate setting, $M=\mathbb{R}^p$, substantial progress has been made from various perspectives, addressing a wide range of related challenges; non-exhaustive examples include nonparametric kernel-based drift and diffusion estimators \cite{burgiere1993theoreme,fan2007aggregation,bandi_moloche_2018}, direct covariance approach \cite{hayashi2005covariance}, maximal likelihood approach \cite{ait2008closed}, sparsity approach \cite{boninsegna2018sparse}, Malliavian-Fourier approach \cite{malliavin2009fourier,akahori2023symmetric}, principal component analysis approach \cite{ait2019principal,chen2020five,bu2025nonparametric}, spectral approach \cite{gobet2004nonparametric,crommelin2011diffusion}, multiscale structure based estimators \cite{pavliotis2007parameter,crommelin2011diffusion}, Gaussian process approach \cite{duncker2019learning}, parametric approach \cite{sorensen2004parametric}, etc. 


Among the challenges posed by this framework, in this paper we focus on estimating the drift and diffusion from $\mathcal{X}$ under the manifold model \eqref{EQ: master equation latent space model}. Estimation of the occupation densities and tangent spaces arises as a byproduct. Our primary motivation is \cite{bandi_moloche_2018} and \cite{LOCHERBACH20081301} (in the continuous-time setting), which focuses on $M = \mathbb{R}^p$. There, Nadaraya-Watson-type nonparametric estimators for drift and diffusion are proposed, and their asymptotic properties are established using Harris recurrence structure. 
To recall the intuition behind the kernel regression approach in \cite{bandi_moloche_2018}, note that the increments $\frac{1}{\Delta}(x_{k+1}-x_k)$ and outer products $\frac{1}{\Delta}(x_{k+1}-x_{k})(x_{k+1}-x_{k})^\top$ act as noisy one-step estimators of the drift $\iota_*\mu(x_k)$ and diffusion matrix $\sum_{l=1}^r\iota_*\sigma_l(x_k)[\iota_*\sigma_l(x_k)]^\top$. Individually, these quantities are highly variable. Kernel regression stabilizes them by aggregating local estimators, assigning larger weights to observations near a target point $x$ via a kernel function applied to a distance scaled by bandwidth $h$. The resulting weighted averages yield consistent estimators of the drift and diffusion.

Our main contribution is multifold, circling around generalization of the kernel regression to the manifold setup with dependent dataset and establishing theoretical guarantees.
First, although one might attempt to directly apply the Euclidean estimator of \cite{bandi_moloche_2018}, arguing that manifolds are locally well approximated by affine subspaces, this intuition fails. Even under continuous observation, curvature generates a leading bias term in the Euclidean drift estimation. 
Specifically, with It\^o's formula, the Euclidean-embedded process
$Z_t := {\iota}(X_t)$ satisfies
\begin{align}
 dZ_t =  \Big( {\iota}_* \nu+ \frac{1}{2} D_{{\iota}_* \sigma_\alpha}({\iota}_* \sigma_\alpha) \Big)(Z_t) \, dt + {\iota}_* \sigma_\alpha(Z_t) \, dW^\alpha_t, \label{Intro section ito-form}
\end{align}
where ${\iota}_* \nu+ \frac{1}{2} D_{{\iota}_* \sigma_\alpha}({\iota}_* \sigma_\alpha)\neq \iota_*\mu={\iota}_* \nu+ {\iota}_*|_x\nabla_{\sigma_\alpha}\sigma_\alpha$ and $D_{{\iota}_* \sigma_\alpha}({\iota}_* \sigma_\alpha)$ is the covariant derivative of the vector field ${\iota}_* \sigma_\alpha$ along itself in $\mathbb{R}^p$ satisfying
$D_{{\iota}_* \sigma_\alpha}({\iota}_* \sigma_\alpha)={\iota}_*|_x\nabla_{\sigma_\alpha}\sigma_\alpha+\textup{I\!I}_x(\sigma_\alpha,\sigma_\alpha)$.
This computation shows that if we estimate the drift term using the Euclidean estimator, the normal component, $\textup{I\!I}_x(\sigma_\alpha,\sigma_\alpha)$, biases the estimator. See Section \ref{section Ito formula summary} for more details.
Therefore, the kernel regression procedure must be carefully designed. By contrast, curvature does not asymptotically affect diffusion estimation, and the diffusion matrix spans the tangent space under mild assumptions. This insight motivates our construction: we first estimate the tangent space using the diffusion estimator, and then recover the drift using the estimated tangent structure. The resulting tangent space estimator, derived from dependent data, appears to be new in the manifold learning literature and is of independent interest.
Second, the sampling scheme poses substantial challenges. Kernel regression requires sufficiently dense local sampling. In contrast to the i.i.d. framework commonly assumed in manifold learning \cite{singer2012vector,Tyagi2013,wuwu18}, our data consist of a {\em single} realization of the SDE, and its empirical distribution does not automatically provide uniform local coverage. To address this issue, we exploit Harris recurrence on manifolds, which guarantees the existence of an invariant measure serving as a surrogate sampling density. To analyze the proposed estimators, we need a generalized Nummelin splitting scheme for Harris recurrent processes, which may or may not have life cycles \cite{LOCHERBACH20081301}. Combined with the Darling-Kac result that quantifies how the time  $X_t$ spends in a region scales over long periods, we convert additive functionals associated with the estimators into integrals on the manifold with respect to $X_t$'s invariant measure. Under suitable conditions on bandwidth and sampling rates, together with careful analysis of the resulting Gaussian mixtures, we establish a central limit theorem for all proposed estimators as the observation horizon tends to infinity.

This paper is organized as follows. Section~\ref{section: math Model} presents the mathematical model. Section~\ref{section: proposed algorithm} details the proposed algorithm. Theoretical results are provided in Section~\ref{section: theory of proposed algorithm}, followed by numerical simulations in Section~\ref{section simulation label}. Section~\ref{section discussion label} discusses the findings and outlines future research directions. 
All proofs and technical details are given in the 
\ifsubmission
Online Supplement. 
\else
Appendix.
\fi
We adapt the following notation in what follows. $\xrightarrow[]{\hspace{0.1cm}d\hspace{0.1cm}}$ and 
$\xrightarrow[]{\hspace{0.1cm}p\hspace{0.1cm}}$ stand for convergence in distribution and convergence in probability.
For two sequences of random variables, $A_n = o_p(B_n)$ means $A_n/B_n \xrightarrow[]{\hspace{0.1cm}p\hspace{0.1cm}} 0$ and
$A_n = O_p(B_n)$ means $A_n/B_n$ is stochastically bounded as $n \rightarrow \infty$. The superscript
$\top$ means matrix transpose. The symbol $\otimes$ denotes the Kronecker product and $\mathbbm{1}_A$ denotes the indicator function of the set $A$. Other commonly used symbols are listed in Table \ref{Table: more symbols}.

\section{Mathematical Model}\label{section: math Model}

We now present our manifold-valued diffusion model \eqref{EQ: master equation latent space model} for high-dimensional time series data, viewed as observations over a finite time period from a stochastic dynamical system. Specifically, we model the data as samples from a solution to an SDE evolving as a diffusion process on a manifold. 

\begin{assumption} \label{manifold-ass}
Let $(M,g)$ be a $d$-dimensional, complete, connected, smooth Riemannian manifold without boundary, isometrically embedded into $\mathbb{R}^p$ via an embedding $\iota$. When the manifold is non-compact, assume the Ricci curvature is lower bounded; that is, $\texttt{Ric}\geq Kg$ for some $K\in \mathbb{R}$, and $\texttt{inj}_x\geq\delta>0$, where $\texttt{inj}_x$ is the injectivity radius of $x$, for all $x\in M$.
\end{assumption}

The Ricci curvature and injectivity assumptions hold automatically when $M$ is compact. In the noncompact case, the lower bound assumption on Ricci curvature helps control the diffusion kernel behavior, and the uniform positive lower bound assumption on the injectivity radius is imposed to preclude local geometric collapse so that the volume doubling property holds. 
Let $d_g(x,y)$ denote the geodesic distance between $x,y\in M$ and $dV_g$ denote the Riemannian volume measure. Denote by $B_b(M)$, $C(M)$, $C_0(M)$, and $C_c^\infty(M)$ the spaces of bounded measurable functions on $M$, continuous functions on $M$, continuous functions vanishing at infinity, and smooth functions with compact support, respectively.
Let $\Gamma(TM)$ and $\Gamma^k(TM)$, $k\in \mathbb{N}\cup\{\infty\}$, denote the space of continuous and $C^k$ vector fields on $M$. When $M$ is noncompact, denote $\widehat{M}:=M\cup \{\partial_M\}$ to be the one-point compactification of $M$; otherwise $\widehat{M}:=M$. 
Denote $\mathcal{B}(M)$ to be the Borel $\sigma$-algebra on $\widehat M$. 
For a sequence $x_n\in M$, we know $x_n\to \partial_M$ if and only if $\|x_n\|_{\mathbb{R}^q}\to \infty$ \cite{hsu_stoch_anal_mflds}. 
Denote $\mathcal{W}^r_0:=\left\{\omega \in C([0, \infty), \mathbb{R}^r)\right\}$ and $\{\mathcal{F}_t\}_{t > 0}$ its topological $\sigma$-field generated by Borel cylinders. 
Consider the probability space $(\Omega,\mathcal{F},\mathbb{P})$ and a $r$-dimensional standard Brownian motion $W_s=[W_s^1,\ldots,W_s^r]^\top$  with the initial distribution $\mu_0$. We then have a filtrated probability space $(\Omega,\mathcal{F}_*,\mathbb{P})$, where $\mathcal{F}_*:=\{\mathcal{F}_t|\,t\geq 0\}$ is a right continuous filtration so that $\mathcal{F}:=\lim_{t\to\infty}\mathcal{F}_t$ and $\mathcal{F}_t=\cap_{\epsilon>0}\sigma\{W_u|\,u\leq t+\epsilon\}$. Clearly, $W_t$ is adapted to $\mathcal{F}_*$.
Let $\mathcal{W}(M) := \{w\in C([0, \infty), \widehat M)|\,w(0)\in M,$ and if $w(t)=\partial_M$ then $w(t')=\partial_M$ for all $t'\geq t\}$,
and $\mathcal{B}_t(\mathcal{W}(M))$ be the topological $\sigma$-field generated by Borel cylinders. The {\em explosion time} is defined on $\mathcal{W}(M)$ by $e(w):=\inf\{t|\, w(t)=\partial_M\}$; that is, $\partial_M$ is the ``trap'' of $w$. When $M$ is compact, $e=\infty$ \cite{Stroock2010pde}. 

\subsection{Quick review of SDE on manifold and assumptions} \label{subsection review SDE}
Adapt notation from \eqref{EQ: master equation latent space model}. It is well-known \cite[Proposition 1.2.9]{hsu_stoch_anal_mflds} that there exists a unique strong solution, which is an $\mathbb{R}^p$-valued adapted semi-martingale $X_t$ supported on $M\subset \mathbb{R}^p$, up to its explosion time $e(X)$, to a time-homogeneous {\em SDE on manifold} in the Stratonovich form: 
\begin{align}
X_t = X_0 + \int_0^t \nu(X_s) ds + \int_0^t \sigma_\alpha(X_s) \circ dW^\alpha_s  \label{strat-SDE}
\end{align}
where we use Einstein summation convention, $r\in \mathbb{N}$, $W_t$ is the standard Brownian motion with $\mu_0=\delta_0$, the initial value $X_0 \in M$ follows the distribution $\lambda$ supported on $M$ and is independent of $W_t$, $\circ$ denotes the Stratonovich stochastic integral (e.g. \cite[e.g., 1.1.13]{hsu_stoch_anal_mflds}), and $\nu, \sigma_1, \cdots, \sigma_r \in \Gamma^\infty(TM)$. 
Recall that a strong solution to \eqref{strat-SDE} up to a stopping time $\tau$ is an $\mathcal{F}_*$-adapted random process $X_t \in \mathcal{W}(M)$ defined on $\Omega$ such that for any $f \in C^\infty(M)$ \cite[Definition 1.2.3]{hsu_stoch_anal_mflds}, 
 \begin{align}
 f(X_t) = f(X_0) + \int_0^t (\nu f)(X_s) ds + \int_0^t (\sigma_\alpha f)(X_s) \circ dW^\alpha_s\,,\ \ \ 0\leq t<\tau\,.\label{function-integral-sde-def}
 \end{align}
The SDE \eqref{strat-SDE} is often written in It\^o differential form as that in \eqref{EQ: master equation latent space model}. We present the Stratonovich form of the integral here because it respects the nonlinear geometry; that is, the Stratonovich SDE is intrinsic to the manifold and its driving vector fields transform naturally under diffeomorphisms via push-forward \cite[Prop. 1.2.4]{hsu_stoch_anal_mflds}.

Let $\mu_X$ denote the probability law on $\mathcal{W}(M)$ of the solution $X_t$ to \eqref{strat-SDE} and $\{P_t\}_{t\geq 0}$ denote the associated transition semigroup. Recall that for $S\in \mathcal{B}(M)$ and $\lambda=\delta_x$, where $x\in M$, for $t\geq 0$, $P_t(x,S):=\mu_X\{w\in \mathcal{W}(M)|\, w(t)\in S\}$ and $P_tf(x):=\mathbb{E}_xf(X_t)=\int_{\mathcal{W}(M)}f(\omega_t)d \mu_Y(\omega)$, where $f\in \mathcal{B}_b(M)$ and $\mathbb{E}_x$ is the expectation associated with $\mathbb{P}_x$. $\{P_t\}$ of the solution to \eqref{strat-SDE} is a strong Markov process 
\cite[IX.$\S$3-$\S$5]{Elworthy1982sde}.

With a coordinate system over a neighborhood $U\subset M$ of $x\in M$, 
locally there is a symmetric and non-negative definite matrix-valued function, $\pi(x):=(\pi^{ij}(x))\in \mathbb{R}^{d\times d}$, where 
\[
\pi^{ij}(x)=\sum_\alpha \sigma^{i}_\alpha(x)\sigma^{j}_\alpha(x)\,. 
\]
We call $\pi$ the {\em diffusion matrix}. 
$X_t$ is a diffusion process whose infinitesimal generator $L$ satisfies 
\[
Lf(x):=
\lim_{t\to 0}\frac{\mathbb{E}_x[f(X_t)]-f(x)}{t}=\frac{1}{2}\sum_{\alpha=1}^r \sigma_\alpha(\sigma_\alpha f)(x)+\mu f(x)\,,
\] 
where $f\in C_c^\infty(M)$; that is, $X_t$ is a $L$-diffusion process \cite[Section 1.3]{hsu_stoch_anal_mflds}. 
In local coordinates, the generator admits the representation 
\[
Lf(x)=\frac{1}{2}\pi^{ij}(x)\frac{\partial^2f}{\partial x^i\partial x^j}(x)+b^i(x)\frac{\partial f}{\partial x^i}(x)\,, 
\]
where $b^i:=\mu^{i}+\sum_\alpha \sigma^{k}_\alpha \partial_k \sigma^{i}_\alpha$. The additional $\sum_\alpha \sigma^{k}_\alpha \partial_k \sigma^{i}_\alpha$ in $b^i$ arises from the conversion between the It\^o form and Stratonovich formulations.

Harris recurrence is the key tool for our algorithm analysis. Recall that \cite[Definition 1.1]{limit_theorems_null} $X=(X_t)_{t\geq 0}$ with $X_0=x \in M$ a.s. is called \textit{Harris recurrent} if there exists a non-trivial $\sigma$-finite measure $\eta$ on $(M,\mathcal{B}(M))$ so that for any $A \in \mathcal{B}(M)$ so that
$\eta(A) > 0$, we have $\mathbb{P}_x (\int_0^\infty \mathbf{1}_A(X_t) dt = \infty) = 1$.
It is well known that a Harris recurrent process implies the existence of a unique invariant measure $\phi_X$ of $\{P_t\}$, up to a constant multiple \cite[Theorem 1.2]{limit_theorems_null}. When $\phi_X(M)=\infty$, we call $X_t$ \textit{null Harris recurrent}. When $\phi_X(M)<\infty$, we call $X_t$ \textit{positive Harris recurrent}. 
Geometrically, the process $X_t$ is Harris recurrent if, for every measurable subset of the manifold with positive volume, the process almost surely visits that set infinitely often (in time). 
For more details on Harris recurrence, we refer readers to \cite{limit_theorems_null}.

We impose the following assumption. 

\begin{assumption} \label{manifold-ass2}
Assume the initial measure of $X_t$, $\lambda$, is compactly supported on $\mathsf C\subset M$ with $dV_g(\mathsf C)>0$ and a bounded density function, and the spectrum of $\pi(x)$ is uniformly bounded from above and away from $0$; that is, $L$ is uniformly elliptic. 

When the manifold is non-compact, we impose further assumptions. 
\begin{enumerate}
\item (conservation) The explosion time is infinite; that is, $e(Y)=\infty$.
\item ($C^0$-property) $\{P_t\}$ is strong Feller. 
\item (diffusion kernel) The diffusion kernel $p_t(x,y)$ on $(0,\infty)\times M\times M$ associated with the transition semi-group is smooth and strictly positive.
\item (Harris recurrence) $X_t$ is Harris recurrent with an invariant measure $\phi_X$.
\end{enumerate}
\end{assumption}

Note that in \cite{bandi_moloche_2018}, rather than uniform ellipticity, the {\em H\"ormander condition} \cite[Chapter 7]{Stroock2010pde} is imposed. Instead, we adopt uniform ellipticity since we need nondegenerate diffusion in order to design our  drift estimator. The compact support assumption on the initial measure is motivated by biomedical applications, where it reflects the relatively homogeneous initial state of the dynamics.

\begin{remark}
Note that not every choice of $r$ is suitable for modeling and analyzing high-dimensional time series data if we need a {\em non-degenerate diffusion}, which is essential for recovering the drift term.
Recall that the well-known Hairy Ball theorem prohibits the existence of a smooth, non-vanishing vector field on $\mathbb{S}^2$. This topological constraint guarantees that choosing $r=2$ cannot yield a non-degenerate diffusion on $\mathbb{S}^2$. Therefore, depending on the manifold, we may need $r>d$.  
From a statistical perspective, taking $r>d$ may seem to introduce identifiability issues. In particular, it is possible to find $r,s\geq d$ with $r\neq s$ such that the diffusion terms $\sum_{\alpha=1}^r \sigma_\alpha(X_t) \circ dW^\alpha_t$ and $\sum_{\beta=1}^s \tilde \sigma_\beta(X_t) \circ d\tilde W^\beta_t$ have the same law on $M$ for some $\tilde\sigma_\beta$ and $\tilde W^\beta_t$, $\beta=1,\ldots, s$, that is independent of $W_t$. This is however not a big trouble in our analysis since our goal is estimating the diffusion matrix, {\em not} recovering each $\sigma_\alpha$. 
\end{remark}

Assumption \ref{manifold-ass2}(1) is necessary for our analysis to be asymptotic. Assumption \ref{manifold-ass2}(2)-(4) are needed for the construction of a recurrent atom \cite[Definition 1.9.A]{limit_theorems_null} for the Nummelin-like splitting argument via an embedding technique used in the analysis. Note that not all Harris recurrent processes have a recurrent atom. The Harris recurrence assumption effectively amounts to requiring sufficiently dense sampling {locally on the manifold over finite recording time}, analogous to the ``lower bounded density assumption'' commonly imposed in the analysis of manifold learning algorithms. See \cite[Chapter 1]{limit_theorems_null} for a summary of Harris recurrence.

When $M$ is compact, Assumption \ref{manifold-ass2}(1)-(4) hold automatically. Specifically, the existence of smooth kernel functions $p_t(x,y)$ on $(0,\infty)\times M\times M$ is ensured by the H\"ormander theorem, and $p_t(x,y)$ has Gaussian upper and lower bounds \cite[Theorem 6.4.1]{Stroock2010pde}, which leads to Assumption \ref{manifold-ass2}(3). Moreover, $\{P_t\}$ is strong Feller \cite[VIII.$\S$6]{Elworthy1982sde} and $X_t$ is positive Harris recurrent \cite{meyn1993stability}.
When $M$ is non-compact, various sufficient conditions ensure Assumption \ref{manifold-ass2}(1)-(4). To avoid distraction, we postpone these details to Section \ref{section sufficient conditions for assumption 2 when M is noncompact}. We shall mention that a non-compact manifold is necessary to host an It\^o diffusion with null Harris recurrence. As our primary focus is on estimating dynamics and quantifying the asymptotic behavior of the estimators, in the non-compact case we are content to work under Assumption \ref{manifold-ass2}(1)-(4) in this paper.

\subsection{Manifold-valued diffusion model and sampling for the high-dimensional time series}
We now detail the proposed manifold-valued diffusion model for the high-dimensional time series. We start with a definition and some nominations.

\begin{definition}[Manifold-valued diffusion model]\label{definition observed space model}
Suppose Assumption \ref{manifold-ass} for a manifold $M$ holds. We call $X_t$ on $M$ satisfying Assumption \ref{manifold-ass2} a {\em manifold-valued diffusion model}.
\end{definition}

The high-dimensional time series $\mathcal{X}=\{x_k\}_{k=1}^n\subset \mathbb{R}^p$ is modeled as the discretization of a realization of the solution $\iota(X_t)$ on the interval $[0, T]$ at uniform times; that is, $x_k=\iota(X_{k\Delta})$, where $\Delta >0$ denotes the sampling interval, $k=1,\ldots,n$, $n=\lfloor T/\Delta\rfloor$, and $\lfloor \cdot\rfloor$ is the floor operator.

\section{Proposed algorithm}
\label{section: proposed algorithm}

Under the manifold-valued diffusion model, we develop a class of manifold-adaptive, Nadaraya-Watson type kernel estimators to recover the invariant measure $\phi_X$, or the occupation density of $X_t$, and {``observed''} drift vector and diffusion matrix, denoted as 
\begin{align}
\mu^{(\texttt{o})}:=\iota_*\mu\in \mathbb{R}^p\ \ \mbox{ and }\ \ \pi^{(\texttt{o})}:=\sum_{l=1}^r(\iota_*\sigma_l)(\iota_*\sigma_l)^\top\in \mathbb{R}^{p\times p} \label{definition: target quantities 1}\,,
\end{align} 
from the high-dimensional time series $\mathcal{X}:=\{x_k\}_{k=1}^n$. 
The idea is straightforward: we first obtain noisy preliminary estimates of the drift and diffusion from process increments, then average them using a nonparametric kernel. This ``plug-in'' approach exploits variance reduction through averaging, yielding more stable and accurate estimates. In all cases, its success hinges on establishing a meaningful quantitative relationship between successive observations $x_{k+1}$ and $x_k$ by taking care of the curvature.

{In practice, the geodesic distance and tangent space of $M$ are not available and must be estimated from the data. For each state $x\in M$, prepare an appropriate ``distance-like'' function $\mathcal{D}_x:M\to \mathbb{R}^+$ defined by
\begin{equation}\label{definition Dy Euclidean}
\mathcal{D}_x(x'):=\|\iota(x')-\iota(x)\|_{\mathbb{R}^p}\,, 
\end{equation}
where $x'\in M$, which accurately estimates the geodesic distance between $x$ and $x'$ on $M$ when they are sufficiently close.
Choose a kernel $\mathcal{K}:\mathbb{R}^+\to \mathbb{R}$, which is smooth and compactly supported.
Define
\begin{align*} 
\hat{L}^{(\texttt{o})}(x) := \frac{\Delta}{h^d} \sum_{k=0}^{n-1} \mathcal{K} \left(\frac{\mathcal{D}_x(X_{k\Delta})}{h}\right)\,,
\end{align*}
which estimates the invariant measure $\phi_X$ at $x$.
With the chosen $\mathcal{D}_x$, estimate the   diffusion matrix at $x{\in M}$ from $\mathcal{X}$ by 
\begin{align}\label{Formula: generic diffusion term estimator observed}
\hat{\pi}^{(\texttt{o})}(x) &=\frac{1}{\Delta} \frac{\sum_{k=0}^{n-1}\mathcal{K}\left(\frac{\mathcal{D}_x(x_{k})}{h} \right)(x_{k+1}-x_{k})(x_{k+1}-x_{k})^\top}{\sum_{k=0}^{n-1}\mathcal{K}\left(\frac{\mathcal{D}_x(x_{k})}{h} \right)}\in \mathbb{R}^{p\times p}\,.
\end{align}
Note that the denominator is $\frac{h^d}{\Delta}\hat{L}^{(\texttt{o})}(x)$.

To estimate the drift term, we need to estimate the tangent space and a projection onto $\iota_*T_xM$, denoted as $\hat{P}_x$. 
Denote the eigenvalue decomposition of $\hat{\pi}^{(\texttt{o})}(x)$ as $\hat{U}_x\hat{\Sigma}_x\hat{U}^\top_x$, where the eigenvalues are ordered decreasingly and define
\begin{equation}\label{definition Py Euclidean}
\hat{P}_x:=\hat U_{d} \hat U_{d}^\top\,,
\end{equation}
where $\hat U_{d} \in \mathbb{R}^{p\times d}$ is formed from the first $d$ columns of $\hat{U}_x$. Clearly, $\hat U_{d}^\top \hat U_{d}=I_{d\times d}$ by construction. As we show below, the column space of $\hat U_{d}$ provides a good approximation of the embedded tangent space $\iota_* T_xM$. 
With the chosen projection matrix $\hat{P}_x$, estimate the   drift term by
\begin{align}\label{Formula: generic drift term estimator observed}
\hat{\mu}^{(\texttt{o})}(x) &=\hat{P}_x\hat{\mu}_E(x)\in \mathbb{R}^{p}\,,
\end{align}
where
\begin{align} 
\hat{\mu}_E(x):=\frac{1}{\Delta} \frac{\sum_{k=0}^{n-1}\mathcal{K}\left(\frac{\mathcal{D}_x(x_{k})}{h} \right)(x_{k+1}-x_{k})}{\sum_{k=0}^{n-1}\mathcal{K}\left(\frac{\mathcal{D}_x(x_{k})}{h} \right)}\in \mathbb{R}^{p}\,.
\end{align}
 }

When $M=\mathbb{R}^p$, our problem reduces to that studied in \cite{bandi_moloche_2018}. Indeed, choosing {$\hat{P}_x=I_{p\times p}$} makes the estimators in  \eqref{Formula: generic diffusion term estimator observed} and \eqref{Formula: generic drift term estimator observed}  coincide with those of \cite{bandi_moloche_2018}. 
In the manifold setting, particularly when $d<p$ and $M$ is nonlinear, the situation differs. Although the diffusion process can be embedded in Euclidean space, the estimator of \cite{bandi_moloche_2018}, denoted $\hat{\mu}_E$, targets the drift of the embedded process $Z_t$ rather than that of $X_t$; see \eqref{Intro section ito-form}. A simple illustration is the curve $M = \{(x,y): y = x^2 \}$. Driftless Brownian motion on $\mathbb{R}$ lifts naturally to $M$, yet a naive drift estimate at $(0,0)$ acquires a positive vertical component normal to $M$, even when geodesic distances on $M$ are used. To correct curvature-induced bias, a projection matrix $\hat{P}_x$ is necessary.

\begin{remark} \label{geodesic-distance-remark}
Recall that when $x$ and $x'$ are sufficiently close, the Euclidean distance $\|\iota(x')-\iota(x)\|_{\mathbb{R}^p}$ provides a good approximation to the geodesic distance \cite{malik2019connecting}{; that is, when $x'=\exp_x(t\theta)$, $\|\theta\|=1$, and $t$ sufficiently small, we have
\[
\|\iota(x)-\iota(x')\|_{\mathbb{R}^q}=t+\frac{1}{6}t^3\iota_*\theta^\top\nabla_\theta\texttt{I\!\!I}_x(\theta,\theta)+O(t^4)\,,
\]
where $\texttt{I\!\!I}_x$ is the second fundamental form at $x$.} More accurate approximations can be obtained through higher-order corrections involving estimating the second fundamental form \cite{malik2019connecting}, and in some settings geodesic information may even be available. 
However, the main difficulty arises not from estimating geodesic distances but from analyzing the increment $x_{k+1}- x_k$. In general, $x_{k+1}- x_k$ does not lie in $\iota_*T_{x}M$  and contains a normal component, which is the primary source of bias. For this reason, we do not pursue higher-order local approximations of the geodesic distance here. 
\end{remark}

\section{Asymptotic Analysis}\label{section: theory of proposed algorithm}

In this section, we study the asymptotic behavior of the estimators introduced in Section \ref{section: proposed algorithm}. We begin with intuition and relevant existing results before presenting main results.

Intuitively, since $X_t$ is continuous, our proposed estimators can be viewed as Riemann sum approximations of their target quantities. To be more specific, by It\^o's formula, we have 
\[
x_{k+1}-x_{k}=\iota(X_{(k+1)\Delta})-\iota(X_{k\Delta})=\int_{k\Delta}^{(k+1)\Delta}\mu\iota(X_s)ds+\int_{k\Delta}^{(k+1)\Delta}\sigma_{\alpha}\iota dW^\alpha_s\,, 
\]
and hence the numerator of $\hat{\mu}_E(x)$ becomes 
\[
\frac{1}{h^d}\sum_{k=0}^{n-1}\mathcal{K}\left(\frac{\mathcal{D}_x(x_{k})}{h} \right)\left(\int_{k\Delta}^{(k+1)\Delta}\mu\iota(X_s)ds+\int_{k\Delta}^{(k+1)\Delta}\sigma_{\alpha}\iota dW^\alpha_s\right)\,.
\]
By approximating $\int_{k\Delta}^{(k+1)\Delta}\mu\iota(X_s)ds$ by $\Delta\mu\iota(X_{k\Delta})$, we have:
\begin{align}
\hat{\mu}^{(\texttt{o})}(x) \approx\, &P_x\bigg( \frac{\frac{\Delta}{h^d}\sum_{k=0}^{n-1}\mathcal{K}\left(\frac{\mathcal{D}_x(x_{k})}{h} \right)\mu\iota(X_{k\Delta})}{\frac{\Delta}{h^d}\sum_{k=0}^{n-1}\mathcal{K}\left(\frac{\mathcal{D}_x(x_{k})}{h} \right)}+ \frac{\frac{1}{h^d}\sum_{k=0}^{n-1}\mathcal{K}\left(\frac{\mathcal{D}_x(x_{k})}{h} \right)\int_{k\Delta}^{(k+1)\Delta}\sigma_{\alpha}\iota dW^\alpha_s}{\frac{\Delta}{h^d}\sum_{k=0}^{n-1}\mathcal{K}\left(\frac{\mathcal{D}_x(x_{k})}{h} \right)}\bigg)\nonumber\,,
\end{align}
where the summation is the Riemann sum approximation; that is,
\begin{align}
\frac{\frac{\Delta}{h^d}\sum_{k=0}^{n-1}\mathcal{K}\left(\frac{\mathcal{D}_x(x_{k})}{h} \right)\mu\iota(X_{k\Delta})}{\frac{\Delta}{h^d}\sum_{k=0}^{n-1}\mathcal{K}\left(\frac{\mathcal{D}_x(x_{k})}{h} \right)}\approx \frac{\frac{1}{h^d}\int_0^T\mathcal{K}\left(\frac{\mathcal{D}_x(X_s)}{h} \right)\mu\iota(X_{s})ds}{\frac{1}{h^d}\int_0^T\mathcal{K}\left(\frac{\mathcal{D}_x(X_s)}{h}\right)ds} \label{equation Lhat approximat additive functional}\,.
\end{align}
The first term in \eqref{equation Lhat approximat additive functional} is the ratio of two integrations involving $X_t$. Recall that for bounded measurable $f\geq 0$, an $\mathbb{R}_+\cup\{\infty\}$ valued, $\mathcal{F}_t$-adapted process $A_t:=\int_0^t f(X_s) ds$, $t\geq 0$, with $A_0=0$, is called an {\em additive functional}. See Section \ref{section LL embedding} for more details. By definition, the denominator $\frac{1}{h^d}\int_0^T\mathcal{K}\left(\frac{\mathcal{D}_x(X_s)}{h}\right)ds$ is an additive functional, and the numerator can be handled similarly using additive functional properties after proper manipulations. In the proof, we quantify these approximations and show that this term, combined with $P_x$ approximates the desired drift. The second term in \eqref{equation Lhat approximat additive functional} involves stochastic integration, $\frac{1}{h^d}\sum_{k=0}^{n-1}\mathcal{K}\left(\frac{\mathcal{D}_x(x_{k})}{h} \right)\int_{k\Delta}^{(k+1)\Delta}\sigma_{\alpha}\iota dW^\alpha_s$, which we control via controlling its quadratic variation using its martingale property. In the end, we show that it converges to a normal distribution.
Since all proposed estimators share a similar structure, while some might be more complicated than others, the analysis reduces to controlling the Riemann sum approximation when $\Delta\to 0$, quantifying the behavior of the integrals when $h\to 0$ and $T\to \infty$, and analyzing the asymptotics of the stochastic integral terms. 
The ratio of two integrations involving $X_t$ reminds us of the ratio-limit theorem. 

\begin{theorem}[Ratio Limit Theorem]\cite[Theorem 1.7]{limit_theorems_null}
\label{ratio-limit0SUPP}
Suppose Assumptions \ref{manifold-ass} and \ref{manifold-ass2} hold. For any Borel-measurable, positive, and $\phi_X$-integrable $f,g: M \rightarrow \mathbb{R}$ such that $0<\langle \phi_X, g \rangle_M:=\int_M g(x)\phi_X(x) <\infty$, we have
\begin{align}
\mathbb{P}_x \left( \lim_{T \rightarrow \infty} \frac{\int_0^T f(X_s) \, ds}{\int_0^T g(X_s)\,ds} = \frac{\langle \phi_X, f \rangle_M}{\langle \phi_X, g \rangle_M}\right) = 1
\end{align}
for all $x \in M$. Moreover,  
\begin{align}
\lim_{T \rightarrow \infty} \frac{\mathbb{E}_x\int_0^T f(X_s) \, ds}{\mathbb{E}_x\int_0^T g(X_s)\,ds} = \frac{\langle \phi_X, f \rangle_M}{\langle \phi_X, g \rangle_M} 
\end{align}
$\phi_X$-a.s., where the exceptional set depends on $f$ and $g$.
\end{theorem}

This theorem shows that, in the long-time limit, an additive functional behaves like an integral over the manifold, {and it links our analysis with that in manifold learning via \eqref{equation Lhat approximat additive functional}}. In other words, the effect of the initial condition is effectively ``washed out'' asymptotically, linking the result to the ergodic theorem, with $\int_0^T g(X_s)\,ds$ acting as a random clock counting effective time. See Section \ref{section: Long-Time Asymptotics for Additive Functionals} for further details.

To quantify the asymptotic behavior of additive functionals, {we need to know} how often the system $X_t$ returns to a given state and how the time spent in a region scales over long periods. {When $X_t$ is Harris recurrent, these quantities} follow a well-defined probabilistic pattern that typically resembles a stable distribution rather than the usual bell-curve behavior. {This fact is quantified in the following general Darling-Kac theorem. To state this theorem, we need} the notion of a function that, while not necessarily bounded, behaves asymptotically like a constant.

\begin{definition} \label{reg-varying}
Let $\ell: \mathbb{R}^+ \rightarrow \mathbb{R}^+$. 
We say that $\ell$ is \textit{regularly varying at infinity with index $\alpha\in \mathbb{R}$} if for all $\zeta > 0$, $ \lim_{x \rightarrow \infty} \frac{\ell(\zeta x)}{\ell(x)} = \zeta^\alpha$.
If $\alpha=0$, $\ell(x)$ is called \textit{slowly varying at infinity}. We similarly say that $\ell$ is \textit{regularly varying at zero with index $\alpha$} if for all $\zeta > 0$, 
$ \lim_{x \rightarrow 0} \frac{\ell(\zeta x)}{\ell(x)} = \zeta^\alpha$.
\end{definition}

Denote $D(\mathbb{R}_+,\mathbb{R})$ to be the Skorohod space with Borel $\sigma$-algebra and canonical filtration. Note that the classical Darling-Kac theorem is the necessary direction of the following theorem, {so we call it the general Darling-Kac theorem.} 

\begin{theorem}[General Darling-Kac Theorem]\cite[Theorem 3.15]{limit_theorems_null}\label{Theorem Darling-Kac Theorem0}
Suppose Assumptions \ref{manifold-ass} and \ref{manifold-ass2} hold. (a) The following two statements are equivalent.
\begin{enumerate}
\item For every non-negative Borel-measurable function $g$ with $0<\langle\phi_X,g\rangle_M<\infty$, one has regular variation at $0$ of resolvants in $X_t$:
\begin{align}
(R_{1/T}g)(x):=\mathbb{E}_x\left[\int_0^\infty \exp\left( \frac{-t}{T}\right)g(X_t) dt \right] \sim  \langle\phi_X,g\rangle_M\Upsilon(T)  \label{resolvent-at-zero}
\end{align}
as $T\to \infty$ for $p_X$-almost all $x\in M$ (the exceptional set depends on $f$), 
where $\Upsilon(T):=U(T)T^\alpha$, $\alpha \in (0,1]$, and $U:\mathbb{R}^+\to\mathbb{R}^+$ is slowly varying at infinity.  
\item For any additive functional $A_t=\int_0^t f(X_s) ds$ of $X_t$ with $0<\mathbb{E}_{\phi_X}(A_1)<\infty$, one has weak convergence   
\begin{align*}
  \frac{(A_{sT})_{s\geq 0}}{\Upsilon(T)}  \xrightarrow[]{\hspace{0.1cm}d\hspace{0.1cm}} \mathbb{E}_{\phi_X}(A_1) g_\alpha(s)
\end{align*}
in $D(\mathbb{R}_+,\mathbb{R})$ as $T \rightarrow \infty$ under $\mathbb{P}_x$ for all $x\in M$, where $g_\alpha$ is the Mittag-Leffler process of index $\alpha$.
\end{enumerate}
(b) The cases in (a) are the only ones that the weak convergence of $\frac{(A_{sT})_{s\geq 0}}{\upsilon(T)}$ to a continuous nondecreasing limit process $\varpi$ so that $\varpi_0=0$ and the law of $\varpi_1$ is not degenerate at $0$, is available for some norming function $\upsilon$.
\end{theorem}

See Section \ref{section: Mittag-Leffler process} for a quick review of Mittag-Leffler process. Call $\Upsilon:\mathbb{R}^+\to \mathbb{R}^+$ the {\em scaling factor} of $X_t$. When the manifold is compact, $X_t$ is positive Harris recurrent and we have $\Upsilon(T)=T$, or $\alpha=1$, and $U(T)=1$.

\subsection{Further model assumptions}

Denote the densities of $X_t$ and $X_s$ and the corresponding joint density as $\tilde{p}_t(\cdot)$, $\tilde{p}_s(\cdot)$, and $\tilde{p}_{s,t}(\cdot,\cdot)$ respectively, and denote
\begin{equation}\label{definition of tilde g st ab}
\tilde{g}_{s,t}(a,b) = \tilde{p}_{s,t}(a,b) - \tilde{p}_s(a)\tilde{p}_t(b)\,,
\end{equation}
where $a,b\in M$,
which quantifies the level of path-dependence of the process $X_t$. We make the following assumption about $\Upsilon$'s behavior near $0$, {which allows us to apply Theorem \ref{Theorem Darling-Kac Theorem0},} and the regularity of the path-dependence of the process via the scaling factor $\Upsilon$.

\begin{assumption}
\label{reg-ass} 
The scaling factor $\Upsilon$ of $X_t$ is regularly varying at 0 with index {$\alpha \in (0,1]$}. Moreover,  
$\displaystyle\lim_{T \rightarrow \infty} \Upsilon(T)^{-1} \iint_{[0,T]^2} \|\tilde{g}_{s,t}\|_r \, ds \, dt < \infty$ for some $r\in [1,\infty]$.\end{assumption}

In other words, the first part of Assumption \eqref{reg-ass} says that Theorem \ref{Theorem Darling-Kac Theorem0}(a)(1) holds.
We impose the following assumption for $\phi_X$, which allows us to carry out asymptotic analysis on $M$. 
\begin{assumption}\label{lebesgue-dens-ass} 
The invariant measure $\phi_X$ associated with $X_t$ is absolutely continuous associated with the Riemannian measure of the manifold $M$ and admits a strictly-positive density function $p_X\in C^3(M)$ so that $\phi_X(dx) = p_X(x) dV_g(x)$ by Radon-Nikodym theorem. 
\end{assumption}

\subsection{Sampling and algorithm assumptions}
Our final set of assumptions concerns the data sampling scheme and algorithm. 
Regarding the kernel function $\mathcal{K}$, we make the following assumption.

\begin{assumption} \label{kernel-ass}
The kernel function $\mathcal{K} \in C^3(\mathbb{R})$ is nonnegative and compactly supported on $[0,L]$, where $L>0$. Denote $\kappa_{p,q}=\int_{\mathbb{R}^d}\mathcal{K}(\|u\|)^p \|u\|^q du$, where $p\in \mathbb{N}$ and $q\in \{0\}\cup\mathbb{N}$, and assume $\kappa_{1,0}=1$.
\end{assumption}

The assumption $\kappa_{1,0}=1$ can be easily achieved by a direct normalization.
Finally, we make an assumption regarding the sampling period $\Delta$ and bandwidth $h$, which depends on the scaling factor $\Upsilon$ of $X_t$.

\begin{assumption}
\label{reg-ass2} 
Assume $\Delta=\Delta(T)\to 0$, $h = h(T) \rightarrow 0$, and $\Upsilon(T)h^{2d/r}\to \infty$, where $r$ is from Assumption \ref{reg-ass}, as $T \rightarrow \infty$. 
\end{assumption}

This assumption plays a crucial role in our analysis. Letting $T\to\infty$ allows us to leverage the equilibrium behavior of the dynamics to estimate both drift and diffusion terms, while $\Delta\to0$ provides sufficiently fine temporal resolution to capture the dynamics accurately, often referred to as {\em infilling} asymptotics. 

\subsection{Asymptotic result for proposed estimators}

Recall notation in Section \ref{section: proposed algorithm}.
We start with the asymptotic behavior of $\hat{L}^{(\texttt{o})}(x)$, which is stated in the following Theorem and its proof is postponed to Section \ref{Section lemmas about occupation density}. 

\begin{theorem}[Occupation density estimator]
\label{occ-density-theorem} 
Assume Assumptions  \ref{manifold-ass}-\ref{reg-ass2} hold. If $\frac{\Delta}{h^2} =o(1)$ as $T \rightarrow \infty$, we have 
\begin{align}
   \frac{\hat{L}^{(\texttt{o})}(x)}{\Upsilon(T)} \overset{d}{\longrightarrow} g_\alpha(1)p_X(x)\ \mbox{ and }\ \mathbb{E}^M_\lambda(\hat{L}^{(\texttt{o})}(x)) = \Theta(\Upsilon(T))\,.
 \label{hat{L}-prob-order}
\end{align}
\end{theorem}

Geometrically, $\hat{L}^{(\texttt{o})}(x)$ is an estimate of the invariance measure associated with $X_t$, which can be viewed as a kernel density estimation of $p_X(x)$ when the samples are dependent and modeled by a diffusion process.

Next, we discuss the diffusion estimator. 
To describe the bias of the estimator $\hat\pi^{(\texttt{o})}(x)$, define 
\begin{align}
 b_\pi^{(\texttt{o})}(x)   :=\kappa_{1,2}\Big( \nabla\pi^{(\texttt{o})}(x) \cdot \nabla \log p_X(x) 
+ \frac{1}{2} \Delta\pi^{(\texttt{o})}(x)  \Big) \label{B_pi_varphi}\,. 
\end{align}
We have the following theorem describing the asymptotic behavior of $\hat{\pi}^{(\texttt{o})}(x)$.

\begin{theorem}[Diffusion estimator]
\label{euclidean-diffusion-estimate-total}
Suppose Assumptions  \ref{manifold-ass}-\ref{reg-ass2} hold. Further suppose that  $\frac{h^d \Upsilon(T)}{\Delta} \xrightarrow{}\infty$, $\frac{h^{d+4}\Upsilon(T)}{\Delta} \xrightarrow[]{} C$  for {a} constant $C>0$, and $ h^{d-4}\Upsilon(T)\Delta \xrightarrow[]{} 0$. When $T\to \infty$, we have
\begin{align}\label{diffusion term convergence Euclidean}
\sqrt{\frac{h^d\hat{L}^{(\texttt{o})}(x)}{\Delta}} (\hat{\pi}^{(\texttt{o})}(x)- \pi^{(\texttt{o})}(x)  - h^2 b_\pi^{(\texttt{o})}(x)) \overset{d}{\longrightarrow}N\big(  0, 2\kappa_{2,0} \pi^{(\texttt{o})}(x) \otimes {\pi}^{(\texttt{o})}(x) \big)\,,
\end{align}
where $\otimes$ is the Kronecker product.
\end{theorem}

When {$M = \mathbb{R}^d$}, our framework and results reduce to those of \cite[Theorem 5]{bandi_moloche_2018}, where the result is introduced using half-vectorization. 
Denote $\texttt{vech}$  and $\texttt{vec}$ to be the half-vectorization and vectorization operators that convert any symmetric matrix $A\in \mathbb{R}^{m\times m}$ into $\texttt{vech}(A)\in \mathbb{R}^{m(m+1)/2}$ and $\texttt{vec}(A)\in \mathbb{R}^{m^2}$.
Denote $D\in \mathbb{R}^{m^2\times \frac{m(m+1)}{2}}$ to be the standard duplication matrix that converts $\texttt{vech}(A)$ into $\texttt{vec}(A)$ for any symmetric matrix $A\in \mathbb{R}^{m\times m}$. Then, \eqref{diffusion term convergence Euclidean} can be rewritten as 
\[
\sqrt{\frac{h^d\hat{L}^{(\texttt{o})}(x)}{\Delta}} \text{vech} (\hat{\pi}^{(\texttt{o})}(x)- \pi^{(\texttt{o})}(x)  - h^2 b_\pi^{(\texttt{o})}(x)) \overset{d}{\longrightarrow}N(  0, 2\kappa_{2,0}P_D\big( {\pi}^{(\texttt{o})}(x) \otimes {\pi}^{(\texttt{o})}(x) \big) {P}_D^\top )\,,
\]
where $P_D=(D^\top D)^{-1}D^\top$. 
In the manifold setting, the deviation of $(x_{i+1}-x_i)$ from $\iota_*T_{x_i}M$ comes into play, which complicates the bias terms in \eqref{diffusion term convergence Euclidean} through curvature effects. Notably, both the bias and variance depend on the intrinsic manifold dimension $d$, rather than the ambient dimension $p$.

\begin{remark}
Note that \cite[Theorems 5]{bandi_moloche_2018} imposes the conditions $\frac{h^d \hat{L}^{(\texttt{o})}(x)}{\Delta} \xrightarrow[]{p}\infty$, $\frac{h^{d+4}\hat{L}^{(\texttt{o})}(x)}{\Delta} \xrightarrow[]{p} C$, and $ h^{d-4}\hat{L}^{(\texttt{o})}(x) \Delta\xrightarrow[]{p} 0$. We choose to replace $\hat{L}^{(\texttt{o})}(x)$ by $\Upsilon(T)$ in our assumption to avoid potential contradiction. Indeed, if we impose $\frac{h^{d+4}\hat{L}^{(\texttt{o})}(x)}{\Delta} \xrightarrow[]{p} C$, with Theorem \ref{occ-density-theorem}  that $\frac{\hat{L}^{(\texttt{o})}(x)}{\Upsilon(T)} \xrightarrow[]{d}   g_\alpha(1)p_X(x)$, Slutsky's theorem and continuous mapping theorem give $\frac{h^{d+4}\Upsilon(T)}{\Delta} \xrightarrow[]{p} \frac{C}{ g_\alpha(1)p_X(x)}$. When $\alpha\in (0,1)$, since $\frac{h^{d+4}\Upsilon(T)}{\Delta}$ is a deterministic sequence and $g_\alpha(1)$ is a nondegenerate random variable, we must have $C=0$, and hence lose the control of $\frac{h^{d+4}\hat{L}^{(\texttt{o})}(x)}{\Delta}$. As discussed in \cite[Remark 12]{bandi_moloche_2018}, it is not a problem to replace $\hat{L}^{(\texttt{o})}(x)$ by $\Upsilon(T)$ in \cite[Theorems 5]{bandi_moloche_2018} (and \cite[Theorems 4]{bandi_moloche_2018} as well), which avoids the contradiction. The same comment holds for Theorem \ref{main theorem drift}.
\end{remark}

Estimating the drift estimator is more delicate, as it requires estimating the tangent-space projection of the embedded manifold from the data. The following theorem provides a key result for tangent space estimation.

\begin{theorem}[Tangent space estimator]
\label{diff-est-tangent-space-proj}
Suppose Assumptions \ref{manifold-ass}-\ref{reg-ass2}  hold. Denote ${P}_x$ to be an orthogonal projection to {$\iota_*T_xM$}. Suppose moreover that $h^d \Upsilon(T) \xrightarrow{} \infty$, $h^{d+4} \Upsilon(T)\xrightarrow[]{} C$  for a constant $C>0$, and $\Delta^2 h^{d-4} \Upsilon(T) \xrightarrow[]{} 0$. Denote the eigenvalue decomposition 
$\hat{\pi}^{(\texttt{o})}(x) = \hat{U} \hat{D} \hat{U}^{\top}$, where the eigenvalues are ordered descreasingly.
Let $ \hat{U}_{d}$ be the $p \times d$ matrix formed from the first $d$ columns of $\hat{U}$ associated with the largest $d$ eigenvalues. Then, 
\begin{align}
\hat{P}_x :=  \hat{U}_{d}  \hat{U}_{d}^{\top} = {P}_x + h^2b^{(\texttt{t})}(x) \,,\label{btan-decomp}
\end{align}
where $b^{(\texttt{t})}(x)= \overline{b} ^{(\texttt{t})}(x) +  \epsilon^{(\texttt{t})}(x)\in \mathbb{R}^{p\times p}$, 
$\overline{b} ^{(\texttt{t})}(x) =O(1)$, and $\epsilon^{(\texttt{t})}(x)= o_p\left(\frac{1}{\sqrt{h^{d+4}\Upsilon(T)}}\right)$. 
\end{theorem}

Note that the central limit theorem for the diffusion estimator depends on the scale $\sqrt{\frac{\Delta}{h^d\hat{L}^{(\texttt{o})}(x)}}$, whereas tangent space recovery only requires the coarser scale $\frac{1}{\sqrt{h^d\hat{L}^{(\texttt{o})}(x)}}$.  This difference arises because the latter scale is already sufficient for establishing the central limit theorem for drift estimation. We do not claim optimality of the proposed diffusion-based tangent space estimator.
We also note that several methods for tangent space estimation have been proposed in the literature (e.g., \cite{singer2012vector, Tyagi2013}), primarily based on local principal component analysis (PCA). In this approach, one constructs a local covariance matrix
\[
C_x:= \sum_{k=0}^{n-1}\mathcal{K}\left(\frac{\mathcal{D}_x(x_{k})}{h} \right)(x_{k}-x)(x_{k}-x)^\top\,,
\]
and estimates $\iota_*T_xM$ using its top $d$ eigenvectors. 
However, these methods typically assume i.i.d. samples from the manifold and thus do not directly apply in our setting. Additional analysis is required to understand the behavior of local PCA under the manifold-valued diffusion model.

To state the asymptotic behavior of $\hat{\mu}^{(\texttt{o})}$, define the bias term:
\begin{align}
b_{\mu}^{(\texttt{o})}(x) =\,& \kappa_{1,2}P_x \left(\nabla \mu^{(\texttt{o})}(x)\cdot\nabla \log p_X(x)
+ \frac{1}{2} \Delta \mu^{(\texttt{o})}(x) \right)  - \overline{b} ^{(\texttt{t})}(x) \mu^{(\texttt{o})}(x) \in \mathbb{R}^p \label{mu-varphi-bias}\,,
\end{align}
where $\overline{b} ^{(\texttt{t})}(x)$ is in \eqref{btan-decomp}.

\begin{theorem}[Drift estimator]\label{main theorem drift}
Suppose Assumptions \ref{manifold-ass}-\ref{reg-ass2}  hold. Suppose moreover that $h^d \Upsilon(T) \xrightarrow{} \infty$, $h^{d+4} \Upsilon(T)\xrightarrow[]{} C$ for a constant {$C>0$}, and $\Delta^2 h^{d-4} \Upsilon(T) \xrightarrow[]{} 0$. 
Suppose ${P}_{x}$ is the orthogonal projection to {$\iota_*T_xM$}.  When $T\to \infty$,  
\begin{align}\label{drift term convergence Euclidean}
\sqrt{h^d   \hat{L}^{(\texttt{o})}(x)}  (\hat{\mu}^{(\texttt{o})}(x) - \mu^{(\texttt{o})}(x)  -
h^2 b^{(\texttt{o})}_{\mu}(x) ) 
\overset{d}{\longrightarrow} N( \mathbf{0}, \,\kappa_{2,0}P_x\pi^{(\texttt{o})}(x)P_x^\top)\,.
\end{align}
\end{theorem}

In the case where {$M =\mathbb{R}^d$}, the drift estimator reduces to that considered in \cite{bandi_moloche_2018}. Since the SDE can be constructed via an embedding procedure, it is tempting to directly apply the estimator from \cite{bandi_moloche_2018}. However, this naive approach introduces an additional bias of order $1$ in the normal direction due to curvature.
More specifically, if we omit $\hat{P}_x$ and simply use $\hat{\mu}_E(x)$ \eqref{Formula: generic drift term estimator observed}, we have
\begin{align*} 
\sqrt{{h^d\hat{L}^{(\texttt{o})}(x)}}  (\hat{\mu}_E(x) - (\mu^{(\texttt{o})}(x)+\mu^{(\texttt{o})}_\eta(x)) -
h^2b_{\mu}^{(\texttt{o})}(x) ) 
\overset{d}{\longrightarrow} N( 0, \,\kappa_{2,0}\pi^{(\texttt{o})}(x))\,,
\end{align*}
which differs from \eqref{drift term convergence Euclidean} by a nontrivial term 
\begin{align*}
\mu^{(\texttt{o})}_{\eta}(x) &=  \frac{1}{2}\left(\begin{matrix} \sum_{a=1}^d \sum_{b=1}^d  \langle \sigma^{(\texttt{o})}_{a}(y), \sigma^{(\texttt{o})}_{b}(y)\rangle e_1^\top I\!\!I_y(\partial_a, \partial_b) \\
\vdots \\
 \sum_{a=1}^d \sum_{b=1}^d  \langle \sigma^{(\texttt{o})}_{a}(y), \sigma^{(\texttt{o})}_{b}(y)\rangle e_p^\top I\!\!I_y(\partial_a, \partial_b) \end{matrix}\right)\in\mathbb{R}^p
\end{align*}
living in the normal space at $x$. The deviation $\mu^{(\texttt{o})}_\eta(x)$ cannot be eliminated even if we can access the geodesic distance of the manifold and set $\mathcal{D}_x(x_{k})=d_g(x,x_k)$ in $\hat{\mu}_E(x)$. See Figure \ref{fig:IB1} in the numerical section for examples.

We emphasize a key technical distinction between the analyses of the diffusion and drift estimators. For the drift, four error terms must be controlled in the expansion of the exponential map. This arises because the drift of a diffusion process acquires second-order contributions under a coordinate transformation via geometry-stochastics relationship, which is made clear via It\^o's lemma. As a result, curvature, the second-order structure of the manifold, induces a non-negligible distortion in the drift that persists asymptotically. Note that even in the Euclidean setting, the bias contains a second-order term originating from the drift itself \cite[Theorem 4]{bandi_moloche_2018}.

We  compare the bandwidth-sampling size pair $(h,n)$ for the drift and diffusion estimators with the choices commonly used in the manifold learning literature. In manifold learning, observations are typically assumed to be independent. For pointwise convergence, the bandwidth $h=h(n)$ is usually chosen such that $nh^d\to \infty$ and $h\to 0$ as $n\to \infty$. Geometrically, if the sampling density is bounded away from $0$, the condition $nh^d\to \infty$ ensures that asymptotically there are infinitely many samples in a shrinking ball of radius $h$, providing sufficient local data for estimation. While the relationship is more complicated, a parallel condition for the drift estimator is $h^d\Upsilon(T)\to \infty$ as $T\to \infty$. When $X_t$ is positive Harris recurrent, for example, $M$ is compact, $\Upsilon(T)=T$. Using the sampling relation $T=n\Delta$ with $\Delta\to 0$ as $T\to \infty$, the assumption can be rewritten as $h^dn\Delta\to \infty$ as $n\to \infty$. The interpretation is similar but not identical to the independent-sampling case. Since $\Delta\to 0$, the condition implies $h^dn\to\infty$, so that asymptotically infinitely many observations fall in a local neighborhood. The additional factor $\Delta$ slows the effective rate and requires more observations when the sampling interval is small, reflecting the stronger temporal dependence between closely spaced samples. Under the same positive Harris recurrent condition, the condition $\frac{h^d\Upsilon(T)}{\Delta}\to \infty$ when $T\to \infty$ for the diffusion estimator is equivalent to $h^dn \to \infty$ as $n\to \infty$. Thus, when $T$ and $\Delta$ are fixed across procedures, and hence the same sample size, $\hat{\pi}^{(\texttt{o})}$ converges faster than $\hat{\mu}^{(\texttt{o})}$ by a factor of $\sqrt{\Delta}$. This theoretical finding has an intuitive interpretation. From \eqref{Intro section ito-form}, 
\begin{align*}
&\iota(X_{(k+1)\Delta})-\iota(X_{k\Delta})\\
=\,&\left(\iota_*\nu+\frac{1}{2}D_{\iota_*\sigma_\alpha}(\iota_*\sigma_\alpha)\right)(\iota(X_{k\Delta}))\Delta+\iota_*\sigma_\alpha(X_{k\Delta}) \sqrt{\Delta}Z^\alpha+o_p(\Delta^{3/2})\,, 
\end{align*}
where $Z^\alpha\sim N(0,1)$. In this approximation, the drift term is of order $\Delta$, while the diffusion term, or noise, is of order $\sqrt{\Delta}$. When $\Delta$ gets smaller, the signal-to-noise ratio for the drift decreases at rate $\sqrt{\Delta}$; the stochastic fluctuations dominate the deterministic drift contribution in each increment, making its recovery more difficult. In contrast, 
\begin{align*}
&(\iota(X_{(k+1)\Delta})-\iota(X_{k\Delta}))(\iota(X_{(k+1)\Delta})-\iota(X_{k\Delta}))^\top\\
=\,&(\iota_*\sigma_\alpha(X_{k\Delta}) Z^\alpha)(\iota_*\sigma_\alpha(X_{k\Delta}) Z^\alpha)^\top\Delta+O_p(\Delta^{3/2})\,, 
\end{align*}
so the diffusion coefficient appears in the leading term of the increment. Consequently, the diffusion component is statistically easier to estimate than the drift.

\section{Simulated Experiments}\label{section simulation label}

In this section, we examine our proposed estimators using two 2-dimensional manifolds, the standard $2$-sphere $\mathbb{S}^2= \{x \in \mathbb{R}^3: \|x\| = 1\}\subset\mathbb{R}^3$ and the Klein bottle embedded in $\mathbb{R}^4$. 
Throughout this section, we use a smooth, compactly supported kernel function defined by $\mathcal{K}(s) =\exp(-(1-(s/3)^2)^{-1})$ when $s \in [0,3)$ and $0$ otherwise. Bandwidth selection is delicate in practice, particularly when the sampling density is nonuniform. Since developing a bandwidth selection algorithm is beyond the scope of this work, we follow the empirical practice in manifold learning and select $h$ so that the kernel is supported on {a neighborhood that is of size $\sim 1\%$ of the total trajectory length.} The problem of identifying an optimal bandwidth will be investigated in future studies (see Discussion). The Python code producing figures and results are available in \url{https://github.com/jacobmcerlean/Functional-Estimation-Manifold-SDEs}.

\subsection{$2$-sphere}  
Consider $M$ as a 2-dim ellipsoid with eccentricity $(a,b,c)$ embedded in $\mathbb{R}^3$. To simulate SDE trajectories on $M$, consider $\mathbb{S}^2\subset \mathbb{R}^3$ and diffeomorphism from $\mathbb{S}^2$ to $M$ via $\varphi:(x,y,z) \mapsto (ax,by,cz)$.  
Consider the SDE with the drift $\mu^{(\texttt{l})}(x,y,z) = (y,-x,0)$ and Riemannian Brownian motion on $\mathbb{S}^2$. 
We simulate the process using a \emph{retraction-based Euler scheme} \cite{schwarz2025efficient}, which is detailed below.
Given a number of observations $n\in \mathbb{N}$, time-step size $\Delta$, and an initial condition $Y_0 = y_0 \in \mathbb{S}^2$, we generate a discrete trajectory
$\{Y_k\}_{k=0}^n$ on $\mathbb{S}^2$ as follows. For $k =0,\ldots,n-1$, we first sample a random unit vector
$w_k \in \mathbb{R}^3$ and an independent random $\chi^2(2)$-distributed radius $r_k$. Then, we set a random tangent vector 
$v_k := r_k\left(\frac{w_k - w_k^\top Y_kY_k}{|w_k - w_k^\top Y_kY_k|}\right)\in T_{Y_k}\mathbb{S}^2$.
This produces an isotropic tangent increment consistent with Riemannian Brownian motion on $\mathbb{S}^2$. The {\em Euler increment} is then defined by
\begin{align}
   \delta Y_k = \sqrt{\Delta}\, v_k + \Delta\, \mu(Y_k)\,. \label{Euler scheme} 
\end{align}
The next state of the SDE trajectory is obtained via a retraction map given by radial projection; that is,
$Y_{k+1} := \frac{Y_k + \delta Y_k}{\|Y_k + \delta Y_k\|}$. Then iterate with $k$.
As is shown in \cite{schwarz2025efficient}, this iteration approximates the intrinsic SDE on $\mathbb{S}^2$. The drift of $X_t=\varphi(Y_k)$ is 
\[
\mu^{\texttt{(o)}}(x,y,z)
=\left(\frac{ay}{b},\,\frac{-bx}{a},\,0\right)+P_{(x,y,z)}\!\left(-(x,y,z)\right),
\]
where
\[
P_{(x,y,z)}w \;=\; w-\frac{\langle w,n(x,y,z)\rangle}{\langle n(x,y,z),n(x,y,z)\rangle}\,n(x,y,z),
\quad
n(x,y,z)=\left(\frac{x}{a^{2}},\,\frac{y}{b^{2}},\,\frac{z}{c^{2}}\right).
\]
and diffusion is $\pi^{\texttt{(o)}}(x,y,z) = D \varphi D \varphi ^\top- (x,y,z)(x,y,z)^\top$, where $D\varphi \in \mathbb{R}^{3 \times 3}$ is the matrix $ \text{diag}(a,b,c)$.

Now we report results on an ellipsoid $M$ with eccentricity $(a,b,c)=(2,1.5,1)$, normalized by applying a global scaling $\sqrt{\frac{{3}}{a^2 + b^2 + c^2}}$, and postpone results for other ellipsoids with eccentricities $(1,1,1)$ and $(1.5,1,1)$ with the same normalization to Section \ref{section more numerical simulations}.

We start with demonstrating the rate of convergence of the empirical density of an SDE trajectory to the invariant measure of $X_t$. Simulate a long trajectory of length $n_{\max} = 10^8$ and $\Delta = 10^{-2}$, for physical time $T = 10^6$, on $\mathbb{S}^2$, and map this trajectory to $M$. In Figure \ref{fig:ellipsoid_density2}, we report the log-log plot of $\|\hat{L}^{(\texttt{o})}-\phi_X\|_{L^2}$ to demonstrate the rate of convergence of the empirical density. The decay rate of around $n^{-1/2}$ agrees with the positive Harris recurrence nature of the dynamics. For a visualization, we plot the empirical density $\hat{L}^{(\texttt{o})}(x)$ of the trajectory at different lengths $n_i$, where $\log_{10}(n_i) \in \{4, 5,6,7,8\}$, in Figure \ref{fig:ellipsoid_density}.

\begin{figure}[hbt!]
    \centering
\includegraphics[width=.6\textwidth]{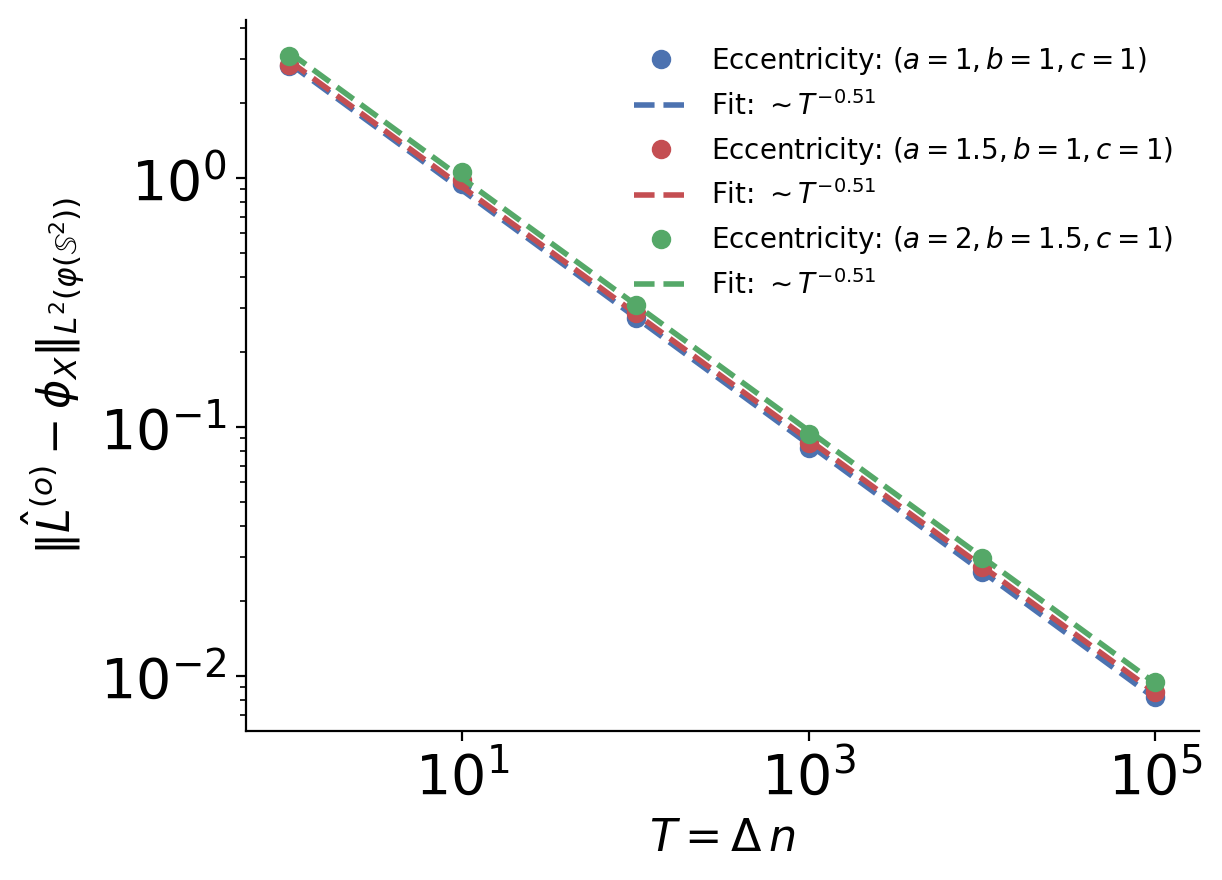}
    \caption{Using the occupation density for trajectory length $n = 10^8$, time-step $\Delta=10^{-2}$, and physical time $T = 10^6$ as an accurate estimate for the invariant density $\phi_X$, we compare the invariant density to the estimate $\hat{L}^{(\texttt{o})}$ based on the first $n$ data-points of the trajectory on each observed ellipsoid to measure the convergence rate of the empirical density.}
    \label{fig:ellipsoid_density2}
\end{figure}

Next, we simulate an SDE trajectory of length $n=10^8$, with the initial point sampled uniformly from $N$, time step $\Delta = 10^{-4}$, and physical time $T = 10^4$. The dataset is obtained by downsampling the trajectory, keeping every 100$^{th}$ observation. This preserves the physical time while reducing the sample size to $n = 10^6$ with uniform time step $\Delta = 10^{-2}$. 
 
Figure~\ref{fig:IB1} illustrates $\hat{\mu}_E(x)$ and $\hat{\mu}^{(\texttt{o})}(x)$ computed from the dataset on $M$ is $(2,1.5,1)$, shown on a spherical cap near $(2,0,0)$. Additional examples with different eccentricities are provided in Figures \ref{fig:IB1-11} and \ref{fig:IB1-12}. 
The curvature-induced bias of $\hat{\mu}_E(x)$ is clearly visible in Figure~\ref{fig:IB1}. For comparison, we also plot $P\hat{\mu}_E$, which is defined as $(P\hat{\mu}_E)(x):= P_x\hat{\mu}_E(x)$, where $P_x$ is the true projection operator onto $\iota_*T_xM$. As expected, when the tangent space is known, $P\hat{\mu}_E$ visually performs better than $\hat{\mu}^{(\texttt{o})}$.

To quantify these observations, we sample $10^4$ points uniformly on $\mathbb{S}^2$, map them to $M$, and evaluate drift and diffusion estimation errors at these base-points. Due to the topological constraints, any continuous drift on $\mathbb{S}^2$ must vanish somewhere. To avoid numerical instability, we stratify points by 
$\|\mu^{(\texttt{o})}(x)\|/\|\mu^{(\texttt{o})}\|_\infty$ with threshold $c=0.05$. For the subset of points with $\{\|\mu^{(\texttt{o})}(x)\|/\|\mu^{(\texttt{o})}\|_\infty>0.05\}$ ($99.9\%$ out of $10^4$), we report: 
(i) relative vector error, or normalized root mean square error (NRMSE), $\frac{\|\hat\mu-\mu^{(\texttt{o})}\|}{\|\mu^{(\texttt{o})}\|}$, where $\hat\mu$ can be any drift estimator, (ii) relative norm error $\frac{|\|\hat\mu\|-\|\mu^{(\texttt{o})}\||}{\|\mu^{(\texttt{o})}\|}$, and (iii) angle error $\Theta(\hat\mu,\mu^{(\texttt{o})})$ with the unit radians. 
For the remaining points, where the drift is near zero, we report absolute drift errors $\|\hat\mu-\mu^{(\texttt{o})}\|$ only.
For diffusion, we report the relative Frobenius error, or NRMSE, $\frac{\|\hat\pi-\pi\|_F}{\|\pi\|_F}$ and a subspace metric based on principal angles between the leading two-dimensional eigenspaces, $U_2$ (resp.\ $\widehat U_2$), of $\pi$ (resp.\ $\hat\pi$), $\|\sin\Theta_2\|_F$, where $\sin\Theta_2:= \sin\Theta(U_2,\widehat U_2)$, with $\Theta(U_2,\widehat U_2)$ the diagonal matrix so that the singular values of $U_2^\top \widehat U_2$ are the cosines of the diagonal matrix $\Theta(U_2,\widehat U_2)$. Note that $\|\sin\Theta_2\|_F$ measures the error of tangent space estimation.

Summary statistics (means $\pm$ standard deviations) are reported in Table~\ref{tab:mag-and-error-plots-obs-ell-table}, with the associated histograms in Figures~\ref{fig:obs-ell-hist} and \ref{fig:obs-ell-hist-2}. Results for $M$ with other eccentricities are shown in Table~\ref{tab:mag-and-error-plots-obs-ell-table-2}. Consistently smaller diffusion errors compared to drift errors reflect the greater difficulty of drift estimation discussed after Theorem~\ref{main theorem drift}.
The drift estimators exhibit the predicted ordering: $\hat\mu_E$ has the largest error, $\hat\mu^{(\texttt{o})}$ improves upon it, and $P\hat\mu_E$ performs best. This pattern holds across relative vector, norm, and angle errors. Paired one-sided Wilcoxon signed-rank tests with Bonferroni correction confirm $P\hat\mu_E < \hat\mu^{(\texttt{o})} < \hat\mu_E$ for all drift metrics (adjusted $p$-values $<10^{-5}$).

Finally, a detailed study at $(0,0,1)^\top$, based on 1000 independent SDE simulations with $n = 10^6$ and $\Delta = 10^{-2}$ can be found in Figures \ref{fig:IA3a} and \ref{fig:IA3b}. The approximately Gaussian error distributions align with the asymptotic normality result; QQ-plots are provided in Figures~\ref{fig:IA2} and \ref{fig:IA2-22}.

\begin{figure}[hbt!]
    \centering
    \includegraphics[trim=0 0 0 0,clip, width=0.95\textwidth]{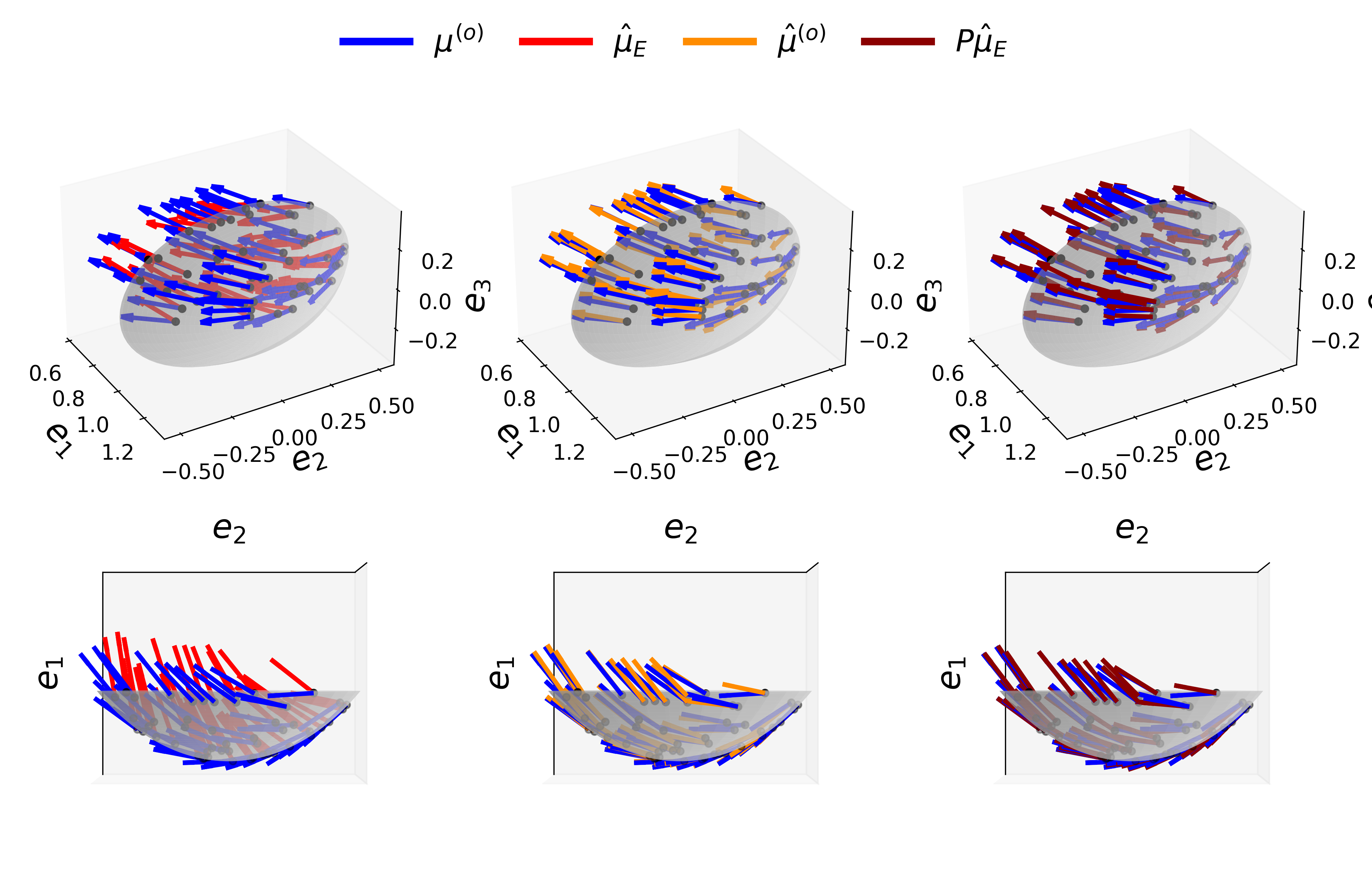}
    \caption{From left to right: visualizations of $\hat{\mu}_E(x), \hat{\mu}^{(\texttt{o})}(x),$ and $P_x \hat{\mu}_E(x)$, where $P_x$ is the projection operator onto the tangent space $T_xM$, for base-point samples $x$ drawn uniformly from a spherical cap centered at $(1,0,0)^\top$ and observed on an ellipsoid with eccentricity $(2,1.5,1)$, shown from two viewing angles. The ground-truth drift vector is superimposed as blue arrows.}
    \label{fig:IB1}
\end{figure}

\begin{table}[t]
\scriptsize
\renewcommand{\arraystretch}{1.15}
\begin{tabular}{lcccc}
\toprule
Error type & $\hat{\mu}_E$ & $\hat{\mu}^{(\texttt{o})}$ & $P_x\hat{\mu}_E$ & $\hat{\pi}^{(\texttt{o})}$ \\
\midrule
 $\frac{\|\hat\mu(x)-\mu^{(\texttt{o})}(x)\|}{\|\mu^{(\texttt{o})}(x)\|}$ ($\frac{\|\mu^{(\texttt{o})}(x)\|}{\|\mu^{(\texttt{o})}\|_\infty}\ge 0.05$)      & $1.060 \pm 0.566$ & $0.215 \pm 0.215$ & $0.208 \pm 0.214$ & --- \\
$\frac{|\|\hat\mu(x)\|-\|\mu^{(\texttt{o})}(x)\||}{\|\mu^{(\texttt{o})}(x)\|}$ ($\frac{\|\mu^{(\texttt{o})}(x)\|}{\|\mu^{(\texttt{o})}\|_\infty}\ge 0.05$)& $0.474 \pm 0.507$& $0.150 \pm 0.169$ & $0.147 \pm 0.167$ & --- \\
  $\Theta(\hat\mu(x),\mu^{(\texttt{o})}(x))$ ($\frac{\|\mu^{(\texttt{o})}(x)\|}{\|\mu^{(\texttt{o})}\|_\infty}\ge 0.05$)            & $0.782 \pm 0.190$& $0.129 \pm 0.192$ & $0.124 \pm 0.193$ & --- \\
 $\|\hat\mu(x)-\mu^{(\texttt{o})}(x)\|$   ($\frac{\|\mu^{(\texttt{o})}(x)\|}{\|\mu^{(\texttt{o})}\|_\infty}< 0.05$)                 & $0.644 \pm 0.013$ & $0.126 \pm 0.060$ & $0.125 \pm 0.059$ & --- \\
 $\|\hat\pi(x)-\pi^{(\texttt{o})}(x)\|_F/\|\pi(x)\|_F$          & --- & --- & --- & $0.038 \pm 0.013$ \\
  $\|\sin\Theta_2(x)\|_F$         & --- & --- & --- & $0.017 \pm 0.008$ \\
\bottomrule
\end{tabular}

\caption{\raggedright
Summary of various evaluation metrics. $\hat{\mu}(x)$ is the estimator of $\mu^{(\texttt{o})}(x)$, which can be $\hat{\mu}_E$, $\hat{\mu}^{(\texttt{o})}$, or $P_x\hat{\mu}_E$, where $P_x$ is the projection to $T_xM$, listed in the top. $\hat\pi(x)$ is the estimator of $\pi^{(\texttt{o})}(x)$, which is $\hat{\pi}^{(\texttt{o})}(x)$. $\Theta(\hat\mu,\mu^{(\texttt{o})})$ is the angle between $\mu^{(\texttt{o})}$ and $\hat\mu$ with the unit radian.  $\|\sin\Theta\|_F$ is the subspace distance between the dominant 2D eigenspaces of $\hat\pi^{(\texttt{o})}$ and $\pi^{(\texttt{o})}$.
}
\label{tab:mag-and-error-plots-obs-ell-table}
\end{table}

\begin{figure}[hbt!]
    \centering
\includegraphics[width=0.85\textwidth]{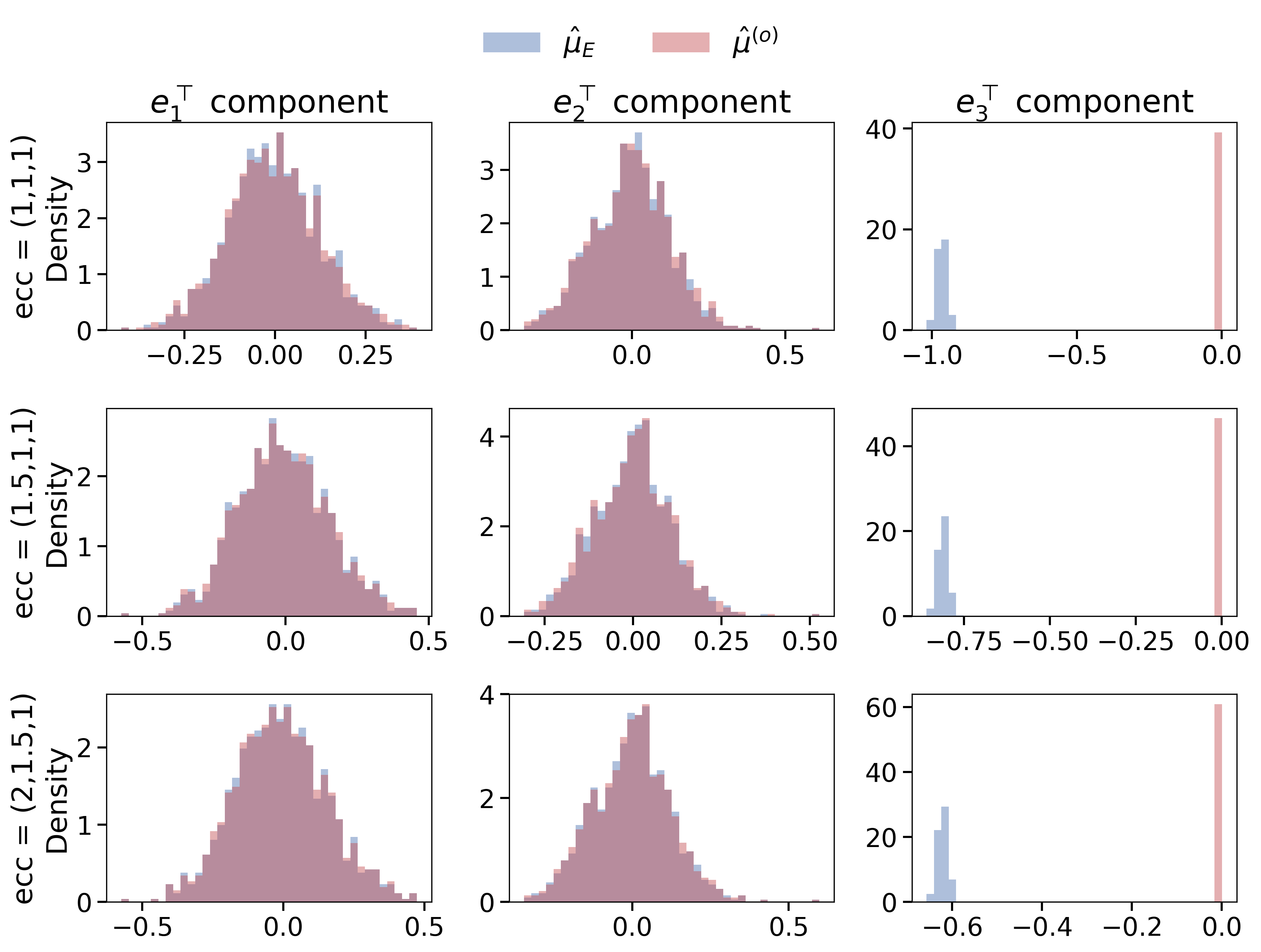}
   \caption{Histograms of drift estimation errors over 1000 independent SDE simulations, comparing $\hat{\mu}_E$ and $\hat{\mu}^{(\texttt{o})}$ to the true drift vector field $\mu^{(\texttt{o})}$ at $(0,0,1)^\top$.}
    \label{fig:IA3a}
\end{figure}

\begin{figure}[hbt!]
    \centering
    \includegraphics[width=0.85\textwidth]{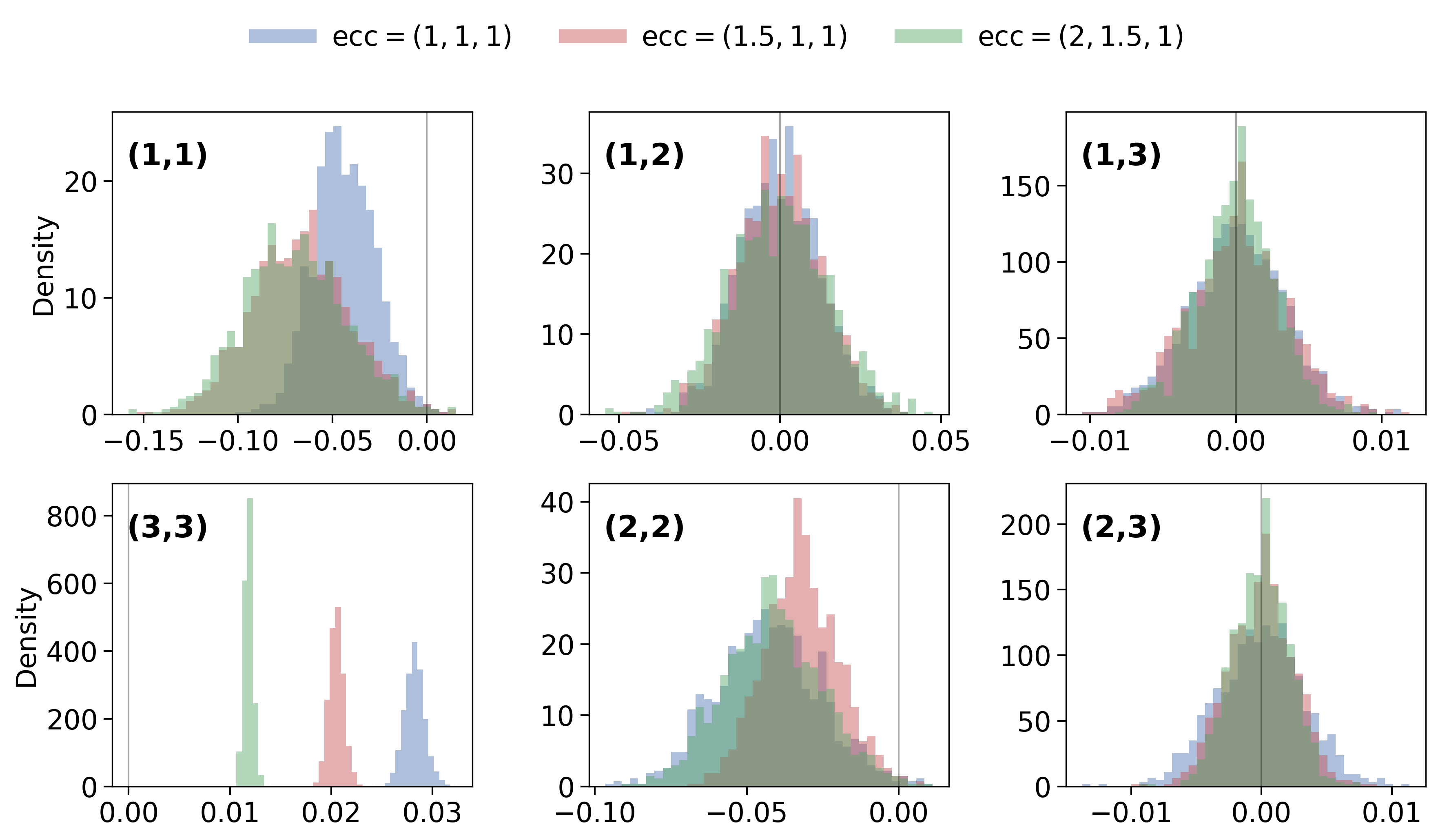}
     \caption{Histograms of the errors of the estimated diffusion matrix $u_i^\top\hat{\pi}^{(\texttt{o})}u_j$ over 1000 independent SDE simulations (labeled as $(i,j)$ in the subplots) compared to the true vector fields $u_i^\top \pi^{(\texttt{o})} u_j$, for ellipsoids of varying eccentricities.}
    \label{fig:IA3b}
\end{figure}

\subsection{The Klein Bottle in $\mathbb{R}^4$}

Set $N:=[0,2\pi)\times [0,2\pi)$. Consider the manifold $M$ to be the Klein bottle $K\subset \mathbb{R}^4$, a compact, non-orientable, $2$-dimensional smooth manifold without boundary, embedded in $\mathbb{R}^4$ via the smooth map
$\varphi:(u,v)\in N\to
    [\cos(u)(a + r\sin v), \ 
        \sin(u)(a + r\sin v), \ 
        r\cos(u/2)\sin v, \ 
        r\sin(u/2)\sin v]^\top$, where $a > r > 0$. 
Let $\Gamma$ be the fundamental group of the Klein bottle, which is a discrete group of diffeomorphisms of $\mathbb{R}^2$ generated by 
\begin{align}
g_1:(u,v) \mapsto (u, v + 2 \pi)\,, \quad
g_2:(u,v) \mapsto (u + 2 \pi, -v), \label{KB-identification}
\end{align}
where $(u,v)\in \mathbb{R}^2$, fulfilling the relation $g_1g_2g_1^{-1}=g_2$.
Note that $g_1$ leads to a cylinder, while $g_2$ twist and glue the cylinder's ends to generate a Klein bottle. Also note that $N$ is the fundamental region for the action of $\Gamma$ on $\mathbb{R}^2$.
Let $\pi:\mathbb{R}^2\to \mathbb{R}^2/\Gamma$ be the quotient map.

By construction, $\varphi$ is $\Gamma$-compatible; that is, $\varphi(u,v) = \varphi(\gamma \cdot (u,v))$ for all $\gamma \in \Gamma$, with $\gamma \cdot$ denoting the group action.
A smooth vector field $\mu$ on $\mathbb{R}^2$ is called {\em $\Gamma$-invariant} if for any $\gamma \in \Gamma$, $\gamma_*\mu(u,v)=\mu(\gamma\cdot(u,v))$. A $\Gamma$-invariant vector field $\mu$ therefore induces a smooth vector field on $N$,  and a smooth vector field on $M$ lifts to a $\Gamma$-invariant vector field on $\mathbb{R}^2$. 
Thus, smooth vector fields on $M$ are in one-to-one correspondence with $\Gamma$-invariant smooth vector fields on $\mathbb{R}^2$. With this property, in this section we simulate the $M$-valued SDE via simulating the Euclidean SDE on $\mathbb{R}^2$.

For the SDE, we set the drift as $\mu(u,v) = (1 + \frac{1}{2}\cos\left(\frac{u}{2} \right)\sin(v), \frac{1}{2}\sin\left(2v\right))^\top$ (shown in Figure~\ref{fig:KB_drift}) on $N$ and the diffusion to be Riemannian Brownian motion to enforce the $\Gamma$-invariant condition, so that the SDE is valid on $N$ with the drift $\mu^{(\texttt{l})}:=\pi_*\mu$, the drift $\mu^{(\texttt{o})}:=\varphi_*\mu^{(\texttt{l})} + \frac{1}{2} \iota_*(\nabla_{\varphi_* \sigma_\alpha} \varphi_* \sigma_\alpha)$ for $\sigma_\alpha$ vector fields generating Brownian motion on $N$, and the diffusion $\pi^{(\texttt{o})}= D\varphi D\varphi^\top$.
We use a standard Euler-scheme as in \eqref{Euler scheme} on the Euclidean plane for generating the SDE, and then apply $\varphi$ to obtain the SDE on $M$ embedded in $\mathbb{R}^4$.
Below, set {$a=2$ and $r=1$} in the parameterization of $M$ by $\varphi$.

Simulate a long trajectory with $n=10^8$ points and time-step $\Delta = 10^{-2}$, and then view the estimated density as a surrogate of $\phi_X$. In Figure \ref{fig:IC1}, we report the log-log plot of the $\|\hat{L}^{(\texttt{o})}-\phi_X\|_{L^2}$ to demonstrate the rate of convergence of the empirical density. The decay rate at around $n^{-0.5}$ is consistent with that of ellipsoid and close to $n^{-0.5}$, as expected. For a visualization, we plot the empirical density $\hat{L}^{(\texttt{o})}(x)$ of the trajectory at different lengths $n_i$, where $\log_{10}(n_i) \in \{4, 5,6,7,8\}$, in Figure \ref{fig:IC1-2}. 

\begin{figure}[hbt!]
    \centering
\includegraphics[width=0.65\textwidth]{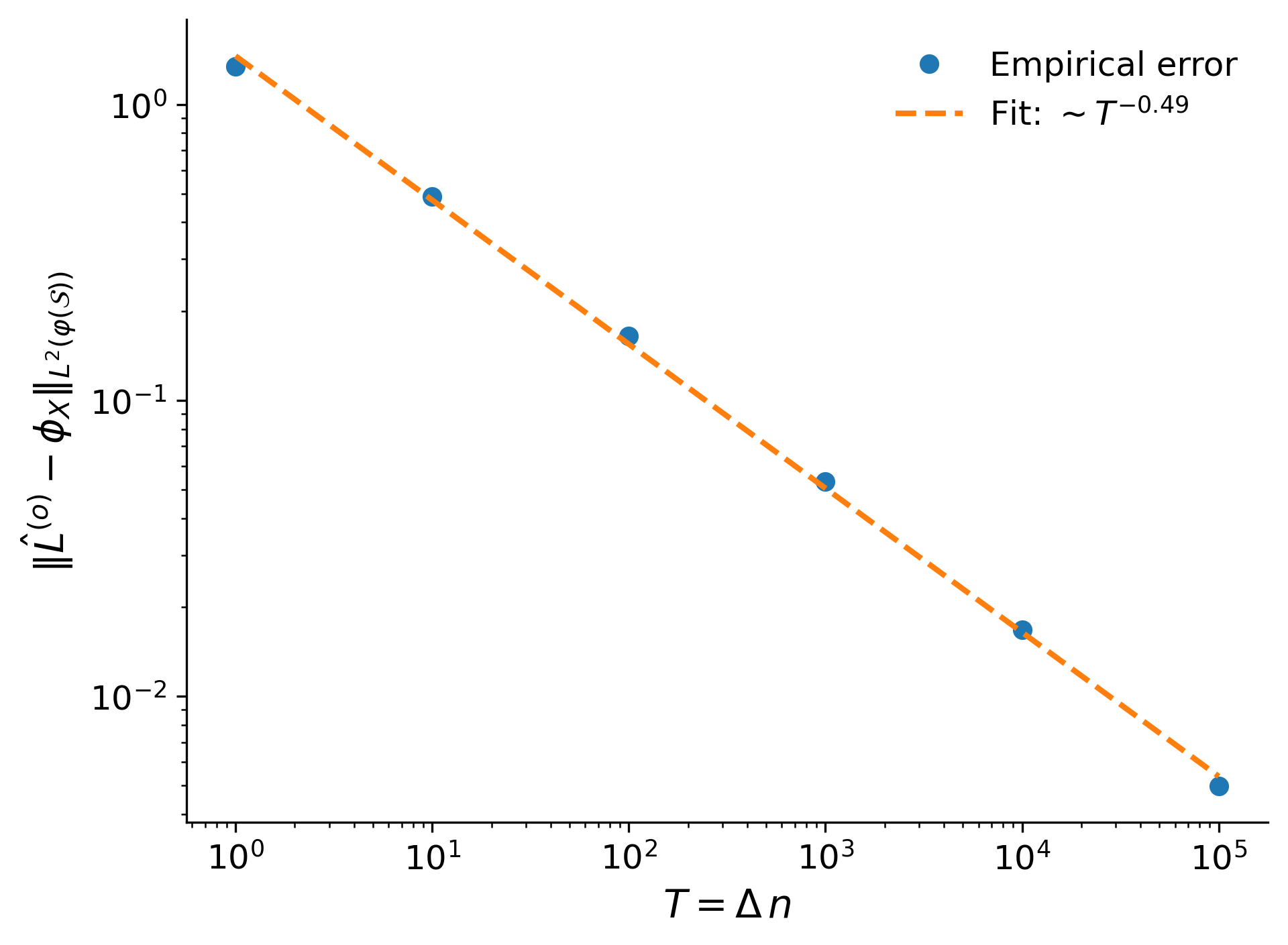}
     \caption{Using the occupation density for trajectory length $n = 10^8$ and $\Delta = 10^{-2}$ as an accurate estimate for the invariant density $\phi_X$, we compare the invariant density to the estimate $\hat{L}^{(\texttt{o})}$ based on the first $n$ data-points of the trajectory to measure the convergence rate of the empirical density.}
    \label{fig:IC1}
\end{figure}

Next, we realize an SDE trajectory with $n=10^8$ with initial point chosen uniformly at random from the sphere. We set the time-step $\Delta = 10^{-4}$, so that the physical time of the trajectory is $T = 10^4$. The trajectory is then down-sampled by every $100$ observations, which preserves the physical time but reduces sampled points to $n = 10^6$ with uniform time-step $\Delta = 10^{-2}$. Thus, the resolution of the trajectory generation is higher than that used in the estimation procedure. The down-sampled single trajectory is then mapped by $\varphi$ to the Klein bottle in $\mathbb{R}^4$. To visualize and compare the drift estimators, we reduce the dimension of a local patch using singular value decomposition (SVD). Take a local patch of $B=465$ trajectory points that lie within $[\pi-1, \pi+1]\times[2 \pi-\frac{1}{2}, 2\pi/2 + \frac{1}{2}] \subset N$, and are mapped to $\mathbb{R}^4$ by $\varphi$. Denote the center $(c_u,c_v)=(\pi,2\pi)\in N$, with image $x_c=\Phi(c_u,c_v)$. Form a centered data matrix $X=\big[x_1-x_c,\ldots,x_{B}-x_c\big]\in\mathbb{R}^{4\times B}$ and compute its SVD as $X=USV^\top$ with singular values ordered decreasingly. Define $U_3=[u_1,u_2,u_3]\in\mathbb{R}^{4\times 3}$ from the top three left singular vectors, and construct 
$z_i =U_3^\top(x_i-x_c)\in\mathbb{R}^3$, $i=1,\ldots,B$.
Then, plot the projected vectors
$U_3^\top \mu^{(\texttt{o})}(x_i)$, $U_3^\top {\hat{\mu}}_E(x_i)$, $U_3^\top {\hat{\mu}}^{(\texttt{o})}(x_i)$, and $U_3^\top {(P_x\hat{\mu}_E)}(x_i)$ on $\{z_i\}$. To aid visualization, we rotate $z_i$ so that the least weighted direction of $z_i$ is on the $z$ axis and plot the rotated patch together with the rotated drift vectors from three viewing angles in Figure \ref{fig:IC4}. The curvature-induced bias in $\hat{\mu}_E$ is visually clear. To aid visualization and avoid crowding, we randomly select only $58$, or approximately $B/8$, of the base-points in $[\pi-1, \pi+1]\times[2 \pi-\frac{1}{2}, 2\pi/2 + \frac{1}{2}]$ to plot the drift estimators.

\begin{figure}[hbt!]
  \centering
    \includegraphics[width=.95\textwidth]{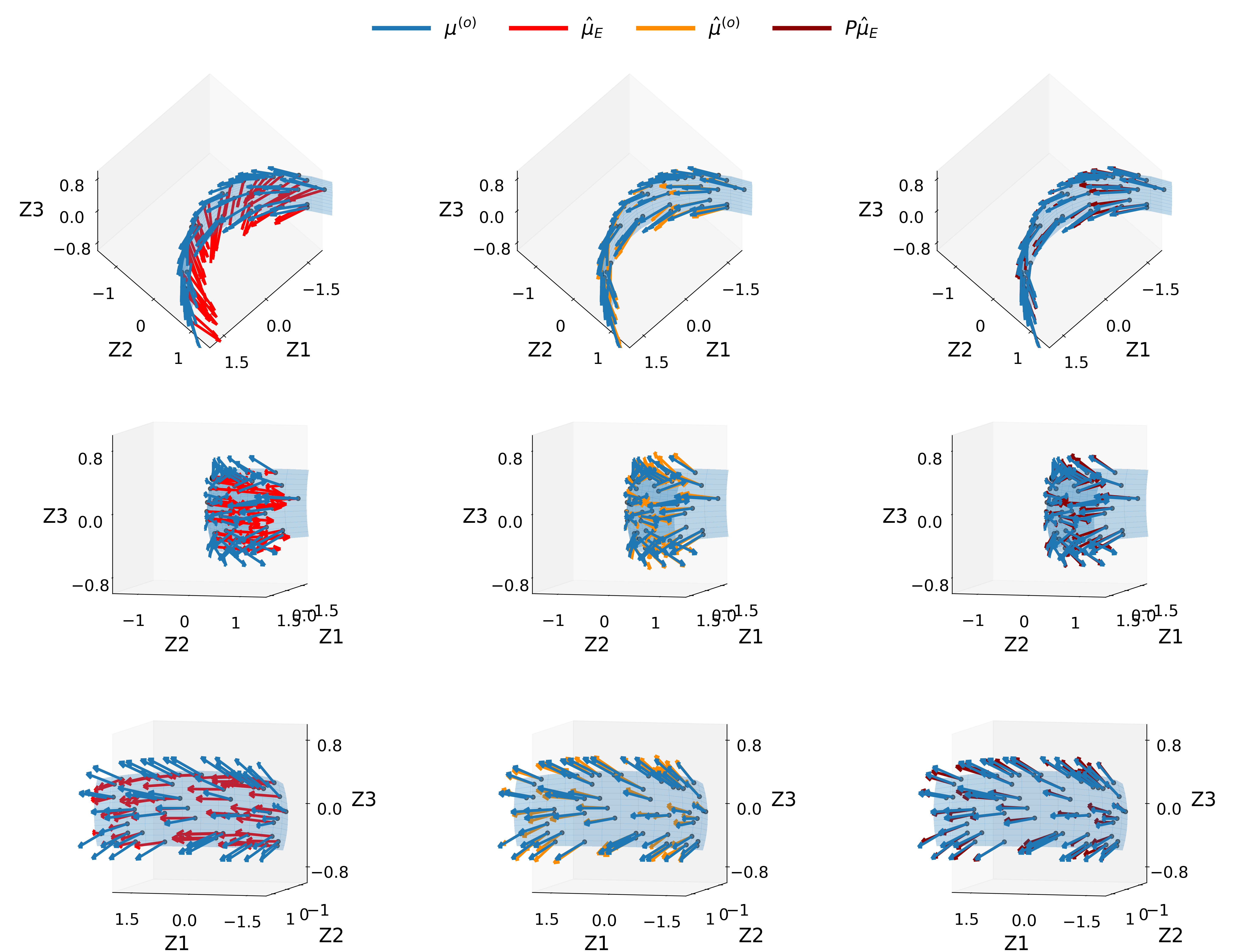} 
   \caption{From left to right: visualizations of $\hat{\mu}_E(x), \hat{\mu}^{(\texttt{o})}(x),$ and $P_x \hat{\mu}_E(x)$, where $P_x$ is the projection operator onto the tangent space $T_xM$. The ground-truth drift vector is superimposed as blue arrows. The associated local patch mapped via $U_3$ and rotation is overlaid to enhance the visualization.}
  \label{fig:IC4}
\end{figure}

To quantify these visual findings, we follow the same principle as that in the ellipsoid. We sample {$10^4$} points on $M$ by taking a uniform $100 \times 100$ grid on $N$ and mapping by $\varphi$. Since the vector field on $N$ is non-vanishing, we evaluate NRMSE at all points. The results are shown in Table~\ref{tab:obskb_summary_with_fro_and_subspace}. The associated histograms can be found in Figure \ref{fig:IC2-2}. For NRMSE and angle error of drift estimates, $\hat\mu_E$ has substantially larger error than $\hat\mu^{(\texttt{o})}$, and $\hat\mu^{(\texttt{o})}$ is slightly worse than $P_x\hat\mu_E$ (adjusted $p$-values $<10^{-5}$). For the {relative norm error}, $\hat\mu^{(\texttt{o})}$ has significantly smaller relative norm error than $\hat\mu_E$ (adjusted $p$-values $<10^{-5}$) while
the comparison between $P_x\hat\mu_E$ and $\hat\mu^{(\texttt{o})}$ does not have statistical significance.

\begin{table}[t]
\centering
\renewcommand{\arraystretch}{1.1}

\begin{tabular}{lcccc}
\toprule
Error type & $\hat{\mu}_E$ & $\hat{\mu}^{(\texttt{o})}$ & $P_x\hat{\mu}_E$ & $\hat{\pi}^{(\texttt{o})}$ \\
\midrule
$\frac{\|\hat\mu(x)-\mu^{(\texttt{o})}(x)\|}{\|\mu^{(\texttt{o})}(x)\|}$              & $0.363 \pm 0.206$ & $0.144 \pm 0.083$ & $0.137 \pm 0.086$  & --- \\
$\frac{|\|\hat\mu(x)\|-\|\mu^{(\texttt{o})}(x)\||}{\|\mu^{(\texttt{o})}(x)\|}$        & $0.117 \pm 0.090$& $0.100 \pm 0.079$& $0.102 \pm 0.080$   & --- \\
$\Theta(\hat\mu(x),\mu^{(\texttt{o})}(x))$                  & $0.319 \pm 0.202$ & $0.089 \pm 0.065$& $0.074 \pm 0.065$& --- \\
$\frac{\|\hat\pi(x)-\pi^{(\texttt{o})}(x)\|_F}{\|\pi(x)\|_F}$                      & ---              & ---                & ---                & $0.048 \pm 0.016$ \\
$\|\sin\Theta_2(x)\|_F$                               & ---              & ---                & ---                & $0.052 \pm 0.041$ \\
\bottomrule
\end{tabular}
\caption{\raggedright Summary of various evaluation metrics. $\hat{\mu}(x)$ is the estimator of $\mu^{(\texttt{o})}(x)$, which can be $\hat{\mu}_E$, $\hat{\mu}^{(\texttt{o})}$, or $P_x\hat{\mu}_E$, where $P_x$ is the projection to $T_xM$, listed in the top. $\hat\pi(x)$ is the estimator of $\pi^{(\texttt{o})}(x)$, which is $\hat{\pi}^{(\texttt{o})}(x)$. $\Theta(\mu^{(\texttt{o})},\hat\mu)$ is the angle between $\mu^{(\texttt{o})}$ and $\hat\mu$ with the unit radian.  $\|\sin\Theta\|_F$ is the subspace distance between the dominant 2D eigenspaces of $\hat\pi^{(\texttt{o})}$ and $\pi^{(\texttt{o})}$.}
\label{tab:obskb_summary_with_fro_and_subspace}
\end{table}

\section{Discussion}\label{section discussion label}
This work develops drift and diffusion estimators, along with tangent space and occupation density estimators, from uniform and high-frequency sampling under a manifold-based time-homogeneous SDE framework with theoretical guarantees. Several important issues merit further discussion.
A central practical question is bandwidth selection, which remains delicate in diffusion-based regression. Achieving an optimal bandwidth choice of $h$ is complicated, particularly when $X_t$ is null Harris recurrent; see \cite[Remark 12]{bandi_moloche_2018}. Existing methods \cite{ait2016bandwidth,bandi2009bandwidth} provide guidance for Euclidean settings, but their extension to manifold-supported, temporally dependent data is largely unexplored. Identifying data-driven bandwidth choices that balance curvature effects, sampling density, and temporal dependence is a promising direction that is also important for manifold learning algorithms. 
The present analysis assumes a fixed sampling interval $\Delta$. In many biomedical applications \cite{wang2023arterial,Chang2025.10.12.25337819}, however, $\Delta$ is nonuniform and dependent on $X_t$. A thorough study of nonuniform sampling, particularly for high-frequency data corrupted by noise or synchronization errors, is beyond the scope of this paper but represents a natural extension. Related challenges include model misspecification \cite{pavliotis2007parameter}, low-frequency sampling \cite{gobet2004nonparametric}, high-frequency noise amplification \cite{zhang2005tale,pavliotis2007parameter}, and asynchrony or randomized sampling schemes \cite{akahori2023symmetric}. All of these issues become more intricate on manifolds.
Another practical question is the convergence rate of the proposed estimators. Weak convergence follows from technical limit theorems \cite{limit_theorems_null} and we need different techniques to derive the rate.
More generally, biomedical time series are often time-inhomogeneous and contaminated by nonstationary noise or measurement artifacts. Extending the framework to accommodate such behavior, possibly through generalized Fourier-Malliavin techniques \cite{malliavin2009fourier,akahori2023symmetric}, is an important and largely open direction. Locally linear manifold regression methods \cite{cheng2013local}, or higher-order variations, may further reduce bias in curved regions, though controlling curvature-dependent remainder terms in the higher-order variations requires additional technical development. From a manifold-learning perspective, an important direction is to adapt existing algorithms to dependent sampling schemes; for instance, can we still approximate the Laplace-Beltrami operator from the graph Laplacian? If not, how to modify it if we need the Laplace-Beltrami operator? 
We leave these extensions to future work.

\ifsubmission

\else
\bibliographystyle{plain}  
\fi

\bibliography{references}

\ifsubmission

\else

\appendix

\setcounter{page}{1}
\setcounter{equation}{0}
\setcounter{definition}{0}
\setcounter{lemma}{0}
\setcounter{theorem}{0}
\renewcommand{\thepage}{S.\arabic{page}}
\renewcommand{\thesection}{S.\arabic{section}}
\renewcommand{\theequation}{S.\arabic{equation}}
\renewcommand{\thelemma}{S.\arabic{section}.\arabic{lemma}}
\renewcommand{\thetable}{S.\arabic{table}}
\renewcommand{\thedefinition}{S.\arabic{definition}}
\renewcommand{\thefigure}{S.\arabic{figure}}

\section{Necessary mathematical background}\label{sec:App1}
In this section, we summarize necessary mathematical background for the analysis of the proposed estimators. 

\subsection{Life-cycle decomposition and L\"ocherbach-Loukianova embedding}\label{section LL embedding}
Recall that an
 {\em additive functional} \cite[Definition 1.5]{limit_theorems_null} of a Harris recurrent $X=(X_t)_{t\geq 0}$ on $(\Omega,\mathcal{F},\mathbb{P})$ taking values in a manifold $M$ with c\`adl\`ag paths, and with $X_0=x\in M$ a.s., is a process $(A_t)_{t\geq 0}\subset\mathbb{R}_+\cup\{\infty\}$ so that $A_t$ is $\mathcal{F}$ adapted, $A_0=0$, all paths are nondecreasing and right-continuous, and for every $x\in M$ and for all $s,t\geq 0$, we have $A_{t+s}=A_t+A_s\circ \vartheta_t$ $P_x$-a.s., where $\vartheta_t$ is the shift operator. 
 In our asymptotic analysis of the proposed estimators, we encounter additive functionals of the form
\begin{align}
A_t:=\int_0^t f(X_s) ds\,, \label{additive-functional}
\end{align}
where $t \geq 0$ and $f\geq 0$ is bounded and measurable, and its discretization. We focus on studying the asymptotic behavior of $A_t$ and its Riemann sum discretization when $t\to \infty$ and $X_t$ is an It\^o diffusion and Harris recurrent with an invariant measure $\phi_X$.  For such $f$, the additive functional $A_t$ defines a measure on $(M,\mathcal{B}(M))$ via $\nu_A(B):=\mathbb{E}_{\phi_X}^M(\int^1_0 \boldsymbol{1}_B(X_s)dA_s)$, where $B\in \mathcal{B}(M)$ is a Borel measurable set. We call the additive functional {\em integrable} if $\nu_A(M)=\mathbb{E}_{\phi_X}^M(A_1)<\infty$ \cite[page 12]{limit_theorems_null}. 
Note that 
\begin{align}\label{section LL embedding additive functional measure is linear}
\mathbb{E}_{\phi_X}^M\left(\int^1_0 \boldsymbol{1}_B(X_s)dA_s\right)=t^{-1}\mathbb{E}_{\phi_X}^M\left(\int^t_0 \boldsymbol{1}_B(X_s)dA_s\right)
\end{align} 
for any $t>0$ since $t\mapsto \mathbb{E}_{\phi_X}^M(A_t)$ is linear. So, alternatively when $\langle \phi_X, f \rangle_M <\infty$, the additive functional is integrable \cite{LOCHERBACH20092312}, and we have $\mathbb{E}_{\phi_X}^M(A_1)= \langle \phi_X, f \rangle_M$. 

\subsubsection{Recurrent atom} \label{subsection recurrent atom}
The analysis of \eqref{additive-functional} would become straightforward if this additive functional of $X_t$ could be decomposed into i.i.d. (or even stationary and strongly mixing) components. Such a decomposition would exist if the invariant measure $\phi_X$ of the process $X_t$ contained a {\em recurrent atom} $\mathsf A$; i.e., whenever the process enters 
$\mathsf A$, it returns to $\mathsf A$ infinitely often with probability $1$.

\begin{definition}[Definition 1.9.A, \cite{limit_theorems_null}]
Consider a random process $X_t$ on the Polish state space $(E, \mathcal{E})$ adapted to the filtration $\mathcal{F}$. The set $\mathsf A$ is an \textit{atom} for $X_t$ if
    \begin{enumerate}
 \item $\sigma_{\mathsf A} := \inf \{ t > 0: X_t \in \mathsf A\}$ and $\tau_{\mathsf A} = \inf \{ \tau > 0: X_t \notin \mathsf A\}$ are $\mathcal{F}$-stopping times;
 
 \item for $x \in \mathsf A$, the distribution of $X_{\tau_{\mathsf A}}$ conditional on $X_0 = x$ does not depend on $x$.
    \end{enumerate}
An atom is called \textit{recurrent} if for all $x \in \mathsf A$, then a.s. for every $N$ there is some $t > N$ with $X_t \in \mathsf A$ given $X_0 = x$.
\end{definition}
The existence of a recurrent atom $\mathsf A$ implies a sequence of stopping times $\{R_m\}_{m=0}^\infty$, called a {\em life-cycle decomposition} \cite[Definition 1.9.B]{limit_theorems_null}, so that $X_{R_m} \in \mathsf A$, and for all $m$, we have
\begin{enumerate}
\item the stopping time $R_m< \infty$ with $R_{m+1} = R_m + R_1 \circ \vartheta_{R_m}$, where $\vartheta$ is the shift operator;
\item $X_{t + R_m}$ is independent of $\mathcal{F}_{R_m}$ for $t \ge 0$.
\end{enumerate}
Supposing such a life-cycle decomposition exists, we have standard tools to understand the asymptotic behavior of the additive functional of the process $X_t$.

However,  Harris recurrence does not in general guarantee life-cycle decomposition, even in the Euclidean setting, since the process $X_t$ may fail to admit a recurrent atom. Nevertheless, the desired result can still be obtained by invoking the embedding technique introduced in \cite{LOCHERBACH20081301}, which we call {\em L\"ocherbach-Loukianova (LL) embedding}. LL embedding extends the classical constructions of of \cite{nummelin_1978,athreya_1978} to continuous-time processes. The key idea is to embed $X_t$ into an augmented process $Z_t$ that possesses a recurrent atom $\mathsf A$, thereby providing a surrogate notion of life-cycle decomposition suitable for our analysis. This technique, central to \cite{bandi_moloche_2018}, will likewise prove applicable in our setting. For self-containedness, and to facilitate its extension to the manifold-valued setting, we outline the construction below.

\subsubsection{L\"ocherbach-Loukianova embedding}

We summarize the LL embedding, $Z_t := Z(X_t)$ of the process $X_t$ defined to solve a manifold-valued SDE, as follows (full details can be found in Section 2.2 of \cite{LOCHERBACH20081301}).

Suppose $X_t$ satisfies Assumptions \ref{manifold-ass}, \ref{manifold-ass2}, and \ref{lebesgue-dens-ass}. As discussed in Section \ref{subsection review SDE}, it is strong Markov \cite[IX.3-5]{Elworthy1982sde} with c\`adl\`ag paths \cite[VIII.6]{Elworthy1982sde}. When $M$ is compact, we have all necessary conditions. Since $X_t$ is positive Harris recurrent, we have an invariant probability measure $\phi_X$ with a smooth density function $p_X > 0$ by Assumption \ref{lebesgue-dens-ass}.  
By \cite[Theorem 6.4.1]{Stroock2010pde}, there exists smooth kernel functions $p_t(x,y)$ on $(0,\infty)\times M\times M$ such that $p_t(x,y)>0$ for all $t>0$ and $x,y\in M$, and $p_t(x,y)$ has a Gaussian upper and lower bounds. Therefore, Assumption 2.1 in \cite{LOCHERBACH20081301} is satisfied. Moreover, \cite[condition (2.8)]{LOCHERBACH20081301} about the resolvent control, $\sup_{x,x' \in M} \int_0^\infty te^{-t}\frac{p_t(x,x')}{u^1(x,x')} dt < \infty$, where $u^1(x,x') := \int_0^\infty e^{-t}p_t(x,x') dt$, holds by a direct calculation. Clearly, $d\nu_{x,x'}(t):=\frac{e^{-t}p_t(x,x')}{u^1(x,x')}dt$ is a probability measure on $\mathbb{R}_+$. With this validation, Assumption 3.1 in \cite{LOCHERBACH20081301} is also satisfied. When $M$ is noncompact, we use the imposed Assumption \ref{manifold-ass2} so that Assumption 3.1 in \cite{LOCHERBACH20081301} is satisfied and the LL embedding can be applied.

We first assemble the ingredients required for the LL embedding.
By Harris recurrence, the process returns  infinitely often to any set with positive $\phi_X$ measure almost surely. Choose a compact set $C \subset M$ so that $\phi_X(C) > 0$, where $\phi_X$ is the invariant measure of $X_t$; the particular choice of $C$ is immaterial to the result. Let $\xi$ be the probability measure equivalent to $\phi_X( \cdot \cap C)$. Recall that the transition kernel of $X_t$ is  given by $P_t(x,dy)=p_t(x,y)dV_g(y)$, where $dV_g(y)$ is the Riemannian volume measure. Recall that the discrete-time random process $\{\bar{X}_n:=X_{\mathsf T_n}\}_{n\in\mathbb{N}}$, where $\mathsf T_n:=\sum_{j=1}^n\mathsf W_j$ and $\mathsf W_j$ is an i.i.d. $\exp(1)$-waiting times, is Harris recurrent with its transition kernel $U^1(x,dy)$ satisfying the minorisation condition: \cite[(2.3)]{LOCHERBACH20081301}
\begin{align*}
U^1(x,dy) &:= \int_0^\infty e^{-t}P_t(x,dy)dt \ge \beta \mathbf{1}_C(x) \xi(dy) \,,
\end{align*}
where $\beta \in (0,1)$. By Assumption, $U^1\ll dV_g$ and the associated density function if $u^1(x,u):=\int_0^\infty e^{-t}p_t(x,y)dt$.

Let $U(du)$ denote the uniform probability measure on $[0,1]$. Following \cite[(2.4)]{LOCHERBACH20081301}, define the transition kernel $Q((x,u),dy)$ from $M \times [0,1]$ to $M$ by
\begin{align*}
Q((x,u),dy) &= \begin{cases}
\xi(dy) \qquad & (x,u) \in C \times [0, \beta] \\
(1- \beta)^{-1}(U^1(x,dy) - \beta\xi(dy))\qquad &(x,u) \in C \times (\beta,1 ] \\
U^1(x,dy) \qquad &x \notin C \,.
\end{cases}
\end{align*}
By construction, $\int_0^1 Q((x,u),dy)du=U^1(x,dy)$ \cite[(2.5)]{LOCHERBACH20081301}. 

We now quickly review the LL embedding of $X_t$ by constructing a new random process $Z_t= (Z^1_t, Z^2_t, Z^3_t)$ that takes values in $\mathcal{M} := M \times [0,1] \times M$  using $Q$. Set $\mathsf T_0=0$. Initially, set $Z^1_0=X_0=x$,  $Z^2_0\sim U$, where $U$ is the uniform distribution on $[0,1]$, and set $Z^3_0\sim Q((x,u),dx')$ conditional on $Z^2_0=u$. The construction is iteratively on $m\in \mathbb{N}$, and it depends on a sequence of independent jump times $\{\mathsf T_m\}_{m=1}^\infty$ iteratively defined as $\mathsf T_{m+1}:= \mathsf T_m+\sigma_{m+1}$, where $\sigma_m$ follows the law $e^{-t}\frac{p_t(x,x')}{u^1(x,x')}$ on $\mathbb{R}_+$ with $x=Z^1_{\mathsf T_m}$ and $z'=Z^3_{\mathsf T_m}$, and $\{\mathsf T_m\}_{m=0}^\infty$ are independent of the process $X_t$.
Denote $I_m := [\mathsf T_m, \mathsf T_{m+1})$. Then, on $I_m$, construct $Z^1$ that bridges $X_{\mathsf T_m}$ and $X_{\mathsf T_{m+1}}$, and set $Z_t^2=Z^2_{\mathsf T_m}$ and $Z_t^3=Z^3_{\mathsf T_m}$ for all $t\in I_{m}$. At the jump time $\mathsf T_{m+1}$, set $Z^1_{\mathsf T_{m+1}}:=Z^3_{\mathsf T_m}$, choose $Z^2_{\mathsf T_{m+1}}\sim U[0,1]$ independent of $Z_s$, $s<\mathsf T_{m+1}$, and conditional on $Z^1_{\mathsf T_{m+1}}=x'$ and $Z^2_{\mathsf T_{m+1}}=u'$, choose $Z^3_{\mathsf T_{m+1}}\sim Q((x',u'),dx'')$. Clearly, the evolution of $Z_t$ during $I_m$ does not depend on $Z^2_{\mathsf T_m}$. Then, iterate. Here, $Z^2_t$ and $Z^3_t$ are piecewise constant, $Z^2_t$ is used to realize the splitting and $Z^3_t$ records the future information. 
We refer readers with interest in details to \cite[Section 2]{LOCHERBACH20081301}.
Denote $\mathsf{F}_t$ be the filtration generated by $Z_t$ and $\mathsf{G}_t$ be the filtration generated by $(Z^1_t,Z^2_t)$.

Some properties follow regarding this construction. While we do not need all of them, we collect them here for the self-containedness. 

\begin{proposition}\label{proposition Zt is Harris recurrent}
The LL embedded process $Z_t$ satisfies:
\begin{enumerate}
\item \cite[Theorem 2.12]{LOCHERBACH20081301} $Z_t$ is strong Markovian for any stopping time $S$ such that $Z^1_S\neq Z^3_S$ a.s. and $Z_t$ is c\`adl\`ag. 

\item $Z_t$ is Harris recurrent with the invariant probability measure 
\[
\Pi(dy,du,dz)=dV_g(y)U(du)\int_M \phi_X(dx) u^1(x,y)\frac{u^1(y,z)}{u^1(x,z)}Q((x,u),dz)\,. 
\]

\item By construction, while $\{\mathsf T_m\}_m$ is not independent of $Z_t$, it is independent of $Z^1_t$. Also, $\{\mathsf T_m\}_m$ is the jump times of $(Z^2_t, Z^3_t)$. 
 
\item \cite[Proposition 2.6]{LOCHERBACH20081301} $(\mathsf T_n- \mathsf T_{n-1})_{n\geq 1}$ are i.i.d. $\exp(1)$ random variables, and $\mathsf T_{n+1}-\mathsf T_n$ is independent of $\mathsf{F}_{\mathsf T_n-}$. 
\end{enumerate}
\end{proposition}

The next property is the most important one concerning us is the existence of a ``recurrent atom'' on the extended space $\mathcal{M}$.

\begin{proposition}\cite[Proposition 2.8(a)]{LOCHERBACH20081301}\cite[Proposition 4.1]{LOCHERBACH20092312}\label{LL-Prop2.8a} 
The set \begin{align}\label{definition Zt recurrent atom}
\mathsf A = C \times [0, \beta] \times M\,,
\end{align}
where $C$ compact with $\phi_X(C) > 0$ and $\beta \in (0,1)$, is a ``recurrent atom'' for $Z_t$ with respect to the measure $\bar{\lambda}$ in \eqref{lambda-extension-measure} in the sense that, for $R:= \inf\{n|\,Z_{\mathsf T_n} \in A\}$, then $[Z_{\mathsf  T_{R+1}}| Z^1_{\mathsf T_R}, Z^2_{\mathsf T_R}]$ has density $\bar{\lambda}$. The law of $(X_{\mathsf T_n})_{n\geq 0}$ conditional on $X_0=x\in M$ is the same as $(Z^1_{\mathsf T_n})_{n\geq 0}$ conditional on $Z^1_0=x\in M$, and the law of $(X_t)_{t\geq 0}$ conditional on $X_0=x\in M$ is the same as $(Z^1_t)_{t\geq 0}$ conditional on $Z^1_0=x\in M$. Moreover, the law of $[Z^3_t|Z^1_t]$ follows $u^1(Z^1_t,x')dx'$.
\end{proposition}

This proposition says that $X_t$ and $Z^1_t$ have the same law, which justified the term ``embedding''. Moreover, in this richer random process, we possess a recurrent atom $\mathsf A$. Note that the third coordinate $M$ of $\mathsf A$ is trivial as the process $Z_t^3$ is always contained in $M$.

Before introducing the ``generalized life-cycle decomposition, we shall make clear the origin and relation of several measures relevant to the embedding construction, which is critical for the upcoming analyses. Let $\lambda$ be a generic, initial measure on $M$ modeling the distribution of $X_0$. 
The associated expectation $\mathbb{E}^M_\lambda$ (or $\mathbb{E}_\lambda$ when there is no confusion) is defined as 
\begin{align*}
\mathbb{E}^M_\lambda\left( \int_0^T f(X_t) dt\right) &:=\int_{X_0=x\in M} \int_0^T f(X_t) \, dt\, d\lambda(x)\,, 
\end{align*}
where $f$ is a real-valued, measurable function on $M$.

\begin{remark}
Although not strictly necessary, it is useful to record an alternative representation of $\mathbb{E}^M_\lambda\left( \int_0^T f(X_t) dt\right)$. 
Let ${\mathcal{W}}^r_0:=\left\{\omega \in C([0, \infty), \mathbb{R}^r) : \omega(0) = 0\right\}$, equipped with its natural filtration $\{\mathcal{F}_t\}_{t > 0}$ and Wiener measure $\mathbb{P}_W$. 
Let ${\mathcal{W}}(M) = C([0, \infty), M)$, endowed with the filtration $\mathcal{B}_t({\mathcal{W}}(M))$ generated by cylinder sets, and denote by $\mathbb{P}_W^M$ the Wiener measure on ${\mathcal{W}}(M)$. 
Equip $M \times \mathcal{W}^r_0$ with the product measure $\mathbb{P}^M_\lambda:= \lambda \otimes \mathbb{P}_W$. Define the path-map $F^M: M \times \mathcal{W}^r_0 \rightarrow {\mathcal{W}}(M)$ by $F^M(x_0,\omega)(t)=X_t(x_0,\omega)$, where $X_t(x_0, \omega)$ denotes the location at time $t$ of the trajectory $X_t$ starting from $x_0\in M$ and driven by the Brownian motion determined by $\omega$. By construction, $F^M$ is measurable. Then, we induce a measure $P^M_\lambda$ on ${\mathcal{W}}(M)$ by
$P^M_\lambda := \mathbb{P}^M_\lambda \circ (F^M)^{-1}$. With this notation, $\mathbb{E}^M_\lambda\left( \int_0^T f(X_t) dt\right):=\int_{\gamma \in  \mathcal{W}({M})} \int_0^T f(\gamma(t)) \, dt\, d{P}^{M}_\lambda(\gamma)$, which provides a path-space formulation.
\end{remark}
From the initial measure $\lambda$ for $X_t$, we consider an initial measure 
\begin{align}
    \overline{\lambda}(dx,du,dy) &:= \lambda(dx) U(du) Q((x,u),dy) \label{lambda-extension-measure}
\end{align} 
for $Z_t$ on $\mathcal{M}:= M \times [0,1] \times M$, where the variables $dx$, $du$, and $dy$ indicate a placeholder for the upcoming integration. 
Since $\mathcal{M}= M \times [0,1] \times M$ is the space in which $Z_t$ evolves, we denote the path-space ${\mathcal{W}}(\mathcal{M})$ to be the space of paths induced by ${\mathcal{W}}(M)$ and the LL embedding construction, with filtration $\mathcal{B}_t({\mathcal{W}}(\mathcal{M}))$ generated by cylinder sets, and $\mathbb{P}_W^\mathcal{M}$ the Wiener measure on ${\mathcal{W}}(\mathcal{M})$. Note that paths in ${\mathcal{W}}(\mathcal{M})$ are not continuous by construction, and $\mathcal{B}_t({\mathcal{W}}(\mathcal{M}))$ is the same as $\mathsf{F}_t$. 
Similarly, for $g$ a real-valued, measurable functions on $\mathcal{M}$, we define
    \begin{align}
    \mathbb{E}^\mathcal{M}_\lambda\left( \int_0^T g(Z_t) dt\right) &:= \int_{Z_0 =z\in \mathcal{M}} \int_0^T g(Z_t) \, dt\, d\bar{\lambda}(z) \label{extended-manifold-path-space-expectation} \,,
\end{align} 
where $\bar{\lambda}$ is the initial measure of $Z_t$ related to the generic initial measure $\lambda$ of $X_t$ via \eqref{lambda-extension-measure}.

In our analysis, we study the {\em deterministic equivalent} of $\int_0^T f(X_t) dt$, defined as $\mathbb{E}^M_\lambda\left( \int_0^T f(X_t) dt\right)$, which is unique up to a positive constant being asymptotically of the same order. 
Note that since $Z^1_t$ follows the same law as $X_t$ by the LL embedding construction, for a function $g$ defined on $\mathcal{M}$ that only depends on the first coordinate; that is, $f=f_0\circ P_1$, where $P_1$ is the projection to the first coordinate, we have 
\begin{equation} \label{expectations-agree}
\mathbb{E}^\mathcal{M}_\lambda\left( \int_0^T f(Z_t) dt\right)=\mathbb{E}^\mathcal{M}_\lambda\left( \int_0^T f_0(Z^1_t) dt\right)=\mathbb{E}^M_\lambda\left( \int_0^T f_0(X_t) dt\right)\,.
\end{equation}

\subsubsection{Generalized life-cycle decomposition and generalized regeneration times}

With the LL embedding and the recurrent atom $A$ of $Z_t$ \eqref{definition Zt recurrent atom}, we introduce the desired {\em generalized life-cycle decomposition}.

\begin{definition}[Generalized life-cycle decomposition] \label{gen-lif-cycle-decomp}
Let $\mathsf A$ be the recurrent atom of $Z_t$ and $\{ \mathsf T_m\}_{m=0}^\infty$ the associated sequence of exponential jump-times used to construct $Z_t$.
We construct a sequence of {\em $\mathsf{F}_t$-stopping times} $\{S_m, R_m\}_{m=0}^\infty$, called the  a {\em generalized life-cycle decomposition}, by the following. Initially, set $R_0 = S_0 = 0$, and subsequent times are iteratively set by $S_{m+1} := \inf\{m'|\,\mathsf T_{m'} > R_m \,, Z_{\mathsf T_m'} \in \mathsf A\}$ and $R_{m}  = \inf\{m'|\, \mathsf T_{m'}\,, \mathsf T_{m'} > S_{m}\}$ for all $m\in \mathbb{N}$.
\end{definition}

By this construction, we summarize properties of $S_m$ and $R_m$ in the following proposition.

\begin{proposition}\label{prop: summary of Rn and Sn}
With the construction in Definition \ref{gen-lif-cycle-decomp}, we have:
\begin{enumerate}
\item \cite[Proposition 2.13(a)]{LOCHERBACH20081301} For all $m$, $Z_{R_m+}$ is independent of $\mathsf{G}_{S_n}$ and $\mathsf{F}_{S_n-}$; that is, at $R_n$, we start fresh and have independence after a waiting time. Hence, $Z_{R_m+}$ is also independent of $\mathsf{F}_{R_{n-1}}$. 

\item \cite[Proposition 4.2 (c)]{LOCHERBACH20092312} The law of $Z_{R_n}^1$ conditional on $\mathsf{G}_{S_n}$ is $\xi$. 

\item \cite[Proposition 2.13(b)]{LOCHERBACH20081301} $\mathbb{E}^{\mathcal{M}}_\lambda(R_n-S_n|\mathcal{F}_{S_n-})\leq 1/\beta$ for all $n\in \mathbb{N}$. 

\item \cite[Proposition 4.2 (b)]{LOCHERBACH20092312} We have 
\begin{align}\label{section LL ZRn distribution}
Z_{R_n}\sim \xi(dx)U(du)Q((x,u),dx')\ \ n\geq 1\,.
\end{align}

\item \cite[Proposition 4.4]{LOCHERBACH20092312} Denote 
\[
U_m:=\int_{R_m}^{R_{m+1}} f(X_s) \, ds
\]
for $m\geq0$. The sequence $(U_n)_{n\geq 0}$ is a stationary ergodic sequence under $\mathbb{P}_\xi$, and the sequence $(U_{2k})_{k\geq 0}$ is i.i.d, as well as the sequence $(U_{2k+1})_{k\geq 0}$.
\end{enumerate}
\end{proposition}

The second point of 
Proposition \ref{prop: summary of Rn and Sn} says that 
the sequence of $\mathsf{F}$-stopping time $\{R_m\}_{m = 0}^\infty$ is a {\em  life-cycle decomposition} for the process $Z_t$, and the sequence $Z_{R_n}$, $n\geq 1$, is i.i.d.. 
The fourth point of Proposition \ref{prop: summary of Rn and Sn} says that at regeneration times, we start from a fixed distribution that does not depend on the past. 
The following proposition is needed in our analysis.

\begin{proposition}\label{2.20and4.3}
\cite[Proposition 2.20]{LOCHERBACH20081301}\cite[Proposition 4.3]{LOCHERBACH20092312}
For a Harris recurrent $X_t$ with invariant measure $\phi_X$ and $\langle \phi_X, |f| \rangle_M < \infty$, denote $A_t:=\int_0^t f(X_s) ds$. For any initial measure $\lambda$ and $n \ge 1$, we have
\begin{align*}
\mathbb{E}^{\mathcal{M}}_\lambda(A_{R_{n+1}}-A_{R_n}) =   
\mathbb{E}^{\mathcal M}_\xi(A_{R_1}) = C_X\langle \phi_X, f \rangle_M \,,
\end{align*}
where $C_X\neq 0$ does not depend on $f$.  
\end{proposition}

Note that the additive functional over the first cycle, $A_{R_1}$, depends on the initial measure. When the initial measure is $\xi$, it is the same as the future cycles $A_{R_{n+1}}-A_{R_n}$, when $n>1$.   This is because of \eqref{section LL ZRn distribution} and $Z_t$ is strong Markov.  If the initial measure is not $\xi$, then $A_{R_1}$ is different from the future cycles, and we need assumptions to control it. For example, in \cite[(2.8) in Proposition 2.16]{LOCHERBACH20081301}, if 
\[
\sup_{x,x'\in M} \int_0^\infty te^{-t}\frac{p_t(x,x')}{u^1(x,x')}dt<\infty\,, 
\]
where $u^1(x,x')=\int_0^\infty e^{-t}p_t(x,x')dt$, then $\mathbb{E}^{\mathcal M}_{\lambda}(A_{R_1}) $ and $\mathbb{E}^{\mathcal M}_{\lambda}(A_{S_1})$ are both bounded. This assumption is used to control $\mathbb{R}^{\mathcal M}_\lambda(\mathsf{T}_1)$. It holds when $M$ is compact since the diffusion kernel has a nice Gaussian control. However, when $M=\mathbb{R}$ and $X_t$ is Brownian motion, which is null Harris recurrent ($\mu=0$ and $\sigma=1$), $p_t(x,x')=\frac{1}{\sqrt{4\pi t}}\exp(-d(x,x')^2/4t)$ and $\int_0^\infty te^{-t}\frac{p_t(x,x')}{u^1(x,x')}dt\asymp d(x,x')$ when $d(x,x')\to \infty$. The assumption fails. Combined with the motivation from biomedical signal analysis, we introduce the assumption that $X_0$ follows a distribution supported on a compact set $C\subset M$ with a bounded density function.  

Recall that due to \eqref{section LL embedding additive functional measure is linear}, we have 
\begin{align*}
\mathbb{E}_{\phi_X}^M(A_1)=\langle \phi_X, f \rangle_M\,, 
\end{align*}
so from time to time, when $f$ is positive measurable, we may use $\mathbb{E}^{M}_{\phi_X}(A_{1})<0$ to indicate that the associated additive functional is integrable.

\subsection{A quick review of Mittag-Leffler process}\label{section: Mittag-Leffler process}

Next, we define quantities that we use to describe the asymptotic behavior of additive functionals of Harris recurrent processes. We focus on our setup and assume the random process $X_t$ is manifold-valued, while is can be easily generalized. See the setup in \cite{limit_theorems_null} for details.

\begin{definition}
\label{stable-inc-process}\cite[Definition 2.5]{limit_theorems_null} For $\alpha \in (0,1)$, a {\em stable increasing process} of index $\alpha$ is a process $S_\alpha(t)$, $t \ge 0$, with the properties that 
 \begin{enumerate}
 \item all paths of $S_\alpha$ are c\`adl\`ag, non-decreasing, and $S_\alpha(0) = 0$ almost surely;
 \item the process has independent and stationary increments satisfying
$ \mathbb{E}[e^{-\lambda S_\alpha(t)}] = e^{-\lambda^\alpha t}$ for all $t, \lambda \ge 0$.
 \end{enumerate}
For $\alpha = 1$, we define the deterministic process $S_1(t)  = t$.
\end{definition}

By construction, the stable increasing process of index $\alpha$ defines a unique probability law on the Skorohod space $D(\mathbb{R}_+,\mathbb{R})$ with Borel $\sigma$-algebra and canonical filtration. See \cite[Definition 2.5]{limit_theorems_null} for more details.

\begin{definition}\cite[Definition 2.6]{limit_theorems_null}\label{MLprocess} 
For $\alpha \in (0,1)$, the process inverse for $S_\alpha(t)$ is the {\em Mittag-Leffler process} of index $\alpha$, denoted as $g_\alpha(t)$ and defined by 
$g_\alpha(t) = \inf\{s > 0: S_\alpha(s) > t\}$ for all $t\geq0$.
For $\alpha = 1$, we define the deterministic process $g_1(t)  = t$.
\end{definition}

By construction, then almost surely $g_\alpha(0) = 0$ and $g_\alpha(t)$ is c\`adl\`ag, continuous, and increasing to $\infty$ as $t \rightarrow \infty$.

\begin{proposition}\cite[Remark 2.8]{limit_theorems_null}\label{MLprocess_and_distr}
Let $g_\alpha$ be the Mittag-Leffler process of index $\alpha \in (0,1]$. Then, $g_\alpha(1) \sim G_\alpha$, where $G_\alpha$ is a Mittag-Leffler random variable of index $\alpha$.
\end{proposition}

\begin{remark}
We note that two closely related but distinct families of Mittag-Leffler random variables appear in the literature. To avoid confusion, we briefly describe their construction and relationship. Both are associated with the {\em Mittag-Leffler function} 
$E_\alpha(z) := \sum_{r=0}^\infty \frac{z^r}{\Gamma(1+\alpha r)}$ defined for $\alpha \in \mathbb{C}$ with $\Re(\alpha)>0$.
For $\alpha \in (0,1]$, the first family  has cumulative distribution function $1 - E_\alpha(-s^\alpha)$. 
These distributions are heavy-tailed and do not possess finite mean for $0 < \alpha \le 1$\cite{pillai1990mittag}.
The second family, which is relevant for our purposes (see Proposition \ref{MLprocess_and_distr}), arises from the inverse stable subordinator. In contrast to the first family, these distributions are not heavy-tailed and admit finite moments given by \cite{kasahara1984limit}
$\mathbb{E}(g_\alpha^n) = \frac{n!}{\Gamma(1 + \alpha n)}$, $n\in \mathbb{N}$. By definition, when $\alpha=0$, this family reduces to the exponential random variable with parameter $1$, which is precisely the case appearing in the classical Darling-Kac theorem (Theorem \ref{Theorem Darling-Kac Theorem0}).
\end{remark}

\subsection{Long-Time Asymptotics for Additive Functionals}\label{section: Long-Time Asymptotics for Additive Functionals}

While we do not directly need it, for the sake of self-containedness we mention a stronger ratio limit theorem called {\em Chacon-Ornstein limit theorem} \cite{LOCHERBACH20081301}. Recall that we call a measurable function $f:M\to \mathbb{R}_+$ {\em special} \cite[Definition 2.3]{LOCHERBACH20081301} when the function $x\mapsto \mathbb{E}_x\int_0^\infty \exp\left[-\int_0^th(X_s)ds\right]f(X_t)dt$ is bounded for all bounded and positive measurable functions $h$ such that $\langle \phi_X,h\rangle_M:= \int_M h(x) \,\phi_X(dx)>0$. Since $X_t$ is strongly Feller, all positive bounded functions with compact support are special \cite[Corollary 2.22]{LOCHERBACH20081301}. We have a stronger ratio limit theorem called Chacon-Ornstein limit theorem \cite{LOCHERBACH20081301}. For two special functions $f,g$ such that $\langle \phi_X,g\rangle_M>0$, for any initial measures $\lambda_1,\lambda_2$, $\phi_X$-a.s. we have 
\begin{align}
\lim_{T \rightarrow \infty} \frac{\mathbb{E}_{\lambda_1}\int_0^T f(X_s) \, ds}{\mathbb{E}_{\lambda_2}\int_0^T g(X_s)\,ds} = \frac{\langle \phi_X, f \rangle_M}{\langle \phi_X, g \rangle_M} \,.
\end{align}

We now state the Darling-Kac theorem that describe the asymptotic behavior of the proposed estimators. 
Recall that for a random process $X_t$, the \textit{quadratic variation} is defined by $[X_t] :=\lim_{n \rightarrow \infty} \sum_{k = 0}^{n-1} (X_{t(k+1)/n} - X_{tk/n})^2$, whenever the limit exists in probability.  Since we work with It\^o processes with continuous sample paths, we do not distinguish between predictable quadratic variation and quadratic variation. Moreover, if $M_t$ is a continuous, square-integrable, and real-valued local martingale, then by the {\em Doob-Meyer decomposition} \cite[Theorem 4.10, pg. 24]{brownian_motion}, the process $M_t^2 - [M_t]$ is itself a local martingale. In our analysis, we encounter additive functionals like $A_t:=\int^t_0 g(X_t)dW_t$, where $X_t$ is an It\^o process and {$\langle \phi_X, |g|\rangle_M < \infty$}. In this case, $[A_t]=\int^t_0 g^2(X_t)dt$. To study the asymptotic behavior of such terms, we will rely on Theorem 3.16 in \cite{limit_theorems_null}. Recall notation in Theorem \ref{Theorem Darling-Kac Theorem0}.

\begin{proposition}\cite[Theorem 3.16 and the following discussion]{limit_theorems_null}\label{thm3.16_limit_theorems_null} 
Let $A_t$ be a $p$-dim additive functional of a Harris recurrent $X_t$ with invariant measure $\phi_X$. Assume $A_t$ is a locally square integrable local martingale that is locally bounded. Denote $J\in \mathbb{R}^{p\times p}$ such that $J_{k,l}:=\mathbb E^M_{\phi_X}([e_k^\top A_1,\ e_l^\top A_1])$,  
and assume $J_{k,k}<\infty$ for $k=1,\ldots,p$. Then, for some $0 < \alpha \le 1$, we have a weak convergence of pairs 
\begin{align*}
   \left(\frac{A_{sT}}{\sqrt{\Upsilon(T)}},\ \frac{[A_{sT}]}{\Upsilon(T)}\right)_{s \ge 0} \xrightarrow[]{\hspace{0.1cm}d\hspace{0.1cm}} ( J^{1/2} W_\alpha,\ J g_\alpha ) 
\end{align*}
in $D(\mathbb{R}_+,\mathbb{R})$ as $T \rightarrow \infty$  under $\mathbb P_x$ for all $x\in M$, where $g_\alpha$ is a Mittag-Leffler process of order $\alpha$, $W_\alpha:=W(g_\alpha)$, and $W$ is a $p$-dim standard Brownian motion independent of $g_\alpha$. 
\end{proposition}

Note that $W(g_\alpha)$ is a {\em centered Gaussian mixture} with covariance $g_\alpha$, or it can be viewed as a time changed Brownian motion with $g_\alpha$; that is, a Brownian motion whose variance is randomized by a Mittag-Leffler clock. When $\alpha<1$, the marginal of $W(g_\alpha)$ is non-Gaussian. In our application, we encounter $t=1$, when $W_\alpha(1)|g_\alpha(1)\sim N(0,g_\alpha(1)I_p)$, or $W_\alpha(1)=\sqrt{g_\alpha(1)}Z$, where $Z\sim N(0,I_p)$ and $Z$ is independent of $g_\alpha(1)$.  
While the quadratic variation $[A_{sT}]$ is itself an additive functional, Proposition \ref{thm3.16_limit_theorems_null} treats the special case that $A_{sT}$ has martingale behavior.

The next inequality, {\em Burkholder-Davis-Gundy (BDG) inequality}, is also necessary, particularly when we analyze the diffusion estimator that involves triple discretization.
\begin{theorem}
\label{BDG-ineq}\cite[Theorem 3.28, page 166]{brownian_motion}
For a continuous local martingale $M_t$ with $M_0 = 0$ and any $p > 0$, there exist process-independent constants $c_p$ and $C_p$ so that for any stopping time $T$, 
\begin{align*}
c_p \mathbb{E}(\left[ M_T\right]^p) \le  \mathbb{E}((M_T^*)^{2p}) \le  C_p \mathbb{E}(\left[ M_T\right]^p)\,,
\end{align*}
where $M_T^* = \sup_{t \le T} |M_t|$ is the maximum process.
\end{theorem}

\subsection{It\^o's formula}\label{section Ito formula summary}

Recall that the SDE $X_t\in M$ in the Stratonovich form satisfies 
\[
dX_t=\nu(X_t)dt+\sigma_\alpha(X_t)\circ dW^\alpha(t)\,, 
\]
which appears in the It\^o form as
\[
dX_t=\underbrace{\Big(\nu+\frac{1}{2}\sum_\alpha \nabla_{\sigma_\alpha}\sigma_\alpha\Big)(X_t)}_{:=\mu(X_t)}dt+\sigma_\alpha(X_t) dW^\alpha(t).
\]
With this fact, here is another way to see why curvature comes into play and why projection is needed in the drift estimator.
Consider $f = \iota : M \to \mathbb{R}^p$, where $\iota$ is the isometric embedding of $(M,g)$ into $\mathbb{R}^p$. The explicit Euclidean-embedded process
$Z_t = {\iota}(X_t)$ satisfies
\begin{align*}
 dZ_t = &\,{\iota}_* \nu(Z_t) \, dt + {\iota}_* \sigma_\alpha(Z_t) \circ dW^\alpha_t=\underbrace{\Big( {\iota}_* \nu+ \frac{1}{2} D_{{\iota}_* \sigma_\alpha}({\iota}_* \sigma_\alpha) \Big)(Z_t)}_{\neq \iota_*\mu(Z_t)} \, dt + {\iota}_* \sigma_\alpha(Z_t) \, dW^\alpha_t, 
\end{align*}
where $D_{{\iota}_* \sigma_\alpha}({\iota}_* \sigma_\alpha)$ denotes the covariant derivative of the vector field $\iota_* \sigma_\alpha$ along itself in the ambient space $\mathbb{R}^p$, which satisfies
\[
D_{{\iota}_* \sigma_\alpha}({\iota}_* \sigma_\alpha)(\iota(x))={\iota}_*|_x\nabla_{\sigma_\alpha}\sigma_\alpha+\texttt{I\!\!I}_x(\sigma_\alpha,\sigma_\alpha)
\]
and contains a nontrivial normal component. 
This computation shows that if we estimate the drift term using the Euclidean estimator, the normal component biases the estimator and it is the dynamics of the extrinsic $Z_t=\iota(X_t)$ that we are dealing with, not the intrinsic drift term of $X_t$ that we want.

Next we document the It\^o's formula. Consider $f\in C^2(M,\, \mathbb{R}^p)$ with $\texttt{supp} f \subset B_r(x)$ for $r < \texttt{inj}_x(M)$. For $0\leq t_1\leq t_2$, the It\^o's formula gives 
\begin{align}
f(X_{t_2}) - \,& f(X_{t_1})  = \int_{t_1}^{t_2} (\nu f)(X_s) ds + \int_{t_1}^{t_2} (\sigma_\alpha f)(X_s) \circ dW^\alpha_s \nonumber\\
 =\,& \int_{t_1}^{t_2} \underbrace{\Big[(\nu f)(X_s) + \frac{1}{2}  \sum_{l=1}^r (\sigma_\alpha\sigma_\alpha f)(X_s)\Big]}_{:=\mu_f(X_s)}  \, ds + \int_{t_1}^{t_2}  \underbrace{(\sigma_\alpha f)(X_s)}_{:=\sigma_{f,\alpha}(X_s)} \, dW^\alpha_s \label{ito-integral}\,.
\end{align}
Recall that the second equality comes from the definition of Hessian $ \texttt{Hess}^q(X,Y) =X(Yq)-(\nabla_XY)q$ \cite[page 23]{hsu_stoch_anal_mflds} for the conversion from Stratonovich formulation to It\^o formulation.  We will use this notation throughout the proofs.

\subsection{Sufficient conditions for Assumption \ref{manifold-ass2} when $M$ is noncompact}\label{section sufficient conditions for assumption 2 when M is noncompact}

For Assumption \ref{manifold-ass2}(1), the simplest one is the linear growth bound and locally Lipschitz condition via extension. Denote the smooth extensions of $\nu$ and ${\sigma}_k$, $k=1,\ldots,r$, in \eqref{strat-SDE} to $\mathbb{R}^p$ as $\widetilde{\nu}$ and $\widetilde{\sigma}_k$, and $\widetilde{\sigma}(x):=[\widetilde{\sigma}_1(x),\ldots,\widetilde{\sigma}_r(x)]\in \mathbb{R}^{p\times r}$. The {\em linear growth bound} condition is 
\[
\|\widetilde{\sigma}(x)\|+|\widetilde{\nu}(x)|\leq C(1+|x|)
\]
for some $C>0$, and their derivatives of any order are bounded. The {\em locally Lipschitz} is for any $N>0$, there exists $C_N>0$ such that 
\[
\|\widetilde{\sigma}(x)-\widetilde{\sigma}(x')\|+|\widetilde{\nu}(x)-\widetilde{\nu}(x')|\leq C_N|x-x'|
\] 
for any $x,x'\in B_0(N)$. The linear growth bound and locally Lipschitz condition of $\widetilde{\nu}$ and $\widetilde{\sigma}$ guarantees the solution does not explode \cite[(1.1.11)]{hsu_stoch_anal_mflds}.
See \cite[Theorem 2.1.1 and Corollary 2.1.2]{wang2014analysis} for other conditions. 

The Foster-Lyapunov criterion \cite{meyn1993survey} provides a sufficient condition for Assumptions \ref{manifold-ass2}(1) and \ref{manifold-ass2}(4). A function $V:M\to \mathbb{R}^+$ is called a {\em norm-like} function if $V(x)\to \infty$ as $x\to \infty$. The norm-like functions are essentially a specific type of Lyapunov function tailored for studying stability and ergodicity of Markov processes.
Assume there exist constants $c\geq 0$ and a norm-like function $V$ so that
\[
(LV)(x)\leq cV(x)\,,
\]
then $\mathbb{P}\{d_g(X_t,x)=\infty \}=0$ for any $x\in M$, which is another sufficient condition for Assumption \ref{manifold-ass2}(1). 
To discuss Harris recurrence, we need more quantities. Define $K_a:=\int_0^\infty P_ta(dt)$, where $a$ is a probability distribution on $\mathbb{R}_+$. A non-empty set $C\in \mathcal{B}(M)$ is called $\phi_a$-petite if $\phi_a$ is a non-trivial measure on $\mathcal{B}(M)$ and $K_a(x,\cdot)\geq \phi_a(\cdot)$ for all $x\in C$. The notion of a petite set is related to irreducibility of $X_t$.
It is known that  \cite[Theorem 3.2]{meyn1993stability} if all compact subsets of $M$ are {\em petite}, and if there exists a compact set $C\subset M$, a constant $d>0$, and a norm-like function $V$ such that the {\em condition for recurrence} 
\[
(LV)(x)\leq d\mathbbm{1}_C(x)
\]
holds for all $x\in M$, then $X_t$ is Harris recurrent. 
Next, if $X_t$ is Feller and non-explosive, and the {\em positive recurrence condition}
\[
(LV)(x)\leq -cf(x)+d\mathbbm{1}_C(x)
\]
holds for all $x\in M$ for some $c,d>0$, $f:M\to [1,\infty)$, a compact set $C\subset M$, and $V\geq 0$, then $X_t$ is positive Harris recurrent \cite[Theorem 4.5]{meyn1993stability}. Note that this sufficient condition does not involves irreducibility. Another sufficient condition for a non-explosive $X_t$ being positive Harris recurrent is that the positive recurrence condition holds with $C$ a closed petite set and that $V$ is bounded on $C$ \cite[Theorem 4.2]{meyn1993stability}.
The existence of a norm-like function satisfying the above conditions guarantees that the drift term pulls the process back toward a compact region, preventing escape to infinity. 
In practice, constructing an explicit norm-like function can be technically demanding.

There are also sufficient conditions for Assumption \ref{manifold-ass2}(3). We call a probability measure $\mu$ on $\mathcal{B}(M)$ {\em quasi-invariant} of the Markov operator $P_t$ if $\mu P_t$ is absolutely continuous with related to $\mu$, where $\mu P_t(A):=\mu(P_t\boldsymbol{1}_A)$ and $A\in \mathcal{B}(M)$. Assume $L=\Delta+Z$, where $Z\in C^\infty(TM)$. Define 
\[
\texttt{Ric}_Z(X,Y)=\texttt{Ric}(X,Y)-\langle \nabla_X Z,Y\rangle\,,
\]
where $X,Y,Z\in C^\infty(\Gamma M)$.
By \cite[Theorem 2.3.3]{wang2014analysis}, $\texttt{Ric}_Z\geq K$ for some $K\in \mathbb{R}$ is equivalent to the {\em Harnack-type} inequality that when $p>1$, for any bounded non-negative measurable function $f$:
\[
(P_tf)^p(x)\leq P_t(f^p)(y)\exp\left(\frac{Kpd_g(x,y)^2}{2(p-1)(e^{2Kt}-1)}\right)\,.
\]
If $\texttt{Ric}_Z\geq K$ for $K\in \mathbb{R}$ and there exists a quasi-invariant probability measure $\mu$ of $P_t$, then $P_t$ is strong Feller \cite[Theorem 1.3.1]{wang2014analysis}. Moreover, $P_t$ has at most one invariant probability measure, and if it has, the kernel $p_t$ of $P_t$ related to the invariant probability measure is strictly positive. 
See \cite{seeley1984contraction,ishiwata2024graph} for sufficient conditions for Feller.

Assumption \ref{manifold-ass2}(3) might be the most tricky one. It holds when the diffusion kernel has a Gaussian-like control. Such Gaussian-like control holds, for example, when $L=\Delta$ and $M$ has non-negative Ricci curvature or is quasi-isometric to those with non-negative Ricci curvature, or that cover a compact manifold with deck-transformation group having polynomial volume growth \cite{grigor2009heat}. It is possible that the diffusion kernel does not have Gaussian-like control but fulfills Assumption \ref{manifold-ass2}(3). A concrete example is when $M$ is a 3-dim Cartan-Hadamard manifold $\mathbb{H}^3_k$ with constant curvature $-k^2<0$, where $k>0$. Its heat kernel is $p_t(x,y)=\frac{1}{(4\pi t)^{3/2}}\frac{kd(x,y)}{\sinh(kd(x,y))}\exp\big(-\frac{d(x,y)^2}{4t}-k^2t\big)$ \cite{grigor1994heat}. See \cite{grigor2009heat} for various examples of non-compact manifolds with finite ends, whose heat kernel has a non-Gaussian control. 
When the drift term is nontrivial, here is a relevant control. Assume again $L=\Delta+Z$. Assume $Z=\nabla V$ for some $V\in C^2(M)$ and the invariant measure is $\mu(dx):=e^{V(x)}dx$. Let $p_t$ be the diffusion kernel of $P_t$ with related to $\mu$. When $\texttt{Ric}_Z\geq K$ for $K\in \mathbb{R}$, for any $\delta>0$, there exists $c(\delta)>0$ such that a Gaussian-like upper bound holds \cite[Theorem 2.4.4]{wang2014analysis}:
\[
p_t(x,y)\leq (\mu(B_{\sqrt{t}}(x)) \mu(B_{\sqrt{t}}(y)))^{-1/2}\exp\left(c(\delta)(1+t)-\frac{\rho(x,y)^2}{2\delta t}\right)
\] 
for all $t>0$ and $x,y\in M$. If we further assume $\phi_X$ is a probability measure, we have the lower bound
\[
p_t(x,y)\geq \exp\left(-K\rho(x,y)^2/2(e^{Kt}-1)\right)
\] 
for $t>0$ and $x,y\in M$. The above examples and sufficient conditions show the complication of the diffusion kernel control, and the necessity of Assumption \ref{manifold-ass2}(3).

 \subsection{General tools}

In the analysis, we encounter various kinds of error controls in the ratio form. We summarize common ones in the following lemmas.

\begin{lemma}
    \label{reciprocal-lemma1} 
    Suppose that $X_t \xrightarrow[]{\hspace{0.1cm}d\hspace{0.1cm}} X$ as $t \rightarrow \infty$, with $X > 0$ almost surely. Moreover, suppose $A_t \xrightarrow[]{\hspace{0.1cm}p\hspace{0.1cm}} A$ and $B_t \xrightarrow[]{\hspace{0.1cm}p\hspace{0.1cm}} B,$ for $A,B \in \mathbb{R}$ and $B \neq 0$, and $C_t,D_t \xrightarrow[]{\hspace{0.1cm}p\hspace{0.1cm}} 0$ as $t \rightarrow \infty$. Then, we have a convergence in probability for the following ratio:
    \begin{align*}
 \frac{X_tA_t + C_t}{X_tB_t + D_t} \xrightarrow[]{\hspace{0.1cm}p\hspace{0.1cm}} \frac{A}{B}
    \end{align*}
    as $t \rightarrow \infty$.
\end{lemma}

\begin{proof}  
   Re-write the ratio as
    $\frac{X_tA_t + C_t}{X_tB_t + D_t} = \frac{A_t + c_t}{B_t + d_t}$, where $c_t = C_t/X_t$.
    We have $c_t = o_p(1)$ and $d_t = o_p(1)$.  
    It follows that $A_t+ c_t \xrightarrow[]{\hspace{0.1cm}p\hspace{0.1cm}} A$ and $B_t+ d_t \xrightarrow[]{\hspace{0.1cm}p\hspace{0.1cm}} B$ and so by the continuous mapping theorem, $\frac{A_t + c_t}{B_t + d_t} \xrightarrow[]{\hspace{0.1cm}p\hspace{0.1cm}} \frac{A}{B}$.
\end{proof}

We have another technical lemma related to ratios.
\begin{lemma}
    \label{reciprocal-lemma2} Suppose $X_t \xrightarrow[]{\hspace{0.1cm}d\hspace{0.1cm}} X$ and $Y_t \xrightarrow[]{\hspace{0.1cm}d\hspace{0.1cm}} Y$ as $ t \rightarrow \infty$, where $X, Y>0$ almost surely.     Suppose $A_t \xrightarrow[]{} 0 $ and $B_t \xrightarrow[]{} 0$ deterministically. Then, 
    \begin{align*}
 \frac{X_t + A_t}{Y_t + B_t} &= \frac{X_t}{Y_t}  +  O_p(A_t) + O_p(B_t)\,.
    \end{align*}
\end{lemma}

\begin{proof}
Rewrite $\frac{X_t+A_t}{Y_t + B_t} = \frac{X_t}{Y_t} +A_t \left( \frac{1}{Y_t} \frac{Y_t}{Y_t + B_t}\right) - B_t\left(\frac{X_t}{Y_t^2}\frac{Y_t}{Y_t + B_t} \right)$.
Since $\frac{Y_t}{Y_t + B_t}\xrightarrow[]{\hspace{0.1cm}p\hspace{0.1cm}} 1$ by Lemma \ref{reciprocal-lemma1} and $\frac{1}{Y_t}\xrightarrow[]{\hspace{0.1cm}d\hspace{0.1cm}}\frac{1}{Y}$ and $\frac{X_t}{Y^2_t} \xrightarrow[]{\hspace{0.1cm}d\hspace{0.1cm}}\frac{X}{Y^2}$ by the continuous mapping theorem, we have $\frac{1}{Y_t}\left(\frac{Y_t}{Y_t + B_t}\right)\xrightarrow[]{\hspace{0.1cm}d\hspace{0.1cm}}\frac{1}{Y}$ and $\frac{X_t}{Y^2_t}\frac{Y_t}{Y_t + B_t}\xrightarrow[]{\hspace{0.1cm}d\hspace{0.1cm}}\frac{X}{Y^2}$ by the Slutsky's theorem. 
Thus, we have $A_t \left( \frac{1}{Y_t} \frac{Y_t}{Y_t + B_t}\right)= O_p(A_t)$ and $ B_t\left(\frac{X_t}{Y_t^2}\frac{Y_t}{Y_t + B_t} \right)=O_p(B_t)$ and hence the proof. 
\end{proof}

Next, we summarize the stochastic integration by parts. To be self-contended, we include a portion of Theorem 2.2 of \cite{veraar2012stochastic} in our Lemma \ref{lem:stoch-fubini}.
\begin{lemma}[Stochastic Fubini (\cite{veraar2012stochastic} Theorem 2.2, Assertion (iii))] \label{lem:stoch-fubini}
Let $(\Omega,\mathcal A,\mathbb P)$ be a probability space equipped with a complete, right-continuous filtration $(\mathcal F_t)_{t\ge 0}$. 
Let $(X,\Sigma,\mu)$ be a $\sigma$-finite measure space and let $S_t=M_t+A_t$ be a continuous semimartingale.
Let $\psi : X\times [0,T]\times\Omega\to\mathbb R$ be progressively measurable, and assume that, for almost all $\omega$,
\begin{align}
& \int_X \Big( \int_0^T |\psi(x,t,\omega)|^2 \, d[M](t,\omega) \Big)^{1/2} d\mu(x) < \infty, \label{sq-int-condition-stoch-fubini}\\
& \int_X \int_0^T |\psi(x,t,\omega)| \, d|A|(t,\omega)\, d\mu(x) < \infty.  \label{finite-variation-condition-stoch-fubini}
\end{align}
Then for almost all $\omega$ and every $t\in[0,T]$,
\begin{align}
\int_X \left( \int_0^t \psi(x,s,\omega)\, dS_s(\omega) \right) d\mu(x)
=
\int_0^t \left( \int_X \psi(x,s,\omega)\, d\mu(x) \right) dS_s(\omega).
\label{eq:SFubini}
\end{align}
\end{lemma}

\begin{lemma}\label{integration-by-parts} 
Let $0 < a < b$ and $f,g \in C(M)$.  
Suppose $X_t$ is adapted to the Brownian motion $W_t$, and $M_t$ is a continuous martingale. Then,
    \begin{align}
      &  \int_{a}^b \left( \int_{a}^s M_v f(X_v) dv\right)g(X_s) ds =  \int_a^b M_s f(X_s)\left(\int_s^b g(X_v) dv \right) ds \,,
\label{time-integral-by-parts}  \\
     & \int_{a}^b \left( \int_{a}^s M_v f(X_v) dW_v \right) g(X_s) ds =  \int_a^b M_sf(X_s)\left(\int_s^b g(X_v) dv \right) dW_s\label{stoch-integral-by-parts-v2}  \\
     &\quad = \left(\int_a^b M_sf(X_s) dW_s \right)\left(\int_a^b g(X_t) dt \right)
     - \int_a^b M_sf(X_s)\left(\int_a^s g(X_v) dv \right) dW_s  \nonumber \,.
    \end{align}
    Moreover when $M_t\equiv 1$ and $g(X_t) \equiv 1$, we have the following stochastic integration by parts formulae:
     \begin{align}
 \int_{a}^b \left( \int_{a}^s f(X_v) dv\right) ds &=  \int_a^b (b-s)f(X_s) ds\,, \label{special-case-stoch-fub1}\\
     \int_{a}^b \left( \int_{a}^s f(X_v) dW_v \right) ds &=  \int_a^b (b-s)f(X_s) dW_s\,.   \label{special-case-stoch-fub2}  
    \end{align}
\end{lemma}

\begin{proof}
For \eqref{time-integral-by-parts}, consider the triangle $\{(v,s):a\le v\le s\le b\}$ in the $(v,s)$-plane and
apply Fubini's theorem.
For \eqref{stoch-integral-by-parts-v2}, we verify that the stochastic Fubini theorem
stated in Lemma~\ref{lem:stoch-fubini} applies. We take the parameter space $X=[0,b]$ with Lebesgue measure,
set $T=b$, and take the semimartingale $S_t=W_t$.
Define the progressively measurable function $\psi:[0,b]\times[0,b]\times\Omega\to\mathbb R$ by
\begin{align*}
    \psi(v,t,\omega)
    := g(X_v(\omega))\,M_t(\omega)\,f(X_t(\omega))\,\mathbf 1_{\{a\le t\le v\le b\}},
    \qquad (v,t)\in[0,b]^2 .
\end{align*}
Since $S_t=W_t$, we have $A_t\equiv 0$, and thus condition \eqref{finite-variation-condition-stoch-fubini} holds.
It suffices to verify the square-integrability condition \eqref{sq-int-condition-stoch-fubini}; i.e. for almost all $\omega$,
$\int_{0}^{b}\Big(\int_{0}^{b}|\psi(v,t,\omega)|^2\,dt\Big)^{1/2}dv<\infty$.
Using $\|g\|_\infty<\infty$ and that $\psi(v,t,\omega)=0$ unless $a\le t\le v\le b$, we obtain for almost all $\omega$,
\begin{align*}
&\int_{0}^{b}\Big(\int_{0}^{b}|\psi(v,t,\omega)|^2\,dt\Big)^{1/2}dv
=\int_{a}^{b} |g(X_v(\omega))|\Big(\int_{a}^{v} |M_t(\omega) f(X_t(\omega))|^2\,dt\Big)^{1/2}dv \\
\le\,&  \|g\|_\infty (b-a)\Big(\int_{a}^{b} |M_t(\omega) f(X_t(\omega))|^2\,dt\Big)^{1/2}<\infty\,,
\end{align*}
where the last inequality follows since $M$ has continuous sample paths and is therefore
bounded on the compact interval $[a,b]$ almost surely; together with
$\|f\|_\infty<\infty$ this implies
$\int_a^b |M_t(\omega) f(X_t(\omega))|^2 dt<\infty$.
Therefore, Lemma~\ref{lem:stoch-fubini} yields \eqref{eq:SFubini}, so that we obtain
\begin{align*}
\int_{0}^{b}\left(\int_{0}^{b}\psi(v,r,\omega)\,dW_r(\omega)\right)dv
=
\int_{0}^{b}\left(\int_{0}^{b}\psi(v,r,\omega)\,dv\right)dW_r(\omega),
\qquad\text{a.s.}
\end{align*}
By the definition of $\psi$, suppressing the $\omega$ notation, the left-hand side equals
\begin{align*}
\int_{a}^{b} g(X_v)\left(\int_{a}^{v} M_r f(X_r)\,dW_r\right)dv
=
\int_{a}^{b}\left(\int_{a}^{s} M_v f(X_v)\,dW_v\right) g(X_s)\,ds,
\end{align*}
and the right-hand side equals
\begin{align*}
\int_{a}^{b} M_r f(X_r)\left(\int_{r}^{b} g(X_v)\,dv\right)\,dW_r
=
\int_{a}^{b} M_s f(X_s)\left(\int_{s}^{b} g(X_v)\,dv\right)\,dW_s.
\end{align*}
This gives the first equality in \eqref{stoch-integral-by-parts-v2}.
The second equality in \eqref{stoch-integral-by-parts-v2} follows from the identity
\[
\int_s^b g(X_v)\,dv=\int_a^b g(X_v)\,dv-\int_a^s g(X_v)\,dv.
\]
When $M_t \equiv 1$ and $g \equiv 1$, the remaining claim follows immediately.
\end{proof}

\section{Technical lemmas}\label{section technical lemma proofs}
In this section, we provide technical lemmas to streamline the proofs of our main results. 
The first lemma involves calculation on the manifold setup.

\begin{lemma}\label{lemma: expectation expansion}  
Assume Assumptions \ref{manifold-ass}, \ref{manifold-ass2}, and \ref{lebesgue-dens-ass}.  { Define}
\begin{align}
   U_1(x) := \frac{1}{h^d} \int_0^1 {K}\left(\frac{
   \mathcal{D}_x(X_s)}{h}\right)  \mathcal{H}(X_s) \, ds \label{L(x)-def}\,,
\end{align}
where $\mathcal{H}\in C_b(M)$ is non-negative and $K \in C(\mathbb{R})$ is non-negative and supported on $[0,L]$, where $L>0$. 
Suppose $\mathcal{D}_x(x')=\|\iota(x) -\iota(x')\|$. When $h\to 0$, we have   
\begin{align*}
    \mathbb{E}^M_\lambda ( U_1(x) )  
   \to C_X\mathcal{H}(x) p_X(x)  \int_{\mathbb{R}^d}  K(\| u \|) du \,,
\end{align*} 
 where $C_X$ is a constant defined in Proposition \ref{2.20and4.3}.

Define
\begin{align}
   U^-_1(x) := \frac{1}{h^d} \int_0^1 {K}\left(\frac{
   \mathcal{D}_x(X_s)}{h}\right) ( \mathcal{H}(X_s)-\mathcal{H}(x)) \, ds \label{L(x)-def2}\,,
\end{align} 
where we further assume $\mathcal{H}\in C^3(M)$ and $K \in C^3(\mathbb{R})$. Since $\mathcal{D}_x(x')=\|\iota(x) -\iota(x')\|$, we have
\begin{align*}
 &\frac{ \mathbb{E}^M_\lambda  ( U^-_1(x))}{h^2} 
   \to 
\frac{C_X}{2}\left(p_X(x)\Delta^M\mathcal{H}(x)+2\nabla^M \mathcal{H}(x)\cdot \nabla^M p_X(x) \right)\int_{\mathbb{R}^d}  K( \| u \|)\|u\|^2 du
\end{align*}
when $h\to 0$, where $\Delta^M$ and $\nabla^M$ are the Laplace-Beltrami operator and covariant derivative on $(M,g)$. 
 
\end{lemma}
\begin{proof}
By Proposition \ref{2.20and4.3} with $f(\cdot)=\frac{1}{h^d} K\left(\frac{\mathcal{D}_x(\cdot) }{h}\right)\mathcal{H}(\cdot)$, we obtain  
\begin{align*}
 \mathbb{E}^M_\lambda ( U_1(x) )  
 &=C_X \left\langle \phi_X, \frac{1}{h^d}K\left(\frac{\mathcal{D}_x(\cdot) }{h}\right)  \mathcal{H}( \cdot) \right\rangle_M \\
&= C_X  \int_M \frac{1}{h^d}K\left(\frac{
\mathcal{D}_x(x') }{h}\right) \mathcal{H}(x')p_X(x') dV_g(x') \,.
\end{align*}
Similarly, by linearity of $\mathbb{E}^M_\lambda$, apply Proposition \ref{2.20and4.3} twice and get
\[
\mathbb{E}^M_\lambda ( U^-_1(x) )=C_X  \int_M \frac{1}{h^d}K\left(\frac{
\mathcal{D}_x(x') }{h}\right) (\mathcal{H}(x')-\mathcal{H}(x))p_X(x') dV_g(x')\,.
\]
Since $\mathcal{D}_x(x')=\|\iota(x)-\iota(x')\|_{\mathbb{R}^p}$, recall that for $x'=\exp_x(t\theta)$, where $\|\theta\|=1$, we have \cite[Lemma 4]{malik2019connecting} 
\[
\iota(x)-\iota(x')=t\iota_*\theta+\frac{1}{2}t^2\texttt{I\!\!I}_x(\theta,\theta)+\frac{1}{6}t^3\nabla_\theta\texttt{I\!\!I}_x(\theta,\theta) +O(t^4)\,,
\]
where $\texttt{I\!\!I}_x$ is the second fundamental form of the embedding $\iota$, and hence
\[
\|\iota(x)-\iota(x')\|_{\mathbb{R}^p}=t+\frac{1}{6}t^3\iota_*\theta^\top\nabla_\theta\texttt{I\!\!I}_x(\theta,\theta)+O(t^4)
\]
when $t$ is sufficiently small. Also recall that under the exponential map with polar coordinate, when $v=t\theta$, $\|\theta\|=1$ and $t>0$ is sufficiently small, we have \cite[Lemma SI.1]{wuwu18}
\[
|\det d\text{exp}_x({v})|=t^{d-1}-\frac{1}{6}\texttt{Ric}_x(\theta,\theta)t^{d+1}+O(t^{d+2})\,.
\]
With a change of variables under the exponential map $\exp_x$ with polar coordinate on $\mathbb{R}^d$ and Assumption \ref{lebesgue-dens-ass}, we have by Taylor's expansion:
\begin{align*}
   &\int_M \frac{1}{{h}^{d}} K\left(\frac{\mathcal{D}_x(x') }{h}\right)  \mathcal{H}(x') p_X(x') dV_g(x') \\
   =\,& \int_{\mathbb{R}^d}\frac{1}{{h}^{d}} \left(K\left(\frac{t }{h}\right)+K'\left(\frac{t }{h}\right)\frac{t^3\iota_*\theta^\top\nabla_\theta\texttt{I\!\!I}_x(\theta,\theta)}{6h}+O(t^4)\right)  \\
   &\qquad \times \Big(\mathcal{H}(x)+t\nabla \mathcal{H}|_x\theta+\frac{t^2}{2}{\nabla}^2 \mathcal{H}|_x(\theta,\theta)+O(t^3)\Big) \\
   &\qquad \times \Big(p_X(x)+t\nabla p_X|_x\theta+\frac{t^2}{2}{\nabla}^2 p_X|_x(\theta,\theta)+O(t^3)\Big)\\
   &\qquad \times \Big(t^{d-1}-\frac{1}{6}\texttt{Ric}_x(\theta,\theta)t^{d+1}+O(t^{d+2})\Big)dtd\theta \,,
\end{align*} 
   where in the Taylor expansion of $K$ we use the fact that $t\leq Lh$.
Thus, by collecting all terms with $t^{d-1}$ power, we have 
\begin{align*}
\mathbb{E}^M_\lambda ( U_1(x) )
\xrightarrow[]{}  
C_X\mathcal{H}(x) p_X(x)  \int_{\mathbb{R}^d}  K( \| u \|) du\,,
\end{align*} 
as $h \rightarrow 0$, where the convergence holds since $ K$ is uniformly continuous with compact { support}. Regarding $\mathbb{E}^M_\lambda ( U^-_1(x) )$, since the dominant term is deleted and all terms with $t^d$ power are odd and get canceled out due to the kernel symmetry, we collect terms with $t^{d+1}$ power, and obtain
\begin{align*}
\frac{1}{h^2}\mathbb{E}^M_\lambda ( U^-_1(x) )
\xrightarrow[]{} 
\frac{1}{2}\Delta(\mathcal{H}(x) p_X(x)) \int_{\mathbb{R}^d}  K( \| u \|)\|u\|^2 du
\end{align*} 
as $h \rightarrow 0$.

\end{proof}

\begin{lemma}\label{Premathcal{B}-conv}
Assume Assumptions \ref{manifold-ass} and \ref{manifold-ass2} hold. 
Denote the LL-embedding of $X_t$ as $Z_t$, which is defined on $\mathcal{M}:=M\times [0,1]\times M$. Denote the generalized life-cycle decomposition determined from $Z_t$ as $\{S_m,R_m\}_{m=0}^\infty$, where $R_0=S_0=0$.
Set 
\begin{align}
N_t = \sum_{m \ge 1} \mathbf{1}\{S_m \le t\}\,,
\end{align}
where $t>0$. Asymptotically when $T\to\infty$, we have
\begin{align*}
N_T - \mathbb{E}^{\mathcal{M}}_\lambda(N_T) =O_p(\Upsilon(T)^{1/2}). 
   \end{align*}
and
\begin{align*}
\frac{N_{sT}}{\Upsilon(T)} \xrightarrow[]{\hspace{0.1cm}d\hspace{0.1cm}}C_X^{-1} g_\alpha(s),\ \ s>0\,,
\end{align*}
in $D(\mathbb{R}_+,\mathbb{R})$ as $T \rightarrow \infty$,
where $g_\alpha$ is the Mittag-Leffler process of index $\alpha\in (0,1]$.
 \end{lemma}
 
\begin{proof}
By Proposition \ref{proposition Zt is Harris recurrent}, $Z_t$ is Harris recurrent, and by definition, $N_t$ is an integrable additive functional of $Z_t$. First, by Proposition \ref{2.20and4.3}, 
\[
\mathbb{E}^M_{\phi_X}(N_1)=C_X^{-1}\mathbb{E}_{\xi}^{\mathcal{M}}(N_{R_1})\,, 
\]
where $\xi$ is the probability measure supported on a nontrivial compact set $C$ used in the construction of $Z_t$.
By the construction of the generalized life cycle decomposition, 
\[
N_{R_1}=\sum_{m \ge 1} \mathbf{1}\{S_m \le R_1\}=1\,.
\]
Therefore, by plugging $\xi(x)\otimes U(u)\otimes Q(x,u,dy)$ into the expectation, since $\int_M\int_0^1 Q(x,u)dydu=U^1(x,dy)$, and $\int_M U^1(x,dy)=1$ for any $x$, we have
\[
\mathbb{E}_{\xi}^{\mathcal{M}}(N_{R_1})=\int_{(x,u,y)\in\mathcal{M}}Q(x,u,dy)du 1 d\xi(x)=1
\]
since $\xi$ is a probability measure.

To study $\frac{N_{sT}}{\Upsilon(T)}$, note that we assume Proposition \ref{Theorem Darling-Kac Theorem0}(1) holds for $X_t$ in Assumption \ref{reg-ass}. We claim that the same assumption holds for $Z_t$ by the ratio limit theorem (see Theorem \ref{ratio-limit0SUPP}). For a bounded measurable function $f:\mathcal{M}\to \mathbb{R}$ satisfying $f(x,u,y)=f_1(x)$, where $f_1$ is a bounded measurable function on $M$, and another bounded measurable function $g:M\to \mathbb{R}$ so that $\langle g,\phi_X\rangle_M\neq 0$, we have 
\[
\frac{\mathbb{E}^{\mathcal{M}}_\lambda(\int_0^Tf(Z_t)dt)}{\mathbb{E}^{M}_\lambda(\int_0^Tg(X_t)dt)}
=\frac{\mathbb{E}^{\mathcal{M}}_\lambda(\int_0^Tf_1(Z^1_t)dt)}{\mathbb{E}^{M}_\lambda(\int_0^Tg(X_t)dt)}
=\frac{\mathbb{E}^{M}_\lambda(\int_0^Tf_1(X_t)dt)}{\mathbb{E}^{M}_\lambda(\int_0^Tg(X_t)dt)}
\to \frac{\langle f_1,\phi_X\rangle_M}{\langle g,\phi_X\rangle_M}
\]
as $T\to \infty$, where the second equality holds by Proposition \ref{LL-Prop2.8a} and the limit holds due to the ratio limit theorem. Therefore, $\mathbb{E}^{M}_\lambda(\int_0^Tg(X_t)dt)\sim \Upsilon(T)$ implies $\mathbb{E}^{\mathcal{M}}_\lambda(\int_0^Tf(Z_t)dt)\sim \Upsilon(T)$, and hence $Z_t$ satisfies the condition in Proposition \ref{Theorem Darling-Kac Theorem0}(1).
Since $0<\mathbb{E}^M_{\phi_X}(N_1)=C_X^{-1}<\infty$, we have
\begin{align*}
\frac{N_{sT}}{\Upsilon(T)} \xrightarrow[]{\hspace{0.1cm}d\hspace{0.1cm}} C_X^{-1} g_\alpha(s)
\end{align*}
in $D(\mathbb{R}_+,\mathbb{R})$ as $T \rightarrow \infty$,
where $g_\alpha$ is the Mittag-Leffler process of index $\alpha$.

For the additive functional $A_t:= N_t-\mathbb{E}^{\mathcal{M}}_\lambda(N_t)$, note that it is compensated and hence a martingale. 
So we can apply Proposition \ref{thm3.16_limit_theorems_null} to $A_t$ and obtain 
\begin{align*}
    \frac{A_{T}}{\sqrt{\Upsilon(T)}} \xrightarrow[]{\hspace{0.1cm}d\hspace{0.1cm}} \sqrt{\mathbb{E}^M_{\phi_X}([A_1])} W_\alpha(1)
\end{align*} 
as $T \rightarrow \infty$. It follows that $A_T =O_p(\Upsilon(T)^{1/2})$. 

\end{proof}

The following Master lemma generalizes results in \cite{bandi_moloche_2018} to the manifold setup, where we carefully take care the geometric challenges.

\begin{lemma}[Master lemma]\label{A1-analog}
Assume Assumptions \ref{manifold-ass}, \ref{manifold-ass2}, \ref{reg-ass}, and \ref{lebesgue-dens-ass} hold. Define
\begin{align}
   I(x) := \frac{1}{h^d} \int_0^T {K}\left(\frac{
   \mathcal{D}_x(X_s)}{h}\right)  \mathcal{H}(X_s) \, ds \label{L(x)-def}\,,
\end{align}
where $\mathcal{H}\in C(M)$ is non-negative and bounded and $K \in C(\mathbb{R})$ is non-negative and supported on $[0,L]$ for $L>0$. 
When $\mathcal{D}(x')=\|x-x'\|_{\mathbb{R}^p}$, as $h \rightarrow 0$ and $T \rightarrow \infty$, 
{  
 \begin{align*}
\frac{{ I(x)}}{\Upsilon(T)}  \xrightarrow[]{\hspace{0.1cm}d\hspace{0.1cm}}  g_\alpha(1)p_X(x) \mathcal{H}(x) \int_{\mathbb{R}^d}  K( \| u \|) du\,.
  \end{align*}
}
  \end{lemma}

Our drift and diffusion estimators naturally take the form of ratios of two additive functionals as in \eqref{L(x)-def}, where the denominator estimates the occupation density.

 \begin{proof}[Proof of Lemma \ref{A1-analog}]
Denote $\lambda$ to be the initial measure for $X_t$ on $M$. Recall that we assume $\lambda$ is compactly supported on $\mathsf C\subset M$ and $dV_g(\mathsf C)>0$ in Assumption \ref{manifold-ass2}. Adapt the LL embedding detailed in Section \ref{section LL embedding} with the probability measure $\xi$ equivalent to $\phi_X( \cdot \cap \mathsf C)$. Since we assume $p_X$ is strictly positive on $M$, we have $\phi_X(\mathsf C)>0$.
The LL-embedded $X_t$ is denoted as $Z_t$. For any initial measure $\lambda$ of $X_t$, let $\mathbb{E}^{\mathcal{M}}_\lambda$ be defined as in \eqref{extended-manifold-path-space-expectation}, where $\mathcal{M}:=M\times [0,1]\times M$. Recall the generalized life-cycle decomposition for the process, $\{S_m,R_m\}_{m=0}^\infty$, where $R_0=S_0=0$, determined from $Z_t$. Denote $ K_{h}(\cdot):=\frac{1}{h^{d}}{K}\left(\frac{{\mathcal{D}}_x(\cdot)}{h} \right)$. For $m\in \mathbb{N}\cup \{0\}$, set
\begin{align}
U_m :=  \int_{R_m}^{R_{m+1}}   K_{h}(X_s) \mathcal{H}(X_s)\, ds\,,\label{definition U1 in lemma proof}
 \end{align}
 which is a decomposition of ${ I(x)}$. Recall Proposition \ref{prop: summary of Rn and Sn}. By construction, $U_l$ is a stationary ergodic sequence under $\mathbb{P}_\xi$ so that $U_1,U_3,U_5,\ldots$ and $U_0,U_2,U_4\ldots$ are two sequences of i.i.d. random variables.
Also set 
\begin{align}
N_T = \sum_{m \ge 1} \mathbf{1}\{S_m \le T\} \label{N_T}\,,
\end{align} 
which is an integrable additive functional of $Z_t$.
Decompose 
 \begin{align*}
 \frac{{ I(x)}}{\Upsilon(T)} 
 =\,&
 \underbrace{
 \frac{{ I(x)}}{\Upsilon(T)} -  \mathbb{E}^{\mathcal{M}}_\lambda\left(\frac{{ I(x)}}{\Upsilon(T)} \right)}_{\mathcal{A}} 
+\underbrace{\mathbb{E}^{\mathcal{M}}_\lambda\left(\frac{{ I(x)}}{\Upsilon(T)} \right)  - \frac{N_T}{\Upsilon(T)} \mathbb{E}^{\mathcal{M}}_\lambda\left( U_{1}(x) \right)}_{\mathcal{B}}  
 +\underbrace{\frac{N_T}{\Upsilon(T)} \mathbb{E}^{\mathcal{M}}_\lambda\left( U_{1}(x) \right)}_{\mathcal{C}}\,. 
\end{align*}
Note that $\mathcal{A}$ is replacing $\frac{{ I(x)}}{\Upsilon(T)}$ by its deterministic equivalence, $\mathcal{B}$ is controlling the error of replacing $\mathbb{E}^{\mathcal{M}}_\lambda\left(\frac{{ I(x)}}{\Upsilon(T)} \right)$ by its generalized life-cycle decomposition, and $\mathcal{C}$ is the main term to quantify; i.e., to finish the proof, we show that asymptotically $\mathcal{A}$ and $\mathcal{B}$ are negligible and $\mathcal{C}$ is the dominant term. 

\begin{claim}
\label{mathcal{A}-conv} The quantity $\mathcal{A} = O_p\left( \frac{1}{\Upsilon(T)^{1/2}}\right)$. 
\end{claim} 

\begin{proof}[Proof of Claim \ref{mathcal{A}-conv}]
Since $X_t$ is strong Feller by Assumption \ref{manifold-ass2}, $K\left( \frac{\mathcal{D}_x( \cdot ) }{h}\right) \mathcal{H}( \cdot)$ is a special function of positive measure (in the sense of \cite{LOCHERBACH20081301}), {  by} \eqref{expectations-agree}  we have
\begin{align*}
\mathbb{E}^{\mathcal{M}}_\lambda\left( \frac{{ I(x)}}{\Upsilon(T)}\right) = \mathbb{E}^M_\lambda\left( \frac{{ I(x)}}{\Upsilon(T)}\right) 
\end{align*}
since the first coordinate of $Z_t$ follows the same law of $X_t$ (see Proposition \ref{LL-Prop2.8a}).
By Chebyshev's inequality, 
\begin{align*}
\mathbb{P}^M_\lambda(|\mathcal{A}| \ge \zeta) &\le \frac{1}{\zeta^2}  \mathbb{E}^M_\lambda \left(\left[\frac{{ I(x)}}{\Upsilon(T)} -  \mathbb{E}^M_\lambda\left( \frac{{ I(x)}}{\Upsilon(T)}\right)  \right]^2\right)\,.
\end{align*}
Recall $\tilde{g}_{s,t}(a,b) = \tilde{\phi}_{s,t}(a,b) - \tilde{\phi}_s(a)\tilde{\phi}_t(b)$ in \eqref{definition of tilde g st ab}.
Then, rewrite the above quantity in terms of quantifying path-dependence, i.e.,
\begin{align}
&\frac{1}{\zeta^2}  \mathbb{E}^M_\lambda\left(\left[\frac{{ I(x)}}{\Upsilon(T)} -  \mathbb{E}^M_\lambda\left( \frac{{ I(x)}}{\Upsilon(T)}\right)  \right]^2\right) \nonumber \\
 =\,  & \frac{1}{\zeta^2 \Upsilon^2(T)} \iint_{[0,T]^2} \text{Cov}\left( K_{h}(X_s) \mathcal{H}(X_s), \, K_{h}(X_t)  \mathcal{H}(X_t)\right) dsdt \nonumber\\
   =\,  & \frac{1}{\zeta^2 \Upsilon^2(T)} \iint_{[0,T]^2} \Bigg( \iint_{M \times M}   K_{h}(a)   \mathcal{H}(a)   K_{h}(b)  \mathcal{H}(b) \tilde{g}_{s,t}(a,b) \, dV_{g \otimes g}(a,b) \Bigg) dsdt  \,,\nonumber
   \end{align}
where $dV_{g \otimes g}(a,b)$ is the Riemannian volume form of the product metric $g \otimes g$ on the product manifold $M \times M$, and the last equality comes from the fact that 
\[
 \mathbb{E}^M_\lambda\left( K_{h}(X_s)   \mathcal{H}(X_s)\right)=\int_M  K_{h}(a) \mathcal{H}(a)\tilde{\phi}_s(a)dV_g(a)\,.
\]
This formula allows us to produce an ergodic-like result as $T \rightarrow \infty$ following our assumptions on $\tilde{g}_{s,t}$. By assumption, the joint and marginal densities are continuous and bounded on $M$ \cite{Stroock2010pde}, and hence on $M \times M$. Thus, $\tilde{g}_{s,t}(a,b) \in L^\infty(M \times M, dV_{g \otimes g})$. By H\"older's inequality, for $1/q+1/r=1$, where $r$ is from Assumption \ref{reg-ass},
\begin{align*}
 &\iint_{M \times M}   K_{h}(a)    \mathcal{H}(a)   K_{h}(b) \mathcal{H}(b)  \tilde{g}_{s,t}(a,b) \,dV_{g \otimes g}(a,b)  \\
\le\,& \frac{1}{h^{2d(1-1/q)}}  \|\tilde{g}_{s,t}(a,b)\|_r   \left( \iint_{M \times M} (K^q)_{h}(a)  \mathcal{H}^q(a)  (K^q)_{h}(b)  \mathcal{H}^q(b) dadb\right)^{1/q} \,,
\end{align*}
where the functional norms are with respect to the measure $dV_{g \otimes g}$ on $M \times M$. Note that by regularity assumption and approximation of identity,
\begin{align*}
\sup_{h > 0}\left|  \iint_{M \times M} (K^q)_{h}(a)  \mathcal{H}^q(a)  (K^q)_{h}(b)  \mathcal{H}^q(b) dadb \right | \le C_1
\end{align*}
for some constant $C_1>0$.  Indeed, by Fubini,
\begin{align*}
\| (K^q)_{h}(a)  \mathcal{H}^q(a)  (K^q)_{h}(b)  \mathcal{H}^q(b)\|_{L^1(M\times M)}
&=\Big(\int_M (K^q)_{h}(a)\mathcal H^q(a)\,dV_g(a)\Big)^2\\
&\le \|\mathcal H^q\|_{L^\infty(M)}^2\Big(\int_M (K^q)_{h}(a)\,dV_g(a)\Big)^2.
\end{align*}
By the approximation-of-identity estimate, there exists a geometric constant $C_{\mathrm{geo}}>0$ such that
\[
\sup_{h>0}\int_M (K^q)_{h}(a)\,dV_g(a)
\le C_{\mathrm{geo}}\int_{\mathbb R^d}K^q(\|u\|)\,du\,,
\]
and one may take
$C_1=\left(\|\mathcal H^q\|_{L^\infty(M)}C_{\mathrm{geo}}\int_{\mathbb R^d}K^q(\|u\|)\,du\right)^2$.
Hence, 
\begin{align*}
&\iint_{M \times M}   K_{h}(a)   \mathcal{H}(a)   K_{h}(b) \mathcal{H}(b)  \tilde{g}_{s,t}(a,b) \, dV_{g \otimes g}  
\le C^{1/q}_1   \|\tilde{g}_{s,t}(a,b)\|_r \,.
\end{align*}
 Finally, we have 
\begin{align*}
\mathbb{P}^M_\lambda(|\mathcal{A}|  \Upsilon(T)^{1/2} \ge \zeta) & \le  \frac{C^{1/q}_1}{\zeta^2}   \left( \frac{1}{\Upsilon(T)h^{2d(1-1/q)}} \iint_{[0,T]^2} \|\tilde{g}_{s,t}(a,b)\|_r \, ds dt \right) \,,
\end{align*}
which goes to 0 as $T\to \infty$ by Assumptions \ref{reg-ass} and \ref{reg-ass2}. We thus conclude $\mathcal{A} = O_p\left( \frac{1}{\Upsilon(T)^{1/2}}\right)$.

\end{proof}

 \begin{claim}
\label{mathcal{B}-conv} The quantity $\mathcal{B} = O_p\left( \frac{1}{\Upsilon(T)^{1/2}}\right)$.
\end{claim} 
\begin{proof}[Proof of Claim \ref{mathcal{B}-conv}]
Recall that by construction, $S_0=R_0=0$ and  $S_{N_T}\leq T < S_{{N_T}+1} < R_{{N_T}+1}$, which leads to 
\begin{align*}
{ I(x)} = \sum_{m=0}^{{N_T}} U_m -  \int_T^{R_{N_T+1}}   K_{h}(X_s)   \mathcal{H}(X_s) \, ds \,.
\end{align*}
Thus, the quantity $\mathcal{B}$ can be decomposed into an initial term, a bulk sum, and a terminal term:
\begin{align*}
 \mathcal{B} =&\, \frac{1}{\Upsilon(T)}\mathbb{E}^{\mathcal{M}}_\lambda
 \left(\sum_{m=0}^{N_T} U_m -  \int_T^{R_{N_T+1}}  K_{h}(X_s)  \mathcal{H}(X_s)ds \right)  - \frac{N_T}{\Upsilon(T)} \mathbb{E}^{\mathcal{M}}_\lambda\left( U_{1}(x) \right)  \\   
 =&\, \underbrace{\mathbb{E}^{\mathcal{M}}_\lambda\left(\frac{U_0(x)}{\Upsilon(T)}\right)}_{\mathcal{B}_1} 
+ \underbrace{\mathbb{E}^{\mathcal{M}}_\lambda
 \left(\frac{1}{\Upsilon(T)} \left(\sum_{m\ge 1} \mathbf{1}\{S_m \le T\} U_m \right)\right)  - \frac{N_T}{\Upsilon(T)} \mathbb{E}^{\mathcal{M}}_\lambda\left( U_{1}(x) \right)}_{\mathcal{B}_2}  \\   
 &- \underbrace{\mathbb{E}^{\mathcal{M}}_\lambda
 \left(\frac{1}{\Upsilon(T)}  \int_T^{R_{N_T+1}}  K_{h}(X_s) \mathcal{H}(X_s)ds \right)}_{\mathcal{B}_3}\,,
\end{align*}
where we rewrite the term $\sum_{m=1}^{N_T} U_m$ as $\sum_{m \ge1} \mathbf{1}\{S_m \le T\}U_m$ in $\mathcal{B}_2$ and $\mathbf{1}$ is the indicator function. To see this, since $S_l < T$ for all $1\leq l\leq N_T$, we have $U_l =  \mathbf{1}\{S_l < T\}U_l$. 
On the other hand, since $T < S_{N_T+1} < R_{N_T+1}$ by construction, the term 
\[
\mathbf{1}\{S_{N_T+1} < T\} U_{N_T+1}=\mathbf{1}\{S_{N_T+k} < T\} U_{N_T+k} =0
\] 
for all $k>1$. Note that $\mathcal{B}_1$ and $\mathcal{B}_3$ are boundary terms, while $\mathcal{B}_2$ is the dominant term. Also note that the distribution of $Z_0$ is in general different from that of $Z_{R_n}$, for $n\geq 1$, unless we choose $\lambda=\xi$ \cite[Proposition 4.4]{LOCHERBACH20092312}, so $\mathcal{B}_1$ is separated from $\mathcal{B}_2$.

\begin{remark}
Rewriting $\sum_{m=1}^{N_T} U_m$ as $\sum_{m \ge1} \mathbf{1}\{S_m \le T\}U_m$ may at first appear superficial, but it is in fact necessary. Recall that the generalized life-cycle decomposition of \cite{LOCHERBACH20081301} produces a dual sequence of stopping times $\{S_m, R_m\}_{m=0}^\infty$ satisfying $S_m < R_m<S_{m+1}<\ldots$. The construction guarantees that $X_{R_m +s}$, for $s> 0$, is independent of the sigma algebra $\sigma\{X_t|\, t \le S_m\}$ but not $\sigma\{X_t|\, t \le R_m\}$. Consequently, the decision whether to include $U_m$ in the sum  must be made at time $S_m<R_m$. 
\end{remark}

We start with $\mathcal{B}_2$. Since $U_m$ and $\mathbf{1}\{S_m < T\}$ are independent, the expectation becomes 
 \begin{align*}
   \mathbb{E}^{\mathcal{M}}_\lambda\left(\mathbf{1}\{S_m \le T\}U_m(x)\right)   &= \mathbb{E}^{\mathcal{M}}_\lambda\left(U_m(x)\right)  \mathbb{E}^{\mathcal{M}}_\lambda\left(\mathbf{1}\{S_m \le T\}\right) \,.
   \end{align*}
Recall that the sequence $U_m$, for $m\geq 0$, is a stationary ergodic sequence under $\mathbb{P}_\xi$ \cite[Proposition 4.4]{LOCHERBACH20092312}  arising from the generalized life-cycle decomposition. By Proposition \ref{2.20and4.3}, with the general initial measure $\lambda$, we have $\mathbb{E}^{\mathcal{M}}_\lambda\left[U_1(x)\right] =  \cdots = \mathbb{E}^{\mathcal{M}}_\lambda\left[U_{N_T}(x)\right]$, and hence
   \begin{align*}
  \mathcal{B}_2 &= \mathbb{E}^{\mathcal{M}}_\lambda
 \left(\frac{1}{\Upsilon(T)} \Big(\sum_{m\ge1} \mathbf{1}\{S_m \le T\} U_m(x) \Big)\right)  - \frac{N_T}{\Upsilon(T)} \mathbb{E}^{\mathcal{M}}_\lambda( U_{1}(x))  \\   
 &=\Upsilon(T)^{-1}\left(\mathbb{E}^{\mathcal{M}}_\lambda(U_1(x))\mathbb{E}^{\mathcal{M}}_\lambda\Big(\sum_{m\ge1} \mathbf{1}\{S_m \le T\}\Big) - N_T\mathbb{E}^{\mathcal{M}}_\lambda( U_{1}(x) )\right) \\  
 &= \mathbb{E}^{\mathcal{M}}_\lambda(U_1(x))\Upsilon(T)^{-1}( \mathbb{E}^{\mathcal{M}}_\lambda(N_T) - N_T) \,.
   \end{align*}
Since $\mathbb{E}^{\mathcal{M}}_\lambda ( U_1(x) ) = \mathbb{E}^{M}_\lambda ( U_1(x) ) $ by \eqref{expectations-agree},  $\mathbb{E}^{\mathcal{M}}_\lambda(U_1(x))$ is analyzed by Lemma \ref{lemma: expectation expansion}. 
With Lemma \ref{Premathcal{B}-conv} for $\mathbb{E}^{\mathcal{M}}_\lambda(N_T) - N_T$, 
it follows that $\mathcal{B}_2 = O_p\left( \Upsilon(T)^{-1/2}\right)$. 

Next we control $\mathcal{B}_1$. By Assumption \ref{manifold-ass2}, the density function of $\lambda$ is bounded, so we write $\lambda(dx)=p_\lambda(x)dV_g(x)$, where $p_\lambda(x)$ is bounded on $\mathsf C$. On the other hand, $\phi_X(x)=p_X(x)dV_g(x)$, where $p_X$ is strictly positive by assumption. Since $\lambda$ is compactly supported on $\mathsf C$, the same as the support of $\xi$ by construction, we have by Radon-Nikodym that
\[
\lambda(dx)=\frac{p_\lambda(x)}{p_X(x)}p_X(x)dV_g(x)\,.
\]
Since $K$ and $\mathcal{H}$ are non-negative, we have 
\begin{align*}
\mathbb{E}^{\mathcal{M}}_\lambda(U_0(x))&\,=\int_{(x,u,x')\in\mathcal{M}} U_0(x) \lambda(dx)\otimes U(du)\otimes Q((x,u),dx')\\
&\,\leq \left(\sup_{x\in \mathsf C}\frac{p_\lambda(x)}{p_X(x)}\right)\int_{\mathcal{M}} U_0(x) \xi(dx)\otimes U(du)\otimes Q((x,u),dx')\\
&\,=\left(\sup_{x\in \mathsf C}\frac{p_\lambda(x)}{p_X(x)}\right)\mathbb{E}^{\mathcal{M}}_\xi(U_0(x))
\end{align*}
By assumption, $\sup_{x\in \mathsf C}\frac{p_\lambda(x)}{p_X(x)}$ is bounded, and $\mathbb{E}^{\mathcal{M}}_\xi(U_0(x))$ is bounded by Proposition \ref{2.20and4.3}, so $\mathbb{E}^{\mathcal{M}}_\lambda(U_0(x))$ is bounded.
Since 
\[
\Upsilon(T) \mathcal{B}_1 = \mathbb{E}^{\mathcal{M}}_\lambda(U_0(x))\,, 
\]
we conclude that $\mathcal{B}_1 = O_p( \Upsilon(T)^{-1})$. Finally, since $ K$ and $\mathcal{H}$ are non-negative functions,
\begin{align*}
|\mathcal{B}_3|  &= \mathbb{E}^{\mathcal{M}}_\lambda
 \left(\frac{1}{\Upsilon(T)}  \int_T^{R_{N_T+1}}  K_{h}(X_s)  \mathcal{H}(X_s)ds \right)  \le \mathbb{E}^{\mathcal{M}}_\lambda
 \left(\frac{1}{\Upsilon(T)}  \int_{R_{N_T}}^{R_{N_T+1}}  K_{h}(X_s) \mathcal{H}(X_s)ds \right)  \\  
 &= \Upsilon(T)^{-1}\mathbb{E}^{\mathcal{M}}_\lambda(U_{R_{N_T}}(x)) = \Upsilon(T)^{-1}\mathbb{E}^{\mathcal{M}}_\lambda(U_{1}(x)) \,,
\end{align*}
where the last equality holds by Proposition \ref{2.20and4.3}.
Since $\mathbb{E}^{\mathcal{M}}_\lambda (U_1(x) )= \mathbb{E}^M_\lambda (U_1(x))$ can be controlled by Lemma \ref{lemma: expectation expansion}, we have $\mathcal{B}_3 = O_p(\Upsilon(T)^{-1})$. In the aggregate, we obtain the claim that
$\mathcal{B} = O_p(\Upsilon(T)^{-1/2})$.
\end{proof}

Finally, we control $\mathcal{C}$. The quantities $\mathbb{E}^{\mathcal{M}}_\lambda \left(U_1(x) \right)$ and $\frac{N_T}{\Upsilon(T)}$ can be controlled by Lemmas \ref{lemma: expectation expansion} and \ref{Premathcal{B}-conv}. By Proposition \ref{MLprocess_and_distr}, $g_\alpha(1) \sim G_\alpha$. 
As a result, since $\mathcal{D}_x(x')=\|\iota(x)-\iota(x')\|$, we have { $\displaystyle
   \mathcal{C} \xrightarrow[]{\hspace{0.1cm}d\hspace{0.1cm}}  G_\alpha p_X(x)  \mathcal{H}(x) \int_{\mathbb{R}^d}  K(  \| u \|) du$} 
as $h \rightarrow 0$ and $T \rightarrow \infty$.  
   
We can now finalize the proof of Lemma \ref{A1-analog} by assembling the above claims. By Claims \ref{mathcal{A}-conv}  and \ref{mathcal{B}-conv}, we have controls 
$\mathcal{A} = O_p\left(\frac{1}{\Upsilon(T)^{1/2}} \right)$ and $\mathcal{B} = O_p\left( \frac{1}{\Upsilon(T)^{1/2}}\right)$ respectively. With the control of $\mathcal{C}$, we finish the proof.
\end{proof}

\begin{remark}
Using our notation, the initial measure of $X_t$ in \cite[Lemma A.1]{bandi_moloche_2018} is implicitly assumed to be $\xi$ (see the proof on page 924, line 19 \cite{bandi_moloche_2018}). We relax this assumption by assuming the distribution of $X_0$ is compactly supported on a nontrivial compact subset with a bounded density function. 
\end{remark}

Next, we need to prepare some controls of the kernel function when derivatives of ${\mathcal{D}_x(\cdot)}$ are involved. As a motivation, consider $M=\mathbb{R}^2$ with $K(s) = e^{-s^2}$ and $\mathcal{D}_0(u) = \sqrt{u_1^2 + u_2^2}$ for $u \in \mathbb{R}^2$. Then, even though $\frac{\partial^2}{\partial u_1 \partial u_2} \mathcal{D}_0(u) = \frac{-u_1u_2}{(u_1^2 + u_2^2)^{3/2}}$ is unbounded near ${0}$, the product function
$K'\left(\frac{\mathcal{D}_0(u)}{h} \right)\frac{\partial^2 \mathcal{D}_0}{\partial u_1 \partial u_2} (u) = \frac{-1}{h} e^{-(u_1^2+u_2^2)/2} \frac{-u_1u_2}{u_1^2+u_2^2}$ is bounded near $0$. This holds due to the fact that $\frac{\partial^2\mathcal{D}_0}{\partial u_1 \partial u_2}(u) = O(\|u\|^{-1})$ at zero, while $K'(\|u\|)$ is approximately linear, with no constant term, near $\|u\| = 0$. Since such a term involving derivatives appears frequently in our analysis, we verify in Lemma \ref{A2-Analog-boundedness-computation-v2} that this property present in the Euclidean distance setting generalizes to the manifold setup.

\begin{lemma}\label{Dx-ass} 
Suppose the manifold model in Definition \ref{definition observed space model} hold. For each $x\in N$, when $r>0$ is sufficiently small, the designed $\mathcal{D}_x: B_x(r) \rightarrow [0,\infty)$ satisfies $\mathcal{D}_x(x') = \tilde{\mathcal{D}}_x(d_g(x,x'), \theta_{x'})$, where $x'=\exp_y(d_g(x,x') \theta_{x'})\in B_x(r)$, for some nonnegative function $\tilde{\mathcal{D}}_x$ defined on $[0,r) \times \mathbb{S}^{d-1}$ with $\tilde{\mathcal{D}}_x(0, \theta) = 0$ for all $\theta \in \mathbb{S}^{d-1}$. If we set $\chi_x: \mathbb{R}^{d}\setminus\{\mathbf{0}\} \rightarrow \mathbb{R}$ by
$\chi_x(u) = \frac{\partial \tilde{\mathcal{D}}_x}{\partial \rho}\left(0, \frac{u}{\|u\|}\right)$, $\chi_x$ is positive and symmetric, i.e., $\chi_x(u) = \chi_x(-u)$ for all $u \in \mathbb{R}^{d}\setminus\{\mathbf{0}\}$. Moreover, with local coordinates $(u_1, \cdots, u_d)$ around $x$, $\frac{\partial^2}{\partial u_i \partial u_j}\mathcal{D}_x(x') = O\left(\frac{1}{\mathcal{D}_x(x')}\right)$ for all $i,j = 1, \cdots, d$. 
\end{lemma}
\begin{proof}
When $\mathcal{D}_x(x') = \|\iota(x) - \iota(x')\|_{\mathbb{R}^p}$ is sufficiently small, we have \cite[Lemma SI.3]{wuwu18} $\|\iota(y) - \iota(x')\|_{\mathbb{R}^p} = d_g(x,x')+ O(d_g(x,x')^3)$.
Therefore, $\tilde{\mathcal D}_x(\rho, \theta) = \rho + O(\rho^3)$, where the implied constant depends on $\theta$, and hence $\chi_x(u) = 1$.
Now, consider that in normal coordinates, then for $x'$ expressed as $ u'= (u'_1, \cdots, u'_d)$, then $\mathcal{D}_x(x') = \|u'\| + O(\|u'\|^3)$. Differentiating twice with respect to coordinates $u_i$ and $u_j$, we see that $\frac{\partial^2}{\partial u_i \partial u_j}\mathcal{D}_x(x') = O\left(\frac{1}{\mathcal{D}_x(x')}\right)$ for all $i,j = 1, \cdots, d$.
\end{proof}

\begin{lemma}\label{A2-Analog-boundedness-computation-v2}
 Let ${K} \in C^3(\mathbb{R})$ be non-negative, compactly supported, and symmetric. Let $x \in N$ and let $\mathcal{D}_x$ be defined as in Lemma \ref{Dx-ass}. Then, $\nabla \mathcal{D}_x(\cdot)$ is bounded on $B_c(x) \setminus\{x\}$, for some $c > 0$.   
{ Then} for all $1 \le i,j \le d$, there exists a bounded function $\tilde{ K}$ satisfying the requirements for the kernel used in Lemma \ref{A1-analog}  and depending on ${K}''(0)$ so that 
  \begin{align*}
  \left| {K}'\left(\frac{{\mathcal{D}}_x(x') }{h}
\right)  \frac{\partial^2 {\mathcal{D}}_x}{\partial u_i \partial u_j}(x')\right|  &\le { \frac{1}{h}\tilde{K} \left(\frac{\mathcal{D}_x(x')}{h} \right)}
\end{align*} 
for $x'\in B_{c'}(x) \setminus \{x\}$ for some $c' >0$, as $h \rightarrow 0$. 
\end{lemma}

\begin{proof}
    By the Taylor expansion, for $u_i\in T_xM$ and $\|u_i\|=1$, we have that 
    \begin{align*}
\lim_{h \rightarrow 0}  \frac{\mathcal{D}_x(\exp_x(h u_i))}{h} 
    &= \frac{\partial \tilde{\mathcal{D}}_x}{\partial \rho} \left(0, u_i \right)=\chi_x(u_i)\,,
\end{align*}
where $\tilde{\mathcal{D}}_x$ and $\chi_x$ are defined  in  Lemma \ref{Dx-ass}. Since $\chi$ is continuous and $\mathbb{S}^{d-1}$ is compact, then we can conclude that $\left|\chi_x(u_i)\right| \le C$ for all $u_i \in \mathbb{S}^{d-1}$, some $C > 0$. So, $\nabla \mathcal{D}_x$ is bounded on $B_c(x) \setminus\{x\}$ for some $c > 0$. 

For the second claim, note that given our symmetric assumption on ${K}$, by the Taylor expansion we have ${K}(s) = {K}(0) + \frac{{K}''(0)}{2}s^2 + O(s^4)$, 
and hence 
\begin{align*}
    {K}'(s) = {K}''(0)s + O(s^3)\,.
\end{align*}
For $s = \frac{{ \mathcal{D}_x(y)}}{h}$, we obtain 
\begin{align}
{K}'\left(\frac{{ \mathcal{D}_x(y)}}{h} \right) &= {K}''(0)\frac{{ \mathcal{D}_x(y)}}{h} + O\left(\left(\frac{{ \mathcal{D}_x(y)}}{h}\right)^3 \right) \label{K'-Taylor-exp}\,.
\end{align} 
By Lemma \ref{Dx-ass}, there exists $c',C'>0$ so that for all $x' \in B_{c'}(x) \setminus \{x\}$, we have  
\begin{align}
   \left| \frac{\partial^2 \mathcal{D}_x}{\partial u_i \partial u_j}(y) \right| \le \frac{C'}{\mathcal{D}_x(x')} \label{hessian-bound}\,.
\end{align}
Note that $C', c'$ do not depend on $x$ due to the manifold assumption.
It follows from \eqref{K'-Taylor-exp} and \eqref{hessian-bound} that for all $x'\in B_{c'}(x) \setminus \{x\}$, 
\begin{align*}
    \left| {K}'\left(\frac{{\mathcal{D}}_x(x') }{h}
\right)  \frac{\partial^2 {\mathcal{D}}_x}{\partial u_i \partial u_j}(x')\right| \le 
\frac{{ 1}}{h} \left(C'{K}''(0) + O\left(\left(\frac{{ \mathcal{D}_x(x')}}{h}\right)^2 \right) \right) \,,
\end{align*}
Since locally a manifold can be well approximated by a linear affine subspace, we have
$\left|\frac{{ \mathcal{D}_x(x')}}{h}\right|\leq C$, 
where $C>0$ is independent of $x$, when $c'$ { is} sufficiently small. As a result, we have
\begin{align*}
    \left| {K}'\left(\frac{{\mathcal{D}}_x(x') }{h}
\right)  \frac{\partial^2 {\mathcal{D}}_x}{\partial u_i \partial u_j}(x')\right| \le 
\frac{{ 1}}{h} \left(C'{K}''(0) +C^2 \right) \,,
\end{align*}
By the compact support assumption of ${K}$, suppose $\texttt{supp}{K}=[0,L]$. We can construct $\tilde{K} \in C^3(\mathbb{R})$ so that $\tilde{K}(t)=C'\mathcal{K}''(0) +C^2$ on $[0,L]$ and decays to $0$ smoothly on $[L,2L]$. Then, $\tilde{K}$ satisfies the requirements for the kernel used in Lemma \ref{A1-analog} so that for all $x'\in B_{c'}(x) \setminus \{x\}$,
\begin{align*}
     \left| K'\left(\frac{{\mathcal{D}}_x(x') }{h}
\right)  \frac{\partial^2 {\mathcal{D}}_x}{\partial u_i \partial u_j}(x')\right| &\le \frac{1}{h} \tilde{K}\left(\frac{{ \mathcal{D}_x(x')}}{h} \right)\,.
\end{align*}
\end{proof}

Before analyzing the discretization of additive functionals related to the estimators of interest, we establish another technical lemma for the stochastic integral.

\begin{lemma}\label{stoch-prob-order-computation-lemma}
Let ${f} \in C^3(\mathbb{R})$ be non-negative and compactly supported on $[0,L] \subset [0,\infty)$. Grant Assumptions \ref{manifold-ass}, \ref{manifold-ass2}, { \ref{reg-ass}, and \ref{lebesgue-dens-ass}. Take a $q \in C^3(M)$.} Then, we have
\begin{align*}
\frac{1}{h^{d/2}} \sum_{k=1}^n \int_{(k-1)\Delta}^{k \Delta} \left(\int_{(k-1)\Delta}^{s} \sqrt{f\left(\frac{\mathcal{D}_x(X_r) }{h}\right)}q(X_r) \, dW_r \right) ds = O_p( \sqrt{\Upsilon(T)}\Delta).
\end{align*}

\end{lemma}

\begin{proof}
Denote $I_0 := 0$, 
\begin{align*}
I_{m\Delta} &\, := \frac{1}{h^{d/2}} \sum_{k=1}^m \int_{(k-1)\Delta}^{k \Delta} \left( \int_{(k-1)\Delta}^{s} \sqrt{f\left(\frac{\mathcal{D}_x(X_r) }{h}\right)}q(X_r) \, dW_r\right) ds\\
&\,=\frac{1}{h^{d/2}} \sum_{k=1}^m \int_{(k-1)\Delta}^{k \Delta} (k \Delta -s) \sqrt{f\left(\frac{\mathcal{D}_x(X_r) }{h}\right)}q(X_r) \, dW_s\,,
\end{align*}
for $m = 1, \cdots, n$, where the second inequality comes from integrating by parts using Lemma \ref{integration-by-parts}, and
for $t \in ((m-1)\Delta, m\Delta]$ for some $m$, 
 \begin{align*}
 I_t := I_{(m-1)\Delta} + \frac{1}{h^{d/2}}\int_{(m-1)\Delta}^t (m\Delta - s) \sqrt{f\left(\frac{\mathcal{D}_x(X_r) }{h}\right)}q(X_r) \, dW_s  \,.
 \end{align*}
 By construction, $I_t$ is a continuous martingale since it is the stochastic integral of the predictable process $\sqrt{f\left(\frac{\mathcal{D}_x(X_r) }{h}\right)}q(X_r)$ with respect to the Brownian motion $W_s$. When $t \in ((m-1)\Delta, m\Delta]$ for some $m$, the quadratic variation of $I_t$ is controlled by 
   \begin{align*}
   [I_t] &= [I_{(m-1)\Delta}] + \int_{(m-1)\Delta}^t \frac{1}{h^{d}} (m\Delta - s)^2 f\left(\frac{\mathcal{D}_x(X_r) }{h}\right)q^2(X_r) \, ds \\
   &\leq [I_{(m-1)\Delta}] + \int_{(m-1)\Delta}^t \frac{\Delta^2}{h^{d}}  f\left(\frac{\mathcal{D}_x(X_r) }{h}\right)q^2(X_r) \, ds\,.
   \end{align*}
We remark that while $I_t$ is defined as the sum of several terms, distinct terms have disjoint time support and hence there is cancellation of crossover terms in quadratic variation. So only quadratic variation of each term remains in $[I_t]$. Thus, by iterating the bound, when $T = n\Delta$, we have 
   \begin{align*}
   [I_{T}]
 &\le  \frac{\Delta^2}{h^{d}} \int_0^T f\left(\frac{\mathcal{D}_x(X_r) }{h}\right)q^2(X_r) \, ds \,.
   \end{align*}
Apply Lemma \ref{A1-analog} with $K, \mathcal{H}$ set to $f, q^2$ to conclude that $[I_T] = O_p(\Upsilon(T)\Delta^2)$, and so $I_T = O_p(\sqrt{\Upsilon(T)}\Delta)$.
\end{proof}

With Lemma \ref{stoch-prob-order-computation-lemma}, we now analyze the discretization of an additive functional associated with the estimators of interest. 
This result corresponds to \cite[Lemma A.2]{bandi_moloche_2018}.  We emphasize that this discretization reflects practical data acquisition and differs from the discretization obtained via the generalized life-cycle decomposition in Lemma~\ref{A1-analog}. Also note that the function $\mathcal N$ here is not restricted to non-negative function like that in Lemma \ref{A1-analog}.

\begin{lemma}\label{A2-analog}
Let $K\in C^2(\mathbb{R})$ be non-negative and compactly supported on $[0,L] \subset [0,\infty)$. Let $\mathcal{N} \in C^2(M)$. Suppose Assumptions \ref{manifold-ass}, \ref{manifold-ass2}, { \ref{reg-ass}, \ref{lebesgue-dens-ass}, and \ref{reg-ass2} hold.  }
   Then, when $T=n\Delta$ is sufficiently large, we have 
\begin{align*}
   &\frac{1}{{h^d}\Upsilon(T)} \int_0^T   K\left(\frac{\mathcal{D}_x(X_s) }{h}\right)\mathcal{N}(X_{s})  \, ds  \\
   -&\, \frac{\Delta}{h^d\Upsilon(T)} \sum_{k=0}^{n-1}   K \left(\frac{\mathcal{D}_x(X_{k\Delta}) }{h}\right)\mathcal{N}(X_{k \Delta})
= O_p \left( \frac{\Delta}{h^2}\right) \,.
\end{align*} 
Moreover, 
\begin{align*}
    & \mathbb{E}^M_\lambda\left(\frac{1}{{h^d}\Upsilon(T)} \int_0^T  K\left(\frac{\mathcal{D}_x(X_s) }{h}\right)\mathcal{N}(X_{s})  \, ds \right)  \\
   -&\, \mathbb{E}^M_\lambda\left(\frac{\Delta}{h^d\Upsilon(T)} \sum_{k=0}^{n-1}  K\left(\frac{\mathcal{D}_x(X_{k\Delta}) }{h}\right)\mathcal{N}(X_{k \Delta}) \right)
= O\left(\frac{\Delta }{h^2} \right)  \,.
\end{align*}
\end{lemma}

\begin{proof}  
 To simplify the notation, denote $ K_{h}(\cdot):=\frac{1}{h^{d}}{K}\left(\frac{{\mathcal{D}}_x(\cdot)}{h} \right)$ and $q\in C^3(M)$ as
\begin{align*}
q(x') :=   K_{h}(x') \mathcal{N}(x') \,.
\end{align*}
Since ${K}$ is compactly supported on $[0,L]$, when $h>0$ is sufficiently small, $2Lh<\texttt{inj}_x$ and $q(X_t) \equiv 0$ when $X_t \notin B_{Lh}(x)$.  
By It\^o's formula \eqref{ito-integral}, for $s\in ((k-1)\Delta,k\Delta]$, we have
\begin{align}
 &q(X_s)  -q(X_{(k-1)\Delta})\nonumber\\
 =\,& \int_{(k-1)\Delta}^{s} \Big[(\mu q)(X_v) + \frac{1}{2}  \sum_{\alpha=1}^r \sigma_\alpha(\sigma_\alpha q)(X_v) \Big] \, dv + \int_{(k-1)\Delta}^{s}  (\sigma_\alpha q)(X_v) \, dW^\alpha_v \nonumber\,.
\end{align}
Also recall that with a local coordinate $(u_1, \cdots, u_d)$ on $B_r(x)$ centered at $x$, we have $\nabla^2_{\sigma_\alpha,\sigma_\alpha} q=\sigma_\alpha^k\sigma_\alpha^l(q_{kl}-\Gamma_{kl}^jq_j)$ and
$(\nabla_{\sigma_\alpha}\sigma_\alpha)(q) 
= \big(\sigma_\alpha^k \partial_k \sigma_\alpha^j + \Gamma^j_{kl} \, \sigma_\alpha^k \sigma_\alpha^l \big) q_j
$, 
where $q_i =:\partial_i q =  \frac{\partial q}{\partial u_i}$, and $q_{ij} := \partial_{ij} q = \frac{\partial^2q}{\partial u_i \partial u_j}$. Assume we take the normal coordinate so that Christoffel symbols disappear. We have $\sigma_\alpha(\sigma_\alpha q) = \sigma_\alpha^k \sigma_\alpha^l q_{kl} + \sigma_\alpha^k \partial_k \sigma_\alpha^j q_j$.
Hence,  
\begin{align}
 &\int_{(k-1)\Delta}^{k\Delta} q(X_s)  - q(X_{(k-1)\Delta})  ds \label{Ito-decomposition} \\
 =\,& \int_{(k-1)\Delta}^{k\Delta}  \int_{(k-1)\Delta}^{s} \left(\mu^j  + \frac{1}{2} \sum_{\alpha =1}^r\sigma_\alpha^k \partial_k \sigma_\alpha^j\right)(X_v)q_j(X_v)\, dv  ds\nonumber \\
&+  \int_{(k-1)\Delta}^{k\Delta}\left( \int_{(k-1)\Delta}^{s}(\sigma^j_\alpha q_j)(X_v)\, dW^\alpha_v\right) ds\nonumber \\
&+ \int_{(k-1)\Delta}^{k\Delta}\left( \int_{(k-1)\Delta}^{s} \frac{1}{2} \sum_{\alpha =1}^r ((\sigma_\alpha^k)^2 q_{kk})(X_v) dv \right) ds\nonumber \\
& + \int_{(k-1)\Delta}^{k\Delta}\left( \int_{(k-1)\Delta}^{s} \frac{1}{2} \sum_{\alpha =1}^r  (1 - \delta_{kl})(\sigma_\alpha^k \sigma_\alpha^l q_{kl})(X_v)dv \right) ds\,.\nonumber
\end{align}
We now calculate derivatives.  For $i =1 , \cdots,d$,  
\begin{align*}
q_i(X_v) 
&= \underbrace{\frac{1}{h} (K')_{h}(X_v)   \frac{\partial \mathcal{D}_x}{\partial u_i}(X_v) \mathcal{N}(X_v)}_{:=\mathcal{A}_i^1(v)} 
+\underbrace{ K_{h}(X_v)  \mathcal{N}_i(X_v)}_{:=\mathcal{A}_i^2(v)}\,. 
\end{align*}
Continuing, for $i = 1,  \cdots, d$, $\mathcal{A}_{ii}^j(v) := \frac{\partial \mathcal{A}_i^j(v)}{\partial u_i}$ for $j \in \{1,2\}$ satisfies
\begin{align*}
\mathcal{A}_{ii}^1(v)
 =\,& \frac{1}{h^2} (K'')_{h}(X_v)   \left(\frac{\partial {\mathcal{D}}_x}{\partial u_i}(X_v)\right)^2  \mathcal{N}(X_v) \\
 &+ \frac{1}{h}  (K')_{h}(X_v)  \Big(\frac{\partial^2 {{D}}_x}{\partial u_i^2}(X_v) \mathcal{N}(X_v) + \frac{\partial \mathcal{D}_x}{\partial u_i}(X_v)  \mathcal{N}_i(X_v)\Big) \\
 \mathcal{A}_{ii}^2(v)=\,& \frac{1}{h}  (K')_{h}(X_v)  \frac{\partial {\mathcal{D}}_x}{\partial u_i}(X_v)\mathcal{N}_i(X_v) +  K_{h}(X_v) \mathcal{N}_{ii}(X_v)\,.
\end{align*}
Aggregating terms and labeling, we compute
 \begin{align*}
&q_{ii}(X_v) =  \mathcal{A}_{ii}^1(v)  + \mathcal{A}_{ii}^2(v)\\
=\,&  
    \underbrace{\frac{1}{h^2} (K'')_{h}(X_v)  \left(\frac{\partial {\mathcal{D}}_x}{\partial u_i}(X_v)\right)^2  \mathcal{N}(X_v)+ 
{\frac{1}{h}  (K')_{h}(X_v) \frac{\partial^2 {\mathcal{D}}_x}{\partial u_i^2}(X_v) \mathcal{N}(X_v)}}_{:=\mathcal{B}_i^1(v)}  \\
 &+ %
 \underbrace{\frac{2}{h}   (K')_{h}(X_v) \frac{\partial {\mathcal{D}}_x}{\partial u_i}(X_v)\ \mathcal{N}_i(X_v)}_{:=\mathcal{B}_i^2(v)}  + 
\underbrace{ K_{h}(X_v)\mathcal{N}_{ii}(X_v)}_{:=\mathcal{B}_i^3(v)} \,.  
\end{align*}
Similarly, for $1 \le i \neq j \le d$, we compute with $ \mathcal{A}_{ij}^k(v) = \frac{\partial \mathcal{A}_j^k(v)}{\partial u_i}  $ for $k \in \{1,2\}$ that
\begin{align*}
  \mathcal{A}_{ij}^1(v)
 =\,&  \frac{1}{h^2}   (K'')_{h}(X_v)  \frac{\partial {\mathcal{D}}_x}{\partial u_j}(X_v)\frac{\partial {\mathcal{D}}_x}{\partial u_i}(X_v)  \mathcal{N}(X_v) \\
&+  \frac{1}{h} (K')_{h}(X_v) \Big( \frac{\partial^2 {\mathcal{D}}_x}{\partial u_i \partial u_j}(X_v)  \mathcal{N}(X_v) +\frac{\partial {\mathcal{D}}_x}{\partial u_j}(X_v)  \mathcal{N}_i(X_v)\Big) \\
\mathcal{A}_{ij}^2(v) =\,& \frac{1}{h}   (K')_{h}(X_v)  \frac{\partial {\mathcal{D}}_x}{\partial u_i}(X_v)    \mathcal{N}_j(X_v) +    K_{h}(X_v) \mathcal{N}_{ij}(X_v)\,.
\end{align*}
Again aggregating, we obtain for $1 \le i \neq j \le d$ that
\begin{align*}
&q_{ij}(X_v)= \mathcal{A}_{ij}^1(v) + \mathcal{A}_{ij}^2(v)\\
=\,&\underbrace{
\frac{1}{h^2} (K'')_{h}(X_v) \frac{\partial {\mathcal{D}}_x}{\partial u_j}(X_v)\frac{\partial {\mathcal{D}}_x}{\partial u_i}(X_v)  \mathcal{N}(X_v) +  \frac{1}{h} (K')_{h}(X_v) \frac{\partial^2 {\mathcal{D}}_x}{\partial u_i \partial u_j}(X_v)  \mathcal{N}(X_v) }_{:=\mathcal{C}_{ij}^1(v)} \\
&+
\underbrace{ \frac{1}{h} (K')_{h}(X_v)  \left(\frac{\partial {\mathcal{D}}_x}{\partial u_j}(X_v)  \mathcal{N}_i(X_v) +    \frac{\partial {\mathcal{D}}_x}{\partial u_i}(X_v)\mathcal{N}_j(X_v)\right) }_{:=\mathcal{C}_{ij}^2(v)}+   
   \underbrace{ K_{h}(X_v) \mathcal{N}_{ij}(X_v)}_{:=\mathcal{C}_{ij}^3(v)}\,. 
\end{align*}
By inserting these terms into \eqref{Ito-decomposition}, we have
\begin{align*}
     &\frac{1}{\Upsilon(T)} \sum_{k=0}^{n-1} \int_{k\Delta}^{(k+1)\Delta}[q(X_s)  - q(X_{(k-1)\Delta})] ds \\
 =\,& \frac{1}{\Upsilon(T)} \sum_{k=1}^n\int_{(k-1)\Delta}^{k\Delta}  \int_{(k-1)\Delta}^{s} \left(\mu^j  + \frac{1}{2} \sum_{\alpha =1}^r\sigma_\alpha^k \partial_k \sigma_\alpha^j\right)(X_v)(\mathcal{A}_j^1(v) + \mathcal{A}_j^2(v))\, dv  ds \\
&+ \frac{1}{\Upsilon(T)} \sum_{k=1}^n\int_{(k-1)\Delta}^{k\Delta} \int_{(k-1)\Delta}^{s}\sigma^j_\alpha(X_v) (\mathcal{A}_j^1(v) + \mathcal{A}_j^2(v))\, dW^\alpha_v ds \\
&+ \frac{1}{\Upsilon(T)} \sum_{k=1}^n\int_{(k-1)\Delta}^{k\Delta} \int_{(k-1)\Delta}^{s} \frac{1}{2} \sum_{\alpha =1}^r \sigma_\alpha^k(X_v)^2 (\mathcal{B}_k^1(v) +  \mathcal{B}_k^2(v)+ \mathcal{B}_k^3(v) )  dv  ds \\
& + \frac{1}{\Upsilon(T)} \sum_{k=1}^n\int_{(k-1)\Delta}^{k\Delta} \int_{(k-1)\Delta}^{s} \frac{1}{2} \sum_{\alpha =1}^r  (1 - \delta_{kl})(\sigma_\alpha^k \sigma_\alpha^l)(X_v) (\mathcal{C}_{kl}^1(v) + \mathcal{C}_{kl}^2(v)+ \mathcal{C}_{kl}^3(v))dv  ds\,,
\end{align*}
and we label terms in this expansion as
\begin{align*}
\mathcal{I}_1 &:= \sum_{k=1}^n\int_{(k-1)\Delta}^{k\Delta}  \int_{(k-1)\Delta}^{s} \left(\mu^j  + \frac{1}{2} \sum_{\alpha =1}^r\sigma_\alpha^k \partial_k \sigma_\alpha^j\right)(X_v)\mathcal{A}_j^1(v) \, dv  ds \\
\mathcal{I}_2 &:= \sum_{k=1}^n\int_{(k-1)\Delta}^{k\Delta} \int_{(k-1)\Delta}^{s} \frac{1}{2} \sum_{\alpha =1}^r \sigma_\alpha^k(X_v)^2 \mathcal{B}_k^1(v)  dv  ds \\
\mathcal{I}_3 &:=  \sum_{k=1}^n\int_{(k-1)\Delta}^{k\Delta} \int_{(k-1)\Delta}^{s} \frac{1}{2} \sum_{\alpha =1}^r  (1 - \delta_{kl})(\sigma_\alpha^k \sigma_\alpha^l)(X_v) \mathcal{C}_{kl}^1(v) dv  ds \\
\mathcal{I}_4 &:=  \sum_{k=1}^n\int_{(k-1)\Delta}^{k\Delta} \int_{(k-1)\Delta}^{s}\sigma^j_\alpha(X_v) \mathcal{A}_j^1(v) \, dW^\alpha_v ds \,.
\end{align*}
\begin{align*}
\mathcal{J}_1 &:=  \sum_{k=1}^n\int_{(k-1)\Delta}^{k\Delta}  \int_{(k-1)\Delta}^{s} \left(\mu^j  + \frac{1}{2} \sum_{\alpha =1}^r\sigma_\alpha^k \partial_k \sigma_\alpha^j\right)(X_v)\mathcal{A}_j^2(v)\, dv  ds \\
\mathcal{J}_2 &:=  \sum_{k=1}^n\int_{(k-1)\Delta}^{k\Delta} \int_{(k-1)\Delta}^{s} \frac{1}{2} \sum_{\alpha =1}^r \sigma_\alpha^k(X_v)^2 \mathcal{B}_k^3(v)  dv  ds \\
\mathcal{J}_3 &:=  \sum_{k=1}^n\int_{(k-1)\Delta}^{k\Delta} \int_{(k-1)\Delta}^{s} \frac{1}{2} \sum_{\alpha =1}^r  (1 - \delta_{kl})(\sigma_\alpha^k \sigma_\alpha^l)(X_v) \mathcal{C}_{kl}^3(v) dv  ds \\
\mathcal{J}_4 &:= \sum_{k=1}^n\int_{(k-1)\Delta}^{k\Delta} \int_{(k-1)\Delta}^{s}\sigma^j_\alpha(X_v) \mathcal{A}_j^2(v) \, dW^\alpha_v ds \,.
\end{align*}
\begin{align*}
\mathcal{L}:=\,&  \sum_{k=1}^n\int_{(k-1)\Delta}^{k\Delta} \int_{(k-1)\Delta}^{s} \frac{1}{2}  \sum_{\alpha =1}^r \left( \sigma_\alpha^k(X_v)^2 \mathcal{B}_k^2(v) + (1 - \delta_{kl})(\sigma_\alpha^k \sigma_\alpha^l)(X_v) \mathcal{C}_{kl}^2(v)\right) dv  ds \,.
\end{align*}
As a result, the lemma statement can be reformulated to 
\begin{align*}
  \int_{(k-1)\Delta}^{k\Delta} q(X_s)  - q(X_{(k-1)\Delta}) ds
  = \mathcal{L} + \sum_{\ell=1}^4 (\mathcal{I}_\ell + \mathcal{J}_\ell )\,.
\end{align*}
Next, we analyze and control the asymptotic for each $\mathcal{I}_\ell, \mathcal{J}_\ell,$ and $\mathcal{L}$, beginning with $\mathcal{I}_1$. Since ${K}$ is compactly supported, we only focus on the case that $X_v$ is sufficiently close to $x$. We have
\begin{align*}
|\mathcal{I}_1| &\le\Delta \sum_{k=1}^n\int_{(k-1)\Delta}^{k\Delta} \Big| \Big(\mu^j  + \frac{1}{2} \sum_{\alpha =1}^r\sigma_\alpha^k \partial_k \sigma_\alpha^j\Big)(X_v)\Big| \left|\mathcal{A}_j^1(v)\right| \, ds \\
&\le\frac{C\Delta}{h} \sum_{k=1}^n\int_{(k-1)\Delta}^{k\Delta} \left| (K')_{h}(X_v) \right| \Big| \Big(\mu^j  + \frac{1}{2} \sum_{\alpha =1}^r\sigma_\alpha^k \partial_k \sigma_\alpha^j\Big)(X_v)\Big|  \left|  \mathcal{N}(X_v)\right| \, ds \,,
\end{align*}
where in the last equality we use the fact that $\left| \frac{\partial {\mathcal{D}}_x}{\partial u_i}(X_v) \right| \le C$ for some constant $C>0$ (see the proof in Lemma \ref{A2-Analog-boundedness-computation-v2}) when $T$ is sufficiently large and $h$ is sufficiently small.
For each $j = 1,  \cdots, d$, we can apply Lemma \ref{A1-analog} with $f$ set to $|{K}'|$ and $\mathcal{H}$ set to $\left| \big(\mu^j  + \frac{1}{2} \sum_{\alpha =1}^r\sigma_\alpha^k \partial_k \sigma_\alpha^j\big)  \mathcal{N}\right|$ and obtain $\mathcal{I}_1 = O_p\left(\frac{\Delta}{h}\Upsilon(T)\right)$. 
Analogous computations lead to 
\[
\mathcal{I}_2 = O_p\Big(\frac{\Delta}{h^2}\Upsilon(T) \Big)\ \ \mbox{and}\ \ \mathcal{I}_3 = O_p\Big(\frac{\Delta}{h^2}\Upsilon(T)\Big)\,.
\]

Next, rewrite
\begin{align*}
\mathcal{I}_4 &= \sum_{k=1}^n\int_{(k-1)\Delta}^{k\Delta} \int_{(k-1)\Delta}^{s}{\sigma^j_\alpha} \mathcal{A}_j^1(v) \, dW^\alpha_v ds \\
&=  \frac{1}{h} \sum_{k=1}^n\int_{(k-1)\Delta}^{k\Delta} \int_{(k-1)\Delta}^{s} (K')_{h}(X_v)   \frac{\partial {\mathcal{D}}_x}{\partial u_j}(X_v) (\sigma^j_\alpha\mathcal{N})(X_v)  \, dW^\alpha_v ds\,. 
\end{align*}
For each $j = 1, \cdots, d$, since $\left| \frac{\partial {\mathcal{D}}_x}{\partial u_i}(X_s) \right|\leq C$, we can apply Lemma \ref{stoch-prob-order-computation-lemma} with $f$ set to $(K')^2$ and $q$ set to $(\sigma^j_\alpha\mathcal{N})^2$ to conclude $\mathcal{I}_4 = O_p\Big(\frac{\sqrt{\Upsilon(T)}\Delta}{h^{(d+2)/2} }\Big)$.
With the same approaches, it follows that 
\begin{align*}
\mathcal{J}_1 =\,& O_p\left( \Delta \Upsilon(T) \right),\ \mathcal{J}_2 = O_p\left( \Delta \Upsilon(T) \right),\ \mathcal{J}_3 =O_p\left( \Delta \Upsilon(T) \right),\\  
\mathcal{J}_4 =\,& O_p\Big( \frac{\sqrt{\Upsilon(T)}\Delta}{h^{d/2} }\Big),\ \mbox{and}\ \ \mathcal{L} = O_p \Big(\frac{\Delta}{h} \Upsilon(T) \Big)\,.
\end{align*}
 Combining all these we obtain
   \begin{align*}
   \frac{1}{\Upsilon(T)} \sum_{k=0}^{n-1} \int_{k\Delta}^{(k+1)\Delta}[q(X_s)  - q(X_{(k-1)\Delta})] ds
=O_p \Big( \frac{\Delta}{h^2}\Big)\,, 
   \end{align*}
where we use Assumption \ref{reg-ass2} that $(h^{d-2} \Upsilon(T))^{-1} = o(1)$, and so we finish the first claim.

Moreover, collecting the bounds of all terms in the decomposition, we can summarize that our argument bounded the difference 
\begin{align*}
    \left| \sum_{k=1}^n \int_{k\Delta}^{(k+1)\Delta} [q(X_s)  - q(X_{(k-1)\Delta})]ds\right| \le |A_T| + |M_T|
\end{align*}
where $A_T:=\mathcal{I}_1  + \mathcal{I}_2 + \mathcal{I}_3 + \mathcal{J}_1 + \mathcal{J}_2 + \mathcal{J}_3 + \mathcal{L}$ contains predictable process and $M_T:=\mathcal{I}_4 +  \mathcal{J}_4$ contains continuous martingale. Then, use Lemma \ref{A1-analog} on $A_T$ and $[M_T]$ to conclude the probability order that $A_T = O_p\Big( \frac{ \Upsilon(T)\Delta}{h^2}\Big)$ and $[M_T] = O_p\Big( \frac{\Upsilon(T)\Delta^2}{h^{d+2}}\Big)$,
and hence that $ \mathbb{E}^M_\lambda|A_T| = O\Big(\frac{\Upsilon(T)\Delta}{h^2} \Big)$ and $ \mathbb{E}^M_\lambda[M_T] = O\Big(\frac{\Upsilon(T)\Delta^2}{h^{d+2}} \Big)$. So, by Cauchy-Schwarz and the It\^o isometry, $ \mathbb{E}^M_\lambda(|M_T|) = O\left(\frac{\sqrt{\Upsilon(T)}\Delta}{h^{(d+2)/2}} \right)$. 
By Assumption \ref{reg-ass2}, it follows that
\begin{align*}
 \mathbb{E}^M_\lambda \left|\frac{1}{\Upsilon(T)}\sum_{k=1}^n \int_{k\Delta}^{(k+1)\Delta} [q(X_s)  - q(X_{(k-1)\Delta})]ds\right| = O\left(\frac{\Delta }{h^2} \right)
\end{align*}
and hence the second claim. 
\end{proof}

Lemma \ref{A2-analog} quantifies the difference between a continuous integral and its ``double discretization'', in which both the kernel and the function arguments replace $X_s$ by $X_{k\Delta}$. In the analysis of our estimators, we also encounter a ``single discretization'', where discretization appears only in the kernel, while function increments such as $f(X_{(k+1)\Delta}) - f(X_{k\Delta})$ are represented via continuous stochastic integrals. To streamline the subsequent analysis, we isolate this single-discretization setting in the following lemma.

\begin{lemma}\label{single-disc-difference}
With the same setup of Lemma \ref{A2-analog}, when $T=n\Delta$ is sufficiently large,
  \begin{align}
      \frac{1}{\Upsilon(T)h^d} \sum_{k=0}^{n-1} \int_{k\Delta}^{(k+1)\Delta}  {K}\left(\frac{
\mathcal{D}_x(X_{k \Delta}) }{h}\right)(\mathcal{N}(X_{k\Delta})-\mathcal{N}(X_{s})) ds &= O_p\left( \Delta  \right) \,.
  \end{align}
\end{lemma}

\begin{proof}    
Denote $ K_{h}(\cdot):=\frac{1}{h^{d}}{K}\left(\frac{{\mathcal{D}}_x(\cdot)}{h} \right)$. Applying Lemma \ref{integration-by-parts} to the decomposition of $\mathcal{N}(X_s)-\mathcal{N}(X_{k \Delta})$ as in \eqref{Ito-decomposition}, we can compute
\begin{align*}
&\underbrace{\sum_{k=0}^{n-1} \int_{k\Delta}^{(k+1)\Delta}   K_{h}(X_{k \Delta}) (\mathcal{N}(X_s)-\mathcal{N}(X_{k \Delta})) ds}_{\mathcal{S}}
\\
=\,& \underbrace{ \sum_{k=0}^{n-1} \int_{k \Delta}^{(k+1)\Delta}  K_{h}(X_{k \Delta})((k+1)\Delta-s)\left(\mu^j  + \frac{1}{2} \sum_{\alpha =1}^r\sigma_\alpha^k \partial_k \sigma_\alpha^j\right)\mathcal{N}_j(X_s) \, ds}_{\mathcal{S}_1}\\
&+ \underbrace{ \sum_{k=0}^{n-1} \int_{k \Delta}^{(k+1)\Delta} K_{h}(X_{k \Delta}) ((k+1)\Delta - s){\sigma^j_\alpha} \mathcal{N}_j(X_s)\, dW^\alpha_s}_{\mathcal{S}_2}
   \\
&+ \underbrace{\frac{1}{2} \sum_{k=0}^{n-1} \int_{k \Delta}^{(k+1)\Delta}    K_{h}(X_{k \Delta}) ((k+1)\Delta-s)\frac{1}{2} \sum_{\alpha =1}^r (\sigma_\alpha^k)^2 \mathcal{N}_{kk}(X_s) ds}_{\mathcal{S}_3} 
   \\
&+\underbrace{ \frac{1}{2} \sum_{k=0}^{n-1} \int_{k \Delta}^{(k+1)\Delta}    K_{h}(X_{k \Delta}) ((k+1)\Delta-s)\frac{1}{2} \sum_{\alpha =1}^r  (1 - \delta_{kl})\sigma_\alpha^k \sigma_\alpha^l \mathcal{N}_{kl}(X_s) ds}_{\mathcal{S}_4}.  
\end{align*}
By the regularity and growth assumptions, there exists a finite constant $C>0$ due to the regularity assumption, such that $\mathcal{S}_1$ is controlled by
\begin{align*}
|  \mathcal{S}_1 | 
&\le  C \sum_{k=0}^{n-1} \int_{k \Delta}^{(k+1)\Delta}    K_{h}(X_{k \Delta}) ((k+1)\Delta-s)  \, ds \leq  C\Delta^2  \sum_{k=0}^{n-1} K_{h}(X_{k \Delta}) \\
  & = C\Delta \int_0^T  K_{h}(X_s)\, ds+ O_p\left(\frac{\Delta^2 \Upsilon(T)}{h^2} \right)= O_p(\Delta \Upsilon(T))\,,
\end{align*}
where the first equality comes from Lemma \ref{A2-analog} with $K$ set to $\mathcal{K}$ and $\mathcal{N}$ set to $1$ and the last equality comes from Lemma \ref{A1-analog} with $K$ set to $\mathcal{K}$ and $\mathcal{H}$ set to $1$ and $\frac{\Delta}{h^2}=o(1)$. 
The same argument applies to $\mathcal{S}_3$ and $\mathcal{S}_4$. 
Finally, $\mathcal{S}_2$ can be evaluated just as in Lemma \ref{stoch-prob-order-computation-lemma} to obtain a probability order of $O_p\left(\frac{\Delta}{h^{d/2}}\Upsilon(T)^{1/2} \right)$. So, overall,  $|\mathcal{S}|=O_p\left( \Delta \Upsilon(T)\right) + O_p\left( \frac{\Delta}{h^{d/2} \Upsilon(T)^{1/2}}\right)$.
As we assume that $h^d \Upsilon(T) \rightarrow \infty$, it follows that $\frac{\Delta \Upsilon(T)^{1/2}}{h^{d/2}} = o_p(\Delta  \Upsilon(T))$. Thus, $|\mathcal{S}|=O_p\left( \Delta \Upsilon(T)\right)$ as claimed.  
\end{proof}

With Lemma \ref{single-disc-difference}, we provide an analog of the result of Lemma \ref{A2-analog} when the kernel is discretized but not the function of interest.

\begin{lemma} \label{kernel-disc-corollary}
With the same setup of Lemma \ref{A2-analog}, when $T=n\Delta$ is sufficiently large,
\begin{align*}
   &\frac{1}{{h^d}\Upsilon(T)} \sum_{k=0}^{n-1} \int_{k\Delta}^{(k+1)\Delta}  \left[{K} \left(\frac{\mathcal{D}_x(X_{k\Delta}) }{h}\right)-{K} \left(\frac{\mathcal{D}_x(X_{s}) }{h}\right)\right]\mathcal{N}(X_{s})  \, ds = {O_p \left( \frac{\Delta}{h^2}\right)} \,.
\end{align*}
\end{lemma}

\begin{proof}
 Denote $ K_{h}(\cdot):=\frac{1}{h^{d}}{K}\left(\frac{{\mathcal{D}}_x(\cdot)}{h} \right)$.   We may decompose
    \begin{align*}
  &\frac{1}{\Upsilon(T)} \sum_{k=0}^{n-1} \int_{k\Delta}^{(k+1)\Delta}  \left[ K_{h}(X_{k\Delta}) - K_{h}(X_s)\right]\mathcal{N}(X_{s})  \, ds\nonumber   \\
   =\,&  \underbrace{\frac{1}{\Upsilon(T)} \sum_{k=0}^{n-1} \int_{k\Delta}^{(k+1)\Delta}   K_{h}(X_{k\Delta}) (\mathcal{N}(X_{s})-\mathcal{N}(X_{k \Delta}) )  \, ds}_{:=\mathcal{S}_1}
   \\
   &+ \underbrace{\frac{1}{\Upsilon(T)} \sum_{k=0}^{n-1} \int_{k\Delta}^{(k+1)\Delta}  \left[ K_{h}(X_{k \Delta})\mathcal{N}(X_{k\Delta})- K_{h}(X_s)\mathcal{N}(X_{s})\right] ds}_{:=\mathcal{S}_2} \,.
    \end{align*}
    Then, apply Lemma \ref{single-disc-difference} to $\mathcal{S}_1$ and Lemma \ref{A2-analog} to $\mathcal{S}_2$.
\end{proof}

In addition to Lemma \ref{single-disc-difference}, which treats additive functionals, we also need to handle the single-discretization case for stochastic integrals. The following Lemma  completes the collection of technical results for the main proofs of the estimators.

\begin{lemma}\label{discrete-approximation-martingale-cross-quad-var}
Let $ K\in C^3(\mathbb{R})$ be non-negative and compactly supported on $[0,L] \subset [0,\infty)$. Suppose Assumptions \ref{manifold-ass}, \ref{manifold-ass2}, \ref{reg-ass}, and \ref{lebesgue-dens-ass} hold.
Consider two continuous martingales $Z^{(1)}_T$ and $Z^{(2)}_T$ of the form
\begin{align*}
 Z^{(i)}_T= &\, \frac{1}{h^{d/2}} \sum_{k=0}^{n-1}  \sqrt{K\left(\frac{
\mathcal{D}_x(X_{k \Delta}) }{h}\right)}\int_{k \Delta}^{(k+1)\Delta}q^{(i)}(X_s) dW_s \,,
\end{align*}
where $i=1,2$ and $q^{(1)},q^{(2)} \in C^2(M)$. Then, when $T=n\Delta$ is sufficiently large, we have
\begin{align*}
\frac{[Z^{(1)},Z^{(2)}]_T}{\Upsilon(T)} = \frac{1}{h^d \Upsilon(T) }  \int_0^T K\left(\frac{
\mathcal{D}_x(X_{s}) }{h}\right)q^{(1)}(X_s)q^{(2)}(X_s) ds + O_p\left( \frac{\Delta}{h^2}\right) \,.
\end{align*}
\end{lemma}

\begin{proof}  
By a direct calculation, we have
\begin{align*}
[Z^{(1)},Z^{(2)}]_t =\,& \frac{1}{h^d} \sum_{k=0}^{m-1} \int_{k\Delta}^{(k+1)\Delta} \mathcal{K}\left(\frac{
\mathcal{D}_x({ X_{k \Delta})}}{h}\right) q^{(1)}(X_s)q^{(2)}(X_s) ds \\
&+ \frac{1}{h^d} \int_{m\Delta}^t  \mathcal{K}\left(\frac{
\mathcal{D}_x(X_{m \Delta}) }{h}\right) q^{(1)}(X_s)q^{(2)}(X_s) ds \,.
\end{align*}
We can approximate this quantity by replacing the continuous process $X_s$ with discrete observations $X_{k\Delta}$. Define
\begin{align*}
\widehat{[Z^{(1)},Z^{(2)}]_T} &:= \frac{1}{h^d} \sum_{k=0}^{n-1} \int_{k\Delta}^{(k+1)\Delta}  \mathcal{K}\left(\frac{
\mathcal{D}_x(X_{k \Delta}) }{h}\right) q^{(1)}(X_{k\Delta})q^{(2)}(X_{k\Delta}) ds
\end{align*}
and consider
\begin{align*}
[Z^{(1)},Z^{(2)}]_T = \widehat{[Z^{(1)},Z^{(2)}]_T} + ( [Z^{(1)},Z^{(2)}]_T - \widehat{[Z^{(1)},Z^{(2)}]_T} )\,.
\end{align*}
First,
\begin{align*}
&[Z^{(1)},Z^{(2)}]_T - \widehat{[Z^{(1)},Z^{(2)}]_T}
\\
=\,& \frac{1}{h^d} \sum_{k=0}^{n-1} \int_{k\Delta}^{(k+1)\Delta}
\mathcal{K}\left(\frac{
\mathcal{D}_x(X_{k \Delta}) }{h}\right)
\Big(q^{(1)}(X_s)q^{(2)}(X_s)-q^{(1)}(X_{k \Delta})q^{(2)}(X_{k \Delta})\Big) ds \,, 
\end{align*}
which by Lemma \ref{single-disc-difference} with $\mathcal{N}=q^{(1)}q^{(2)}$, we obtain
\begin{align*}
  [Z^{(1)},Z^{(2)}]_T - \widehat{[Z^{(1)},Z^{(2)}]_T}
  = O_p\left(\Delta \nu(T) \right).
\end{align*}
Moreover, by Lemma \ref{A2-analog},
\begin{align*}
  \widehat{[Z^{(1)},Z^{(2)}]_T}
  &= \frac{1}{h^d} \int_{0}^{T} \mathcal{K}\left(\frac{
   {\mathcal{D}}_x(X_{s}) }{h}\right) q^{(1)}(X_{s})q^{(2)}(X_{s}) ds + O_p\left( \frac{\Delta \nu(T)}{h^2} \right).
\end{align*}
So, overall,
\begin{align*}
[Z^{(1)},Z^{(2)}]_T
= \frac{1}{h^d} \int_{0}^{T} \mathcal{K}\left(\frac{
   {\mathcal{D}}_x({X}_{s}) }{h}\right) q^{(1)}(X_{s})q^{(2)}(X_{s}) ds  + O_p\left( \frac{\Delta \nu(T)}{h^2} \right),
\end{align*}
as claimed. 
\end{proof}

\section{Proof of Theorem \ref{occ-density-theorem}}\label{Section lemmas about occupation density}

With the technical preparation in Section \ref{section technical lemma proofs}, we can easily prove Theorem \ref{occ-density-theorem}.

\begin{proof}
By Lemma \ref{A2-analog} with { $\mathcal{N}\equiv 1$}, when $T$ is sufficiently large, we immediately have 
\begin{align*}
\frac{\hat{L}^{(\texttt{o})}(x)}{\Upsilon(T)} &= \frac{L^{(\texttt{o})}(x)}{\Upsilon(T)} + O_p\left( \frac{\Delta}{h^2}\right)\,,
\end{align*}
where $L^{(\texttt{o})}(x):=\frac{1}{h^d} \int_0^T  \mathcal{K} \big(\frac{\mathcal{D}_x(X_s)}{h}\big) ds$.  
As we have assumed that $\frac{\Delta}{h^2} = o(1)$, and we have $ \Upsilon(T)^{-1}L^{(\texttt{o})}(x) \xrightarrow[]{\hspace{0.1cm}d\hspace{0.1cm}} g_\alpha(1)p_X(x) $ by Lemma \ref{A1-analog} with { $\mathcal{H}=1$}, we obtain \eqref{hat{L}-prob-order}. Moreover, by Lemma \ref{A2-analog}, we have when $T$ is sufficiently large, 
\begin{align*}
    \frac{1}{\Upsilon(T)}\mathbb{E}^M_\lambda\left| \hat{L}^{(\texttt{o})}(x) - L^{(\texttt{o})}(x)\right| = O\left(\frac{\Delta }{h^2} \right) 
\end{align*}
and as in Lemma \ref{A1-analog}, 
we have that $\mathbb{E}^M_\lambda (L^{(\texttt{o})}(x)) = O(\Upsilon(T))$. 
 \end{proof}

Note that when assumptions of Theorem \ref{occ-density-theorem} hold, since $g_\alpha(1) \neq 0$ almost surely, a ratio quantity satisfies
\begin{align}
  \frac{\Upsilon(T)}{\hat L^{(\texttt{o})}(x)}  \xrightarrow[]{\hspace{0.1cm}d\hspace{0.1cm}} ( g_\alpha(1)p_X(x) )^{-1}\label{reciprocal-lemma-nu(T)-hatL}\,,
    \end{align}
which follows from \eqref{hat{L}-prob-order} together with the continuous mapping theorem applied to the function $s\mapsto s^{-1}$.

\section{Preparation for the proof of Theorems \ref{main theorem drift}}\label{section appendix preparation for main theorem proof}

\subsection{Generalized Drift Estimator}

\begin{theorem}
\label{generalized-drift-est}
Suppose Assumptions \ref{manifold-ass}-\ref{reg-ass2} hold. 
Fix $x\in M$. Take $f\in C^2(M,\, \mathbb{R}^p)$ with $\texttt{supp} f \subset B_r(x)$ for $r < \texttt{inj}_x(M)$. { Define} 
\begin{align}
   \hat{\mu}_f(x)  := &\frac{1}{\Delta} \frac{\sum_{k=0}^{n-1}  \mathcal{K} \left(\frac{\mathcal{D}_x(X_{k \Delta}) }{h}\right)(f(X_{(k+1)\Delta})-f(X_{k\Delta}))}{\sum_{k=0}^{n-1} \mathcal{K}\left(\frac{\mathcal{D}_x(X_{k\Delta}) }{h}\right)} \label{definition generalized drift estimator}\,,
   \end{align}

\begin{equation}
{ \hat L(x)}:=\frac{\Delta}{h^d}\sum_{k=0}^{n-1} \mathcal{K} \left(\frac{\mathcal{D}_x(X_{k\Delta}) }{h}\right)\,,\label{definition Lepsilon}
\end{equation}
\begin{align}
 {B}^{\mu,\texttt{o}}_{f}(x) &:=\kappa_{2,0} \Big(\nabla^M {\mu}_{f}(x)\cdot \nabla^M \log(p_X(x))
+ \frac{1}{2}\Delta^M   {\mu}_{f}(x) \Big) \label{drift-bias}\,,
\end{align}
 where $\mu_f$ is defined in \eqref{ito-integral}, $\Delta^M$ is the Laplace-Beltrami operator on $M$, and $\nabla^M$ is the associated covariant derivative, 
and
\begin{align*}
\pi_{f}(x) := \sum_{\alpha=1}^r \sigma_{f,\alpha}(x)\sigma_{f,\alpha}(x)^\top\in \mathbb{R}^{p\times p}\,, 
\end{align*}
where $\sigma_{f,\alpha}$ is defined in \eqref{ito-integral}.
Suppose $h^d \Upsilon(T) \xrightarrow{} \infty$, $h^{d+4} \Upsilon(T) \xrightarrow[]{} C>0$, $C$ a constant, and $\frac{\Delta}{h^2}\sqrt{h^d \Upsilon(T)} \xrightarrow[]{} 0$. Since $\mathcal{D}_x(x')=\|\iota(x)-\iota(x')\|_{\mathbb{R}^p}$,  we have
\begin{align*}
\sqrt{h^d   { \hat{L}^{\texttt{(o)}}(x)}}  \big(\hat\mu_f(x) - \mu_f (x) - h^2{B}_f^{\mu,\texttt{o}}(x) \big)
\xrightarrow[]{\hspace{0.1cm}d\hspace{0.1cm}}  N(\mathbf{0}, \kappa_{2,0} \pi_f(x))\,.
\end{align*}
\end{theorem}

Note that $h^d \Upsilon(T) \xrightarrow{} \infty$ and $\frac{\Delta}{h^2}\sqrt{h^d \Upsilon(T)} \xrightarrow[]{} 0$ jointly imply $\frac{\Delta}{h^2}\to 0$ as $T\to \infty$. The quantity $\hat{\mu}_f(x)$ is a {\em generalized} drift estimator in the sense that it estimates the drift of the transformed process $f(X_t)$ at $f(x)$, and $B^\mu_{f}(x)$ quantifies the associated bias. The key idea behind the proof is to control the discrepancy between the additive functional and its uniform discretization, and then apply the generalized life-cycle decomposition to quantify the resulting error.
Before proving Theorem \ref{generalized-drift-est}, we establish a technical lemma that streamlines the argument and will also be used in the analysis of the diffusion estimator. 

\begin{lemma}\label{Theta_T-analysis-lemma} 
Assume assumptions in Theorem \ref{generalized-drift-est} hold. Consider an $\mathbb{R}^p$-valued continuous semimartingale $\Theta^f(t)$ defined for $t \ge 0$ as 
 \begin{align}
 \Theta^f(t):= &\, \frac{1}{h^{d/2}} \sum_{k=0}^{m-1}  \mathcal{K} \left(\frac{{\mathcal{D}}_x(X_{k \Delta}) }{h}\right)\int_{k\Delta}^{(k+1)\Delta}    \sigma_{f,\alpha}(X_{s}) \, d W^\alpha_s  \label{theta_t_definition} \\
 &+ \frac{1}{h^{d/2}} \mathcal{K} \left(\frac{{\mathcal{D}}_x(X_{m \Delta})  }{h}\right) \int_{m \Delta}^t  \sigma_{f,\alpha}(X_{s}) \, d W^\alpha_s  \nonumber \,,
 \end{align}
where $t \in [m \Delta, (m+1) \Delta)$ for $m=0,1,2,\ldots$, and so $\Theta^f(0)=0$ immediately. 
Denote $J_o(x):=\kappa_{2,0}p_X(x)$. Then, when $\mathcal{D}_x(x')=\|\iota(x)-\iota(x')\|_{\mathbb{R}^p}$, we have
\begin{align}
\left(\frac{\Theta^f(T)}{\sqrt{\Upsilon(T)}},\, \frac{[\Theta^f(T)]}{\Upsilon(T)} \right) \xrightarrow[]{\hspace{0.1cm}d\hspace{0.1cm}}\left(\sqrt{J_o(x) \pi_f(x)} W_\alpha(1),\ J_o(x)\pi_f(x)
g_\alpha(1)\right)\label{using-3.16}
\end{align}
in $D(\mathbb{R}_+,\mathbb{R})$ as $T \rightarrow \infty$  under $P_x$ for all $x\in M$, and
\begin{align}
 \frac{  [\Theta^f(T)]}{  { \hat{L}}(x) }
\xrightarrow[]{\hspace{0.1cm}p\hspace{0.1cm}} \frac{\kappa_{2,0}}{}\pi_f(x)
\ \mbox{ and }\   \frac{\Theta^f(T)}{\sqrt{{ \hat{L}}(x)}}
  \xrightarrow[]{\hspace{0.1cm}d\hspace{0.1cm}} N(\mathbf{0}, \kappa_{2,0}\pi_f(x))  \label{convergence-of-Theta^f}
\end{align}
as $T\to \infty$ under $P_x$ for all $x\in M$. Moreover,
\begin{align}
    \Theta^f(T) = O_p(\sqrt{\Upsilon(T)}) \label{prob-order-Theta^f_T}\,.
\end{align}
\end{lemma}

\begin{proof}  
Denote
\begin{align*}
 \overline\Theta^f(t):= &\, \frac{1}{h^{d/2}}  \int_0^T \mathcal{K} \left(\frac{{\mathcal{D}}_x(X_{s}) }{h}\right)    \sigma_{f,l}(X_{s}) \, d W^l_s  \,.
 \end{align*}
Denote $\Theta^f(t)=[\Theta^f_1(t),\ldots,\Theta^f_p(t)]^\top\in \mathbb{R}^p$. 
By Lemma \ref{discrete-approximation-martingale-cross-quad-var}, for any $\alpha,\beta=1,\ldots,p$,
\begin{align}\label{eq Theta and overlineTheta QV difference}
\frac{[\Theta^f_\alpha(T), \Theta^f_\beta(T)]}{\Upsilon(T)} &= \frac{[\overline\Theta^f_\alpha(T), \overline\Theta^f_\beta(T)]}{\Upsilon(T)} + O_p\left(\frac{\Delta}{h^2} \right) \,.
\end{align}
When $\alpha=\beta$, since $\pi_f^{\alpha, \alpha}(x)\geq 0$ is finite by the regularity assumption and $\frac{\Delta}{h^2}\to 0$ as $T\to \infty$, then by Lemma \ref{A1-analog} with ${K}(\cdot)$ and $\mathcal{H}(\cdot)$ set to $\mathcal{K}^2(\cdot)$ and $\pi_f^{\alpha, \alpha}(\cdot)$, we have  
\begin{align*}
 \frac{ [\overline\Theta^f_\alpha(T), \overline\Theta^f_\alpha(T)] }{ \Upsilon(T) }\xrightarrow[]{\hspace{0.1cm}d\hspace{0.1cm}} J_o(x)\pi_f^{\alpha, \alpha}(x)  g_\alpha(1) 
\end{align*}
as $T\to \infty$, where the coefficient $J_o(x):=\kappa_{2,0}p_X(x)$ is calculated via Proposition \ref{2.20and4.3}. 
When $\alpha\neq \beta$, in general we may not have $\pi_f^{\alpha, \beta}(x) \ge 0$, but the analysis is similar to the diagonal case by applying Lemma \ref{discrete-approximation-martingale-cross-quad-var} to $[\Theta^f_\alpha(T), \Theta^f_\beta(T)]$ followed by invoking the polarization identity, i.e., 
\[
[\overline\Theta^f_\alpha(T), \overline\Theta^f_\beta(T)] = \frac{1}{4}\left([\overline\Theta^f_\alpha(T)+ \overline\Theta^f_\beta(T)]  - [\overline\Theta^f_\alpha(T)- \overline\Theta^f_\beta(T)] \right)\,.
\] 
We can then apply Lemma \ref{A1-analog} to get
\[
\frac{[\overline\Theta^f_\alpha(T),\overline \Theta^f_\beta(T)]}{\Upsilon(T)}\xrightarrow[]{\hspace{0.1cm}d\hspace{0.1cm}} J_o(x)\pi_f^{\alpha, \beta}(x)  g_\alpha(1) 
\] 
when $T\to \infty$, since we can collapse the polarization identity in the limiting distribution by applying the ratio limit theorem.
Denote $[\Theta^f(T)]:=([\Theta^f_\alpha(T), \Theta^f_\beta(T)])_{\alpha,\beta=1}^p$ and $[\overline\Theta^f(T)]:=([\overline\Theta^f_\alpha(T), \overline\Theta^f_\beta(T)])_{\alpha,\beta=1}^p$. We have 
\[
\frac{[\Theta^f(T)]}{\Upsilon(T)}\xrightarrow[]{\hspace{0.1cm}d\hspace{0.1cm}} J_o(x) \pi_f(x)  g_\alpha(1)
\] 
when $T\to \infty$, and hence \eqref{prob-order-Theta^f_T}. 
Since $\Theta^f(T)$ is locally square integrable and locally bounded, we conclude \eqref{using-3.16} from Proposition \ref{thm3.16_limit_theorems_null}; that is,
\begin{align*}
&\left(\frac{\Theta^f(T)}{\sqrt{\Upsilon(T)}},\, \frac{[\Theta^f(T)]}{\Upsilon(T)} \right)=\left(\frac{\overline\Theta^f(T)}{\sqrt{\Upsilon(T)}},\, \frac{[\overline\Theta^f(T)]}{\Upsilon(T)} \right) + O_p\left(\frac{\Delta }{h^2} \right) \\
\xrightarrow[]{\hspace{0.1cm}d\hspace{0.1cm}}&\,\left(\sqrt{J_o(x) \pi_f(x)} W_\alpha(1),\ J_o(x) \pi_f(x) g_\alpha(1)\right)\,,
\end{align*}
where the first equality comes from Lemma \ref{kernel-disc-corollary} and \eqref{eq Theta and overlineTheta QV difference}.
Recall that $\pi_f(x)$ is non-negative definite by the uniform ellipticity assumption. 

Next, by the same argument as that in \cite[Remark 4.26 and Chapter 7]{limit_theorems_null}, we have
\begin{equation}\label{Remark 4.26 and Chapter 7 limit_theorems_null}
\left(\frac{\Theta^f(T)}{\sqrt{\Upsilon(T)}},\, \sqrt{\frac{{ \hat{L}}(x)}{\Upsilon(T)}} \right)\xrightarrow[]{\hspace{0.1cm}d\hspace{0.1cm}} \left(\sqrt{J_o(x)\pi_f(x)}\sqrt{g_\alpha(1)}Z,\, \sqrt{p_X(x) g_\alpha(1)}\right)\,,
\end{equation} 
where $Z \sim N(0, I_p)$.
By the continuous mapping theorem with $(x,y)\mapsto x/y$, we obtain 
\[
\frac{\Theta^f(T)}{\sqrt{{ \hat{L}}(x)}}\xrightarrow[]{\hspace{0.1cm}d\hspace{0.1cm}}N(\mathbf{0}, \kappa_{2,0}\pi_f(x)) \,.
\] 
The asymptotic behavior of $\frac{  [\Theta^f(T)]}{  { \hat{L}}(x) }$ can be analyzed similarly with a weak convergence, or analyzed directly by applying Lemma \ref{A1-analog} to analyze the denominator and numerator of $\frac{  [\Theta^f_\alpha(T), \Theta^f_\beta(T)] }{ { \hat{L}}(x) }$ simultaneously, which leads to
\begin{align*}
 \frac{  [\Theta^f_\alpha(T), \Theta^f_\beta(T)] }{ { \hat{L}}(x) }
= 
\frac{ \mathbb{E}^{\mathcal M}_\lambda \Big( \int_{R_1}^{R_{2}} (\mathcal{K}^2)_{h}(X_{s}) \pi_f^{\alpha, \beta}(X_{s})  \, ds \Big) + \frac{\Upsilon(T)}{N_T} O_p\Big( \frac{1}{\sqrt{\Upsilon(T)}} \Big) + \frac{\Upsilon(T)}{N_T} O_p\left(\frac{\Delta}{h^2} \right)}
{ \mathbb{E}^{\mathcal M}_\lambda\left( \int_{R_1}^{R_{2}} \mathcal{K}_{h}(X_{s})  \, ds \right) + \frac{\Upsilon(T)}{N_T} O_p\Big( \frac{1}{\sqrt{\Upsilon(T)}}\Big) + \frac{\Upsilon(T)}{N_T} O_p\left(\frac{\Delta}{h^2} \right)}\,.
\end{align*}
where $\mathcal{K}_{h}(\cdot):=\frac{1}{h^{d}}\mathcal{K}\left(\frac{{\mathcal{D}}_x(\cdot)}{h} \right)$.
Since $\frac{\Upsilon(T)}{N_T}$ converges weakly to a strictly positive Mittag-Leffler random variable in Lemma \ref{Premathcal{B}-conv}, with Lemma \ref{reciprocal-lemma1} we conclude that
\begin{align*}
\frac{  [\Theta^f(T)]}{  { \hat{L}}(x) }
\xrightarrow[]{\hspace{0.1cm}p\hspace{0.1cm}} \kappa_{2,0}\pi_f(x) 
\end{align*} 
when $T\to \infty$, and hence the proof.

\end{proof}

\begin{proof}[Proof of Theorem \ref{generalized-drift-est}]
Denote $\mathcal{K}_{h}(\cdot):=\frac{1}{h^{d}}\mathcal{K}\left(\frac{{\mathcal{D}}_x(\cdot)}{h} \right)$. Take a normal coordinate chart on $B_r(x)$. 
First, we proceed in the case that $\mathcal{D}_x(x')=\|\iota(x)-\iota(x')\|_{\mathbb{R}^p}$.
Plugging in It\^o's formula \eqref{ito-integral} that $f(X_{(k+1)\Delta}) - f(X_{k\Delta})= \int_{k \Delta }^{(k+1) \Delta}  \mu_f(X_t)\,dt +  \int_{k \Delta}^{(k+1) \Delta} \sigma_{f, \alpha}(X_t) \, d W^\alpha_t$
into \eqref{definition generalized drift estimator}, we have the bias and variance decomposition of the error: 
\begin{align}
\mathcal{E} := \hat{\mu}_f(x) - \mu_f(x)
=\,& \underbrace{ \frac{ \sum_{k=1}^{n-1} \mathcal{K}_{h}(X_{k \Delta})  \int_{k \Delta}^{(k+1) \Delta} \mu_f(X_{s})\,ds }{{ \hat{L}}(x)}-\mu_f(x)}_{:=\mathsf{B}} \label{definition Err B}\\
&+ \underbrace{ \frac{ \sum_{k=1}^{n-1}  \mathcal{K}_{h}(X_{k \Delta})  \int_{k \Delta}^{(k+1) \Delta} \sigma_{f,\alpha}(X_{s}) \, d W^\alpha_s }{{ \hat{L}}(x)}}_{:=\mathsf{V}} \,.\nonumber
\end{align}
We start with $\mathsf{V}$. Note that $\mathsf{V}= \frac{1}{h^{d/2}}\frac{ \Theta^f(T)}{{ \hat{L}}(x)} $, where $\Theta^f(T)$ is defined in \eqref{theta_t_definition}. 
By \eqref{convergence-of-Theta^f} of Lemma \ref{Theta_T-analysis-lemma},
\begin{align}
    \sqrt{h^d\hat{L}^{\texttt{(o)}}(x)}\mathsf{V} &= \frac{ \Theta^f(T)}{\sqrt{{ \hat{L}}(x)}} \xrightarrow[]{\hspace{0.1cm}d\hspace{0.1cm}} N(\mathbf{0}, \kappa_{2,0}\pi_f(x))\,.
\end{align}
Next, consider $\mathsf{B}$.
For $i = 1, \cdots, p$, by Lemma \ref{kernel-disc-corollary} with $\mathcal{N}$ set to $e_i^\top\mu_f$, we obtain 
\begin{align*}
\sum_{k=0}^{n-1} \mathcal{K}_{h}(X_{k \Delta}) \int_{k\Delta}^{(k+1)\Delta} \mu_f(X_s)ds =  \int_0^T \mathcal{K}_{h}(X_s) \mu_f(X_s) ds + O_p\left(\frac{\Delta   \Upsilon(T)}{h^2}\right)  \,.
\end{align*}
We have a similar expression for ${ \hat{L}}(x)$. Therefore, 
\begin{align*}
\mathsf{B} =\,& \frac{\frac{1}{\Upsilon(T)} \int_0^T \mathcal{K}_{h}(X_s) \left(\mu_{f}(X_s) - \mu_{f}(x)\right)  ds + O_p\left(\frac{\Delta}{h^2}\right)   
}{\frac{1}{\Upsilon(T)}\int_0^T \mathcal{K}_{h}(X_s)  ds+ O_p \left( \frac{\Delta }{h^2}\right)}  \\
=\,& \frac{ \int_0^T \mathcal{K}_{h}(X_s) \left(\mu_{f}(X_s) - \mu_{f}(x)\right)  ds}{ \int_0^T\mathcal{K}_{h}(X_s)  ds} + O_p \left( \frac{\Delta}{h^2}\right) \,,
\end{align*}
where the first equality holds since $\mu_f$ is bounded by the drift assumption and the second equality comes from Lemma \ref{reciprocal-lemma2}.  Write
\begin{align*}
\mathsf{B} = \underbrace{ \frac{  \int_0^T \mathcal{K}_{h}(X_s) \left(\mu_{f}(X_s) - \mu_{f}(x)\right)  ds}{ \int_0^T \mathcal{K}_{h}(X_s)   ds}- \mathsf{B}_0 }_{:=\mathsf{B}_1} +   \mathsf{B}_0+ O_p \left( \frac{\Delta}{h^2}\right) \,,
\end{align*}
where
\begin{align}
\mathsf{B}_0 &=  \frac{N_T   \mathbb{E}^{\mathcal M}_\lambda\left( \int_{R_1}^{R_{2}} \mathcal{K}_{h}(X_s) \left(\mu_{f}(X_s) - {\mu}_{f}(x)\right)  ds\right)}{  N_T   \mathbb{E}^{\mathcal M}_\lambda\left( \int_{R_1}^{R_{2}} \mathcal{K}_{h}(X_s)  ds\right) }\,. \label{mathcal{E}_B-defn}
\end{align}
Analyzing $\mathsf{B}$ is thus reduced to controlling two terms. We have
$h^{-2} \mathsf{B}_0  \xrightarrow[]{\hspace{0.1cm}p\hspace{0.1cm}} {B}^{\mu,\texttt{o}}_{f}(x)$ when $T\to \infty$, where ${B}^{\mu,\texttt{o}}_{f}(x)$ is defined as in \eqref{drift-bias}.
To see this claim, since $\frac{\Delta}{h^2}\to 0$ by the assumption $h^d \Upsilon(T) \xrightarrow{} \infty$ and $\frac{\Delta}{h^2}\sqrt{h^d \Upsilon(T)} \xrightarrow[]{} 0$, we can apply exactly the same analysis as that of $\mathcal{C}$ in Lemma \ref{A1-analog} to both the denominator and numerator with Lemma \ref{lemma: expectation expansion}
Second, applying exactly the same analysis of  $\mathcal{A}$ and $\mathcal{B}$ in Lemma \ref{A1-analog}, we have $ \mathsf{B}_1 = o_p(h^2)$. 
As a result, since $\mathcal{D}_x(x')=\|\iota(x)-\iota(x')\|_{\mathbb{R}^p}$, 
\begin{align*}
&\sqrt{h^d { \hat{L}}(x)}\left( \mathsf{B} - h^2 {B}_f^{\mu,\texttt{o}}(x)\right)  \\
=&\, o_p\Big(\sqrt{h^{d+4}{ \hat{L}}}\Big)+\sqrt{h^{d+4} { \hat{L}}(x)}\left(h^{-2}\mathsf{B}_0 -  {B}_f^{\mu,\texttt{o}}(x)\right) + O_p\Big(\frac{\Delta}{h^2}\sqrt{h^d { \hat{L}}(x)}\Big)=o_p(1)\,,
\end{align*}
where the last control comes from Slutsky's theorem, and jointly from the assumptions $\frac{\Delta}{h^2}\sqrt{h^d \Upsilon(T)}\to 0$, $h^{d+4}\Upsilon(T) \to C>0$, $\frac{{ \hat{L}}(x)}{\Upsilon(T)}  \xrightarrow[]{\hspace{0.1cm}d\hspace{0.1cm}}   g_\alpha(1)p_X(x)$ by Lemma \ref{A1-analog},  and $h^{-2} \mathsf{B}_0  \xrightarrow[]{\hspace{0.1cm}p\hspace{0.1cm}} {B}^{\mu,\texttt{o}}_{f}(x)$ shown above.
We therefore obtain the claim 
\begin{align*}
\sqrt{h^d   { \hat{L}}(x)}  \big(\hat{\mu}_f (x) - \mu_f(x) - h^2 {B}_f^{\mu,\texttt{o}}(x) \big) 
\xrightarrow[]{\hspace{0.1cm}d\hspace{0.1cm}} N(\mathbf{0}, \kappa_{2,0}\pi_f(x))\,.
\end{align*}

\end{proof}

\section{Preparation for the proof of Theorem \ref{euclidean-diffusion-estimate-total}}

\begin{theorem}\label{generalized-diff-est} 
Suppose Assumptions \ref{manifold-ass}, \ref{manifold-ass2}, \ref{reg-ass},  \ref{lebesgue-dens-ass}, \ref{kernel-ass}, and \ref{reg-ass2} hold. 
Fix $x\in M$. Consider functions $f,q\in C^3(M,\, \mathbb{R}^p)$ with $\texttt{supp} f \subset B_r(x)$ and $\texttt{supp} q \subset B_r(x)$ for $r < \texttt{inj}_x(M)$. For a kernel function $\mathcal{K}$ and $\mathcal{D}_x(x')=\|\iota(x)-\iota(x')\|_{\mathbb{R}^p}$, define 
\begin{equation*}
\hat{\pi}_{f,q}(x) :=   \frac{1}{\Delta}  \frac{\sum_{k=0}^{n-1}  \mathcal{K} \left(\frac{
   {\mathcal{D}}_x({X}_{k\Delta}) }{h}\right)({f}(X_{(k+1)\Delta})-{f}(X_{k\Delta}))(q(X_{(k+1)\Delta})-q(X_{k\Delta}))^\top}{\sum_{k=0}^{n-1} \mathcal{K} \left(\frac{
   {\mathcal{D}}_x({X}_{k\Delta}) }{h}\right)}
 \end{equation*}
and
\begin{equation*}
\hat{L}(x):=\frac{\Delta}{h^d}\sum_{k=0}^{n-1} \mathcal{K}\left(\frac{\mathcal{D}_x(X_{k\Delta}) }{h}\right)\,.
\end{equation*}
Denote 
\[
\pi_{f,q}(x) := \sum_{l=1}^r \sigma_{f,\alpha}(x)\sigma_{q,\alpha}(x)^\top\in \mathbb{R}^{p\times p}\,, 
\]
where $\sigma_{f,\alpha}\in \mathbb{R}^p$ is defined in \eqref{ito-integral},
\begin{align}
B^\pi_{f,q}(x) := \sum_{i=1}^d  \partial_i {\pi}_{f,q}(
x) \partial_i \log p_X(x) + \frac{1}{2} \sum_{i=1}^d  \partial_i^2{\pi}_{f,q}(x)   \label{diffusion-bias}\,,
\end{align} 
where $\partial_i$ are the normal coordinate vector fields centered at $x$, and 
\begin{align*}
{\Xi}_{f,q}(x) := \frac{1}{2}(\pi_{f,f}(x) \otimes \pi_{q,q}(x)+\pi_{q,q}(x) \otimes \pi_{f,f}(x))\,.
\end{align*}
With conditions $\frac{h^d  \Upsilon(T)}{\Delta} \xrightarrow{}\infty$,  $\frac{h^{d+4} \Upsilon(T)}{\Delta} \xrightarrow[]{} C$, $C>0$,  and $h^{d-4}  \Upsilon(T) \Delta \xrightarrow[]{} 0$, we have 
\begin{align*}
\sqrt{\frac{h^d   { \hat{L}}(x)}{\Delta}}   \left( \hat{{\pi}}_{f,q}(x) - {\pi}_{f,q}(x)   - h^2 {B}^{\pi}_{f,q}(x)\right)
\xrightarrow[]{\hspace{0.1cm}d\hspace{0.1cm}}  N(\boldsymbol{0}, \kappa_{2,0}\Xi_{f,g}(x))\,.
\end{align*}
\end{theorem}

Note that $\hat{\pi}_{f,q}$ resembles a kernel diffusion estimator, but is generalized to capture the interaction between the functions $f$ and $q$. When $q = f$, it estimates the $p \times p $ diffusion matrix of the process ${f}(X_t)$ at $f(x)$.  Here, ${B}^{\pi}_{f,q}(x)$ quantifies the bias of the estimator.
Before proving Theorem \ref{generalized-diff-est}, we present some technical lemmas that streamline the argument.

\begin{lemma} \label{diff-est-proof-lemma1}
Assume assumptions in Theorem \ref{generalized-diff-est} hold. Consider  $f\in C(M)$ and $q\in C^3(M,\, \mathbb{R}^p)$ with $\texttt{supp} f \subset B_r(x)$ and $\texttt{supp} q \subset B_r(x)$ for $r < \texttt{inj}_x(M)$.
Denote
\begin{align*}
A_T\,&:=\frac{1}{h^d}\sum_{k=0}^{n-1} \mathcal{K}\left(\frac{{\mathcal{D}}_x(X_{k \Delta}) }{h} \right)\int_{k\Delta}^{(k+1)\Delta} 
\left( \int_{k\Delta}^s f(X_t) dt  \right){\sigma}_{q, \alpha}(X_s) dW^\alpha_s  \\
B_T\,&:=\frac{1}{h^d}\sum_{k=0}^{n-1} \mathcal{K}\left(\frac{{\mathcal{D}}_x(X_{k \Delta})}{h} \right)\Delta f(X_{k\Delta})  \int_{k\Delta}^{(k+1)\Delta}{\sigma}_{q, \alpha}(X_s) dW^\alpha_s \,, 
\end{align*}
where $T=n\Delta$. When $T$ is sufficiently large, we have 
\begin{align*}
\frac{A_T}{\sqrt{\Upsilon(T)}}= O_p\left( \frac{\Delta}{h^{d/2}}\right)\ \ \mbox{and}\ \ \frac{B_T}{\sqrt{\Upsilon(T)}}=  O_p\left( \frac{\Delta}{h^{d/2}}\right)\,.
\end{align*}
=\end{lemma}

\begin{proof}[Proof of Lemma \ref{diff-est-proof-lemma1}]
With $\mathcal{K}_{h}(\cdot):=\frac{1}{h^{d}}\mathcal{K}\left(\frac{{\mathcal{D}}_x(\cdot)}{h} \right)$ and \texttt{Tr} denoting matrix trace, we have
 \begin{align*}
\left[A_T\right]
=\,&\sum_{k=0}^{n-1} \mathcal{K}_{h}(X_{k \Delta})^2 \int_{k\Delta}^{(k+1)\Delta} 
 \left( \int_{k\Delta}^s f({X_t})dt  \right)^2 \texttt{Tr}\pi_{q,q}(X_s) \, ds \\
\le\,& \sum_{k=0}^{n-1} \mathcal{K}_{h}(X_{k \Delta})^2 \int_{k\Delta}^{(k+1)\Delta} 
 \left( \int_{k\Delta}^{(k+1)\Delta} f(X_t)dt  \right)^2 \texttt{Tr}\pi_{q,q}(X_s) \, ds  \\ 
\leq\,& \Delta^2   C_1^2   \sum_{k=0}^{n-1} \mathcal{K}_{h}(X_{k \Delta})^2 \int_{k\Delta}^{(k+1)\Delta} \texttt{Tr}\pi_{q,q}(X_s) ds 
= \frac{ \Delta^2C_1^2 }{h^{d}}  [\Theta_T^q] \,,
\end{align*}
where $C_1>0$ is finite depending on the regularity assumption of $f$, and $\Theta_T^q$ is defined in \eqref{theta_t_definition}.
By the same argument for $[\Theta_T^q] $ in Lemma \ref{Theta_T-analysis-lemma}, we have $[\Theta_T^q] =  O_p(\Upsilon(T))$ since our assumptions imply $\frac{\Delta}{h^2} \xrightarrow{} 0$, and thus
$\left[A_T\right] = O_p\big( \frac{\Delta^2   \Upsilon(T)}{h^d}\big)$,
which leads to
$A_T = O_p\big( \frac{\Delta   \sqrt{\Upsilon(T)}}{h^{d/2}}\big)$. 
For $B_T$, we have
\begin{align*}
    [B_T] &= \Delta^2 \sum_{k=0}^{n-1} \mathcal{K}_{h}(X_{k \Delta})^2 f(X_{k\Delta})^2 \int_{k\Delta}^{(k+1)\Delta} \pi_{q,q}(X_s) \, ds 
    \le  \frac{ \Delta^2C_1^2 }{h^{d}}  [\Theta_T^q]\,,
\end{align*} 
and hence the claim.
\end{proof}

Now, we prove the main theorem of the section.

\begin{proof}[Proof of Theorem \ref{generalized-diff-est}]
We start with introducing notation. Denote $\mathcal{K}_{h}(\cdot):=\frac{1}{h^{d}}\mathcal{K}\left(\frac{{\mathcal{D}}_x(\cdot)}{h} \right)$. For $k=0,1,\ldots,n-1$, we suppress the notation for the dependence of $\Omega$ on $k$ and define 
\begin{align*}
   {\Omega}(x) := (f(x) - f(X_{k\Delta}))(q(x) - q(X_{k\Delta}))^\top\in C_0^3(M,\mathbb{R}^{p \times p}) \,,
\end{align*}
which appears in the numerator of $\hat{\pi}_{f,q}(x)$.
We use superscripts to index coordinates. Specifically, $f^a$ and $q^b$ denote the $a$-th and $b$-th components of $f$ and $q$, respectively, and $\Omega^{a,b}$ denotes the $(a,b)$-th entry of $\Omega$, where $a,b = 1, \cdots, p$. We work in normal coordinates on $B_r(x)$. 
Since $\Omega(X_{k\Delta})=0$, for $s \in[k \Delta, (k+1)\Delta]$, It\^o's formula \eqref{ito-integral} gives  
\begin{align}
    \Omega(X_s) 
 =\,&  \int_{k\Delta}^{s} \underbrace{\Big(\mu^j  + \frac{1}{2} \sum_{l =1}^r\sigma_l^i \partial_i \sigma_l^j\Big)(X_v)\Omega_j(X_v)}_{:=\mathcal{M}(X_v)}\, dv \nonumber\\
&+ \int_{k\Delta}^{s}\underbrace{(\sigma^j_l \Omega_j)(X_v)}_{:=\mathcal{S}_l(X_v)}\, dW^l_v  + \int_{k\Delta}^{s} \underbrace{ \frac{1}{2} \sum_{l =1}^r  (\sigma_l^i \sigma_l^j \Omega_{i,j})(X_v)}_{:=\mathcal{H}(X_v)}dv \label{mathcal{S}^{a,b}}\,.
\end{align}
We therefore have
\begin{align*}
    \hat{\pi}_{f,q}(x) &= \hat{\Pi}_1(x) + \hat{\Pi}_2(x) + \hat{\Pi}_3(x)\,,
\end{align*}
where 
\begin{align}
    \hat{\Pi}_i(x) := \frac{\Gamma_i(x)}{{ \hat{L}}(x)} \label{hat-pi-defn}
\end{align}
for $i = 1,2,3$, and 
\begin{align}
     \Gamma_1(x) &=  \sum_{k=0}^{n-1} \mathcal{K}_{h}(X_{k\Delta}) \int_{k\Delta}^{(k+1)\Delta} \mathcal{M}(X_s) ds  \in \mathbb{R}^{p \times p} \nonumber\\
      \Gamma_2(x) &=  \sum_{k=0}^{n-1} \mathcal{K}_{h}(X_{k\Delta}) \int_{k\Delta}^{(k+1)\Delta} \mathcal{S}_\alpha(X_s) dW^\alpha  \in \mathbb{R}^{p \times p}\label{Gamma2} \\
      \Gamma_3(x) &=  \sum_{k=0}^{n-1} \mathcal{K}_{h}(X_{k\Delta}) \int_{k\Delta}^{(k+1)\Delta} \mathcal{H}(X_s) ds   \in \mathbb{R}^{p \times p} \nonumber\,.
\end{align}
As we will show below, asymptotically $\hat{\Pi}_1(x)$ is negligible, $\hat{\Pi}_2(x)$ goes to a Gaussian random matrix, and $\hat{\Pi}_3(x)$ gives the targeting diffusion matrix. Note that with subscripts denoting differentiation in normal coordinates,
\begin{align}
 \Omega_i^{a,b}(x) = &\, f_i^a(x)(q^b(x) - q^b(X_{k\Delta})) +  q_i^b(x)(f^a(x) - f^a(X_{k\Delta}))\nonumber \\
  \Omega_{i,j}^{a,b}(x) = &\, f_{i,j}^a(x)(q^b(x) - q^b(X_{k\Delta})) \label{Omega_i}  \\
&+  q_{i,j}^b(x)(f^a(x) - f^a(X_{k\Delta})) +f_i^a(x) q_j^b(x) +    f_j^a(x) q_i^b(x) \,. \nonumber 
\end{align}

We start with $ \hat{\Pi}_3(x)$. By plugging \eqref{Omega_i} into $\mathcal{H}(X_s)$, by symmetry we obtain 
\begin{align*}
e_a^\top\mathcal{H}(X_s)&e_b-\pi_{f,q}^{a,b}(X_s)=\frac{1}{2}\sum_{\alpha=1}^r(\sigma_\alpha^i\sigma_\alpha^j)(X_s)\\
&\times \big(f^a_{i,j}(X_s)(q^b(X_s)-q^b(X_{k\Delta}))+q^a_{i,j}(X_s)(f^b(X_s)-f^b(X_{k\Delta}))\big)\,.
\end{align*}
Therefore, by the same analysis of $\mathsf{B}$ as in \eqref{definition Err B}, for $a,b=1,\ldots,p$, we get 
\begin{align*}
e_a^\top(\hat{\Pi}_3(x)- \pi_{f,g}(x) )e_b &= h^2 e_a^\top B^{\pi}_{f,g}e_b + O_p\left( \frac{\Delta}{h^2}\right) + o_p(h^2) \,.
\end{align*}
Hence, under the assumptions that $\frac{h^d \Upsilon(T)}{\Delta} \xrightarrow{}\infty$ and $h^{d-4} \Upsilon(T) \Delta \xrightarrow[]{} 0$, we have $\frac{{ \hat{L}}(x)}{\Upsilon(T)}  \xrightarrow[]{\hspace{0.1cm}d\hspace{0.1cm}}  g_\alpha(1)p_X(x)$ by Lemma \ref{A1-analog}, which combined with the assumption that $\frac{h^{d+4}  \Upsilon(T)}{\Delta} \xrightarrow[]{\hspace{0.1cm}\hspace{0.1cm}} C$ leads to
\begin{align*}
\sqrt{\frac{h^d{ \hat{L}}(x)}{\Delta}}e_a^\top (\hat{\Pi}_3(x)-\pi_{f,g}(x)-h^2B^{\pi}_{f,g})e_b= o_p(1)\,. 
\end{align*}

Next, we analyze $\hat{\Pi}_1(x)$. When $s \in [k\Delta, (k+1)\Delta]$, by plugging the It\^o's formula \eqref{ito-integral}, 
\[
q^b(X_s) - q^b(X_{k\Delta}) = \int_{k\Delta}^s \mu_{q^b}(X_t) \, dt + \int_{k\Delta}^s \sigma_{q^b, \alpha}(X_t) \, dW^\alpha_t 
\]
into \eqref{Omega_i}, we have a decomposition
$\Gamma_1(x)=\Gamma^{f,\mu}_1(x)+\Gamma^{f,\sigma}_1(x)+\Gamma^{q,\mu}_1(x)+\Gamma^{q,\sigma}_1(x)$, where 
\begin{align*}
& e_a^\top \Gamma_1^{f,\mu}(x) e_b := 
 \sum_{k=0}^{n-1} \mathcal{K}_{h}(X_{k\Delta}) \int_{k\Delta}^{(k+1)\Delta} 
\left( \int_{k\Delta}^s  \mu_{q^b}(X_t) dt\right) \zeta^a(X_s) ds\,,
\\
&e_a^\top \Gamma_1^{f,\sigma}(x) e_b :=  
 \sum_{k=0}^{n-1} \mathcal{K}_{h}(X_{k\Delta})\int_{k\Delta}^{(k+1)\Delta} 
\left( \int_{k\Delta}^s  \sigma_{q^b, \alpha}(X_t)dW^\alpha _t\right)\zeta^a(X_s) ds\,,
\end{align*}
$\zeta^a:=\left(\mu^i  + \frac{1}{2} \sum_{\alpha =1}^r\sigma_\alpha^\ell \partial_\ell \sigma_\alpha^i\right) f_i^a$, and $\Gamma^{q,\mu}_1(x)$ and $\Gamma^{q,\sigma}_1(x)$ are defined similarly. We only analyze $\Gamma^{f,\mu}_1(x)$ and $\Gamma^{f,\sigma}_1(x)$ since the analysis for the reversely defined quantities is identical. By the regularity assumption, for all indices $a,b$ we have  
\begin{align*}
|e_a^\top \Gamma_1^{f,\mu}(x) e_b| \le    C  \Delta^2 \sum_{k=0}^{n-1} \mathcal{K}_{h}(X_{k\Delta})
\end{align*}
for some finite constant $C>0$ due to the regularity assumption. By Lemmas \ref{A1-analog} and \ref{A2-analog}, 
\begin{align*}
 \Delta \sum_{k=0}^{n-1} \mathcal{K}_{h}(X_{k\Delta}) &= O_p(\Upsilon(T)) + O_p\left( \frac{\Delta   \Upsilon(T)}{h^2}\right) 
\end{align*}
and hence
\begin{align*}
e_a^\top \Gamma_1^{f,\mu}(x) e_b=O_p(\Delta   \Upsilon(T)) + O_p\left( \frac{\Delta^2   \Upsilon(T)}{h^2}\right) \,.
\end{align*}
The term $\Gamma_{1}^{f,\sigma}$ involves martingales and is more complicated. Using stochastic integration by parts (Lemma \ref{integration-by-parts}), we decompose
\begin{align*}
e_a^\top {\Gamma}_1^{f,\sigma} e_b 
=\,& \underbrace{ \sum_{k=0}^{n-1} \mathcal{K}_{h}(X_{k\Delta}) \int_{k\Delta}^{(k+1)\Delta} 
\zeta^a(X_s)ds   \times \int_{k\Delta}^{(k+1)\Delta} {\sigma}_{q^b, \alpha}(X_s) dW^\alpha_s}_{:=e_a^\top {\Gamma}_{1,1}^{f,\sigma} e_b}\\
&- \underbrace{ \sum_{k=0}^{n-1} \mathcal{K}_{h}(X_{k\Delta}) \int_{k\Delta}^{(k+1)\Delta} 
\left( \int_{k\Delta}^s \zeta^a(X_t)dt  \right){\sigma}_{q^b, \alpha}(X_s) dW^\alpha_s}_{:=e_a^\top {\Gamma}_{1,2}^{f,\sigma} e_b}\,.
\end{align*}
Lemma \ref{diff-est-proof-lemma1} gives $e_a^\top \Gamma_{1,2}^{f,\sigma} e_b = O_p\big( \frac{\Delta   \sqrt{\Upsilon(T)}}{h^{d/2}}\big)$. We turn to $e_a^\top {\Gamma}_{1,1}^{f,\sigma} e_b$, which involves kerneled summation of products of $\int_{k\Delta}^{(k+1)\Delta} \zeta^a(X_s)ds$ and  $\int_{k\Delta}^{(k+1)\Delta} {\sigma}_{q^b, \alpha}(X_s) dW^\alpha_s$. We estimate it by leveraging approximations at uniform discrete times:
\begin{align*}
e_a^\top {\Gamma}_{1,1}^{f,\sigma}\,& e_b
= \underbrace{ \Delta \sum_{k=0}^{n-1} \mathcal{K}_{h}(X_{k\Delta})  \zeta^a(X_{k\Delta}) \times \int_{k\Delta}^{(k+1)\Delta}{\sigma}_{q^b, \alpha}(X_s) dW^\alpha_s}_{:=e_a^\top \hat{\Gamma}_{1,1}^{f,\sigma} e_b} \\
&+\underbrace{  \sum_{k=0}^{n-1} \mathcal{K}_{h}(X_{k\Delta})   \int_{k\Delta}^{(k+1)\Delta} 
  {\sigma}_{q^b, \alpha}(X_s) dW^\alpha_s \int_{k\Delta}^{(k+1)\Delta} 
  (\zeta^a(X_s) - \zeta^a(X_{k\Delta}))\, ds }_{:=e_a^\top\left( {\Gamma}_{1,1}^{f,\sigma} -\hat{\Gamma}_{1,1}^{f,\sigma}\right) e_b}\,. 
\end{align*}
Applying Lemma \ref{diff-est-proof-lemma1} again, we obtain $e_a^\top \hat{\Gamma}_{1,1}^{f,\sigma} e_b=  O_p\big( \frac{\Delta   \sqrt{\Upsilon(T)}}{h^{d/2}}\big)$. We continue to control $e_a^\top\big( {\Gamma}_{1,1}^{f,\sigma} -\hat{\Gamma}_{1,1}^{f,\sigma}\big) e_b$. Note that this term is a kernel summation of products of a martingale term and a discrete-time approximation, which needs a treatment different from all those above. Recall that $\mathcal{F}_{\Delta k}$ is the sub $\sigma$-algebra of the filtration generated by the process $X_t$ up to time $\Delta k$. By the linearity of expectation and the tower property, we have 
 \begin{align*}
& \mathbb{E}^M_\lambda \left|e_a^\top \Big({\Gamma}_{1,1}^{f,\sigma} -\hat{\Gamma}_{1,1}^{f,\sigma}\Big)e_b \right|  \\
  \le\,&  \sum_{k=0}^{n-1}  
 \mathbb{E}^M_\lambda \bigg(\mathcal{K}_{h}(X_{k\Delta})\bigg|\int_{k\Delta}^{(k+1)\Delta} {\sigma}_{q^b, \alpha}(X_s) dW^\alpha_s    \int_{k\Delta}^{(k+1)\Delta} 
  (\zeta^a(X_s)- \zeta^a(X_{k\Delta})) \, ds  \bigg|\bigg)\\
=\,&  \sum_{k=0}^{n-1}  \mathbb{E}^M_\lambda \Bigg( \mathcal{K}_{h}(X_{k\Delta})   \mathbb{E}^M_\lambda\bigg( \bigg|\int_{k\Delta}^{(k+1)\Delta}   {\sigma}_{q^b, \alpha}(X_s) dW^\alpha_s \bigg|
  \bigg| \int_{k\Delta}^{(k+1)\Delta} (\zeta^a(X_s)- \zeta^a(X_{k\Delta})) \, ds  \bigg|  \Bigg| \mathcal{F}_{\Delta k}\bigg) \Bigg) \,.
 \end{align*}
Then, apply the Cauchy-Schwarz inequality to bound the inner expectation by
\begin{align}
&\Bigg(\underbrace{ \mathbb{E}^M_\lambda\bigg( \bigg| \int_{k\Delta}^{(k+1)\Delta} 
  (\zeta^a(X_s)-  \zeta^a(X_{k\Delta})) \, ds  \bigg|^2 \Bigg| \mathcal{F}_{\Delta k} \bigg)}_{:=\mathcal{I}_1}\Bigg)^{1/2} \nonumber \\
  &\quad\times \Bigg(\underbrace{ \mathbb{E}^M_\lambda \bigg(\bigg|\int_{k\Delta}^{(k+1)\Delta}   {\sigma}_{q^b, \alpha}(X_s) dW^\alpha_s \bigg|^2\Bigg| \mathcal{F}_{\Delta k} \bigg)}_{:=\mathcal{I}_2}\Bigg)^{1/2}\,. \nonumber 
\end{align}
First, we use the trivial bound 
\begin{align*}
\mathcal{I}_2 \le   \mathbb{E}^M_\lambda \bigg( \sup_{s \in [k\Delta, (k+1)\Delta]} \bigg|\int_{k\Delta}^{s}{\sigma}_{q^b, \alpha}(X_s) dW^\alpha_s \bigg|^2\Bigg| \mathcal{F}_{\Delta k} \bigg)  \le C_1 \mathbb{E}([M^{k}_{\Delta}])\,,
\end{align*}
where the second bound is via the BDG inequality in Theorem \ref{BDG-ineq} with $p = 1$ and 
\[
M^k_t:= \int_{{k\Delta}}^{k\Delta + t} \sigma_{q^b, \alpha}(X_s)dW^\alpha_{s}
\] 
for $t \in [0, \Delta]$. Then, since 
\[
[M^k_\Delta] = \int_{k\Delta}^{(k+1)\Delta} \sum_{\alpha=1}^r \sigma_{q^b,\alpha}(X_s)^2 ds \le C_2\Delta\,,
\] 
where $C_2>0$ is finite by the regularity assumption, we obtain $\mathcal{I}_2\le C_1 C_2 \Delta$.
With It\^o's formula \eqref{ito-integral},  
\[
\zeta^a(X_s) - \zeta^a(X_{k\Delta}) = \int_{k\Delta}^{s} \mu_{\zeta^a}(X_t) \, dt + \int_{k\Delta}^s  
\sigma_{\zeta^a,\alpha}(X_t) \, dW^\alpha_t \,,
\] 
when $s \in [k\Delta, (k+1)\Delta]$, the bound $(a+b)^2 \le 2(a^2 + b^2)$ gives the control of $\mathcal{I}_1$:
\begin{align*}
\mathcal{I}_1  \le  &\,
  2\underbrace{ \mathbb{E}^M_\lambda\bigg( \Delta   \int_{k\Delta}^{(k+1)\Delta} 
  \bigg( \int_{k\Delta}^{s}  \mu_{\zeta^a}(X_t)\, dt \bigg) ^2 \, ds\bigg| \mathcal{F}_{k\Delta} \bigg)}_{:=\mathcal{I}_{1,1}} \\
    &+ 2\underbrace{ \mathbb{E}^M_\lambda\bigg( \Delta   \int_{k\Delta}^{(k+1)\Delta} \bigg( \int_{k\Delta}^s  \sigma_{\zeta^a,\alpha}(X_t) \, dW^\alpha_t \bigg) ^2 \, ds\bigg| \mathcal{F}_{k\Delta} \bigg)}_{:=\mathcal{I}_{1,2}} \,
\end{align*}
By the Cauchy-Schwarz inequality again, 
\begin{align*}
 \mathcal{I}_{1,1} \le \,&  \mathbb{E}^M_\lambda\bigg( \Delta   \int_{k\Delta}^{(k+1)\Delta} 
   (s- k\Delta) \int_{k\Delta}^{s}  \mu_{\zeta^a}(X_t)
  ^2 \, dt \, ds\bigg| \mathcal{F}_{k\Delta} \bigg) \\   
  \le \,& \mathbb{E}^M_\lambda\bigg( \Delta^2   \int_{k\Delta}^{(k+1)\Delta} 
   \int_{k\Delta}^{(k+1)\Delta}  \mu_{\zeta^a}(X_t)^2 \, dt  \, ds\bigg| \mathcal{F}_{k\Delta} \bigg)
  \leq  C_3\Delta^4 \,,
\end{align*}
where $C_3>0$ is finite due to the regularity assumption. We turn to $\mathcal{I}_{1,2}$.  
By invoking the BDG inequality (Theorem \ref{BDG-ineq}) for the process $M^k_t:= \int_{k\Delta}^{k\Delta + t} \sigma_{\zeta^a,\alpha}(X_s) dW^\alpha_s$ and setting  
\begin{align*}
s^* \in { \text{argmax}_{s\in [k\Delta ,(k+1)\Delta)}} \ ( M^k_{s - k\Delta} )^2\,,
\end{align*}
we bound 
\begin{align*}
\mathcal{I}_{1,2}  \le  \,& \mathbb{E}^M_\lambda\bigg( \Delta   \int_{k\Delta}^{(k+1)\Delta} 
  ( M^k_{s^*-k\Delta} ) ^2 \, ds\bigg| \mathcal{F}_k \bigg)\\ 
  \le \,& C_1 \mathbb{E}^M_\lambda\bigg( \Delta   \int_{k\Delta}^{(k+1)\Delta} 
\Big( \int_{k\Delta}^{(k+1)\Delta} C_4\, dt\Big) \, ds\bigg| \mathcal{F}_k \bigg) = C_1C_4 \Delta^3\,,
\end{align*}
where the first bound is trivial by definition of $s^*$, and the second bound arises from the BDG inequality applied with $p=1$ and $|\mathcal{S}^{a}_\alpha(X_t)|^2 \le C_4$ for a finite constant $C_4>0$ by  the regularity assumption. 
We thus have $\mathcal{I}_1=O(\Delta^3)$. 
As a result, 
\begin{align*}
\left| \mathbb{E}^M_\lambda\left( e_a^\top \left({\Gamma}_{1,1}^{f,\sigma} -\hat{\Gamma}_{1,1}^{f,\sigma}\right)e_b \right) \right| 
  \le \,&
 C_1\sqrt{C_2C_4}  \sum_{k=0}^{n-1}  \mathbb{E}^M_\lambda \bigg( \mathcal{K}_{h}(X_{k \Delta}) \Delta^{2} \bigg) \\
= \,& C_1\sqrt{C_2C_4}    \Delta   \mathbb{E}^M_\lambda ({ \hat{L}}(x)) = O(\Delta   \Upsilon(T)) \,,
\end{align*}
where the last equality comes from Lemmas \ref{A1-analog} and \ref{A2-analog}, and hence
\begin{align}
 e_a^\top \big({\Gamma}_{1,1}^{f,\sigma} -\hat{\Gamma}_{1,1}^{f,\sigma}\big)e_b &= O_p(\Delta \Upsilon(T)) \,. \label{Gamma11-approx-eqn}
\end{align}
By the assumptions $\frac{h^d \Upsilon(T)}{\Delta} \xrightarrow{}\infty$ and $h^{d-4} \Upsilon(T) \Delta \xrightarrow[]{} 0$, we have $\frac{\Delta}{h^2} \xrightarrow{} 0$. Combining the analyses,  we conclude 
\begin{align*}
e_a^\top {\Gamma}_1 e_b= O_p(\Delta   \Upsilon(T)) + O_p\left( \frac{\Delta   \sqrt{\Upsilon(T)}}{h^{d/2}}\right) \,.
\end{align*} 
By assumptions $h^{d-4}\Upsilon(T)\Delta \to 0$ and $\Delta \xrightarrow[]{} 0$, with \eqref{reciprocal-lemma-nu(T)-hatL} we conclude
\begin{align*}
&\,\sqrt{\frac{h^d{ \hat{L}}(x)}{\Delta} }e_a^\top \hat{\Pi}_1 e_b = \sqrt
 {\frac{h^d}{{ \hat{L}}(x)\Delta}}e_a^\top {\Gamma}_1 e_b\\
=&\, O_p\left(\sqrt{h^d  \Upsilon(T)\Delta}\sqrt{\frac{\Upsilon(T)}{{ \hat{L}}(x)}} \right) + O_p\left( \Delta^{1/2}\sqrt{\frac{\Upsilon(T)}{{ \hat{L}}(x)}}\right) = o_p(1) \,. 
\end{align*}

It remains to analyze $\hat{\Pi}_2(x)$. By applying It\^o's formula to $q^b(X_s)-q^b(X_{k\Delta})$ and $f^b(X_s)-f^b(X_{k\Delta})$, we rewrite $\Gamma_2(x)=\Gamma_2^{q,\mu}(x)+\Gamma_2^{q,\sigma}(x)+\Gamma_2^{f,\mu}(x)+\Gamma_2^{f,\sigma}(x)$, where 
\begin{align*}
& e_a^\top {\Gamma}_2^{q,\mu} e_b:=  \sum_{k=0}^{n-1} \mathcal{K}_{h}(X_{k\Delta}) \int_{k\Delta}^{(k+1)\Delta} 
\left( \int_{k\Delta}^s  \mu_{q^b}(X_t) dt\right) \sigma_{f^a,\alpha}(X_s) dW^\alpha_s\,,
\\
&e_a^\top {\Gamma}_2^{q,\sigma} e_b :=  \sum_{k=0}^{n-1} \mathcal{K}_{h}(X_{k\Delta}) \int_{k\Delta}^{(k+1)\Delta} 
\left( \int_{k\Delta}^s  {\sigma}_{q^b, \beta}(X_t)dW^\beta _t\right) \sigma_{f^a,\alpha}({X_s}) dW^\alpha_s\,,
\end{align*}
and $\Gamma_2^{f,\mu}(x)$ and $\Gamma_2^{f,\sigma}(x)$ are defined similarly. As we did for ${\Gamma}_1$, we only analyze $\Gamma_2^{q,\mu}(x)$ and $\Gamma_2^{q,\sigma}(x)$ as the analysis of the reverse scenario is identical. 
By Lemma \ref{diff-est-proof-lemma1}, we have  
$e_a^\top{\Gamma}_2^{q,\mu} e_b = O_p\big( \frac{\Delta   \sqrt{\Upsilon(T)}}{h^{d/2}}\big) $. So,
\begin{align}
 \frac{h^{d/2}}{\sqrt{\hat{L}^{\texttt{(o)}}(x) \Delta}}e_a^\top{\Gamma}_2^{q,\mu} e_b &=  O_p\left(\sqrt{\frac{\Delta\Upsilon(T)}{{ \hat{L}}(x)}} \right) = o_p(1)\,, \label{normalized-order-gamma2mu}
\end{align}
where we use \eqref{reciprocal-lemma-nu(T)-hatL} and $\Delta\to 0$ in the last equality.
The analysis of ${\Gamma}_2^{q,\sigma}$, as a locally square integrable local martingale, is more involved due to the double stochastic integration. Denote 
\[
M_{b,k}(t):= \int_{k\Delta}^{k\Delta + t} \sigma_{q^b, \alpha}(X_s)dW^\alpha_{s}
\] 
for $t \in [0, \Delta)$. The quadratic variation of $e_a^\top{\Gamma}_2^{q,\sigma} e_b$ becomes
\begin{align*}
&\big[e_a^\top{\Gamma}_2^{q,\sigma} e_b\big]
=  \sum_{k=0}^{n-1} \mathcal{K}_{h}(X_{k\Delta})^2\int_{k\Delta}^{(k+1)\Delta} 
 M_{b,k}(t- k \Delta) ^2 \pi_{f}^{a,a}(x)(X_t) \, dt\,,
\end{align*}
where we use $\pi_{f}:=\pi_{f,f}$ to simplify the notation.
By the Doob-Meyer decomposition we have 
\begin{align*}
    M_{b,k}(s- k \Delta)^2 &= \int_{k\Delta}^s \pi_{q}^{b,b}(X_t)\, dt + 2\int_{k\Delta}^s \left( \int_{k\Delta}^t {\sigma}_{q^b, \gamma}(X_u) dW^\gamma_u \right){\sigma}_{q^b, \beta}(X_t)dW^\beta_t \,,
\end{align*}  
and hence
\begin{align*}
\big[&e_a^\top\Gamma_2^{q,\sigma} e_b\big]
= \underbrace{\sum_{k=0}^{n-1} \mathcal{K}_{h}(X_{k\Delta})^2\int_{k\Delta}^{(k+1)\Delta} \left( \int_{k\Delta}^s  \pi_{q}^{b,b}(X_t) dt\right)\pi_{f}^{a,a}(X_s) ds}_{:=e_a^\top \mathsf{A} e_b}
\\
&+\underbrace{\sum_{k=0}^{n-1} \mathcal{K}_{h}(X_{k\Delta})^2\int_{k\Delta}^{(k+1)\Delta}  \int_{k\Delta}^s  M_{b,k}(t - k \Delta)\sigma_{q^b, \beta}(X_t)dW^\beta_t  \pi_{f}^{a,a}(X_s) \, ds}_{:=e_a^\top \mathsf{Z}e_b}\,.
\end{align*}
By the stochastic Fubini Theorem (Lemma \ref{integration-by-parts}), $e_a^\top \mathsf{Z}e_b$ becomes
\begin{align*}
& \underbrace{2\sum_{k=0}^{n-1} \mathcal{K}_{h}(X_{k\Delta})^2 \int_{k\Delta}^{(k+1)\Delta} M_{b,k}(t - k\Delta) {\sigma}_{q^b, \beta}(X_t)dW^\beta_t \int_{k\Delta}^{(k+1)\Delta} \pi_{f}^{a,a}(X_s) \, ds   }_{e_a^\top\mathsf{Z}_1e_b}\\
&- \underbrace{2\sum_{k=0}^{n-1} \mathcal{K}_{h}(X_{k\Delta})^2\int_{k\Delta}^{(k+1)\Delta} M_{b,k}(s - k\Delta) {\sigma}_{q^b, \beta}(X_s)  \int_{k\Delta}^s  \pi_{f}^{a,a}(X_t)dt dW_s^\beta}_{e_a^\top \mathsf{Z}_2e_b}.
\end{align*}
The quadratic variation of $e_a^\top \mathsf{Z}_2e_b$ can be directly controlled by
\begin{align*}
&\left[\frac{h^d}{\Delta }e_a^\top \mathsf{Z}_2e_b\right]  \\ 
=\,& \frac{4h^{2d}}{\Delta^2}\sum_{k=0}^{n-1} \mathcal{K}_{h}(X_{k\Delta})^4\int_{k\Delta}^{(k+1)\Delta} 
  \left(\int_{k\Delta}^s  \pi_{f}^{a,a}(X_t) dt\right)^2 M_{b,k}(s - k\Delta)^2  \pi_{q}^{b,b}(X_s)
 \,ds  \\
\le \, 
  &  4C^3h^{2d}\sum_{k=0}^{n-1} \mathcal{K}_{h}(X_{k\Delta})^4\int_{k\Delta}^{(k+1)\Delta}  M_{b,k}(s - k\Delta)^2  ds \,,
\end{align*} 
where in the last inequality we use $\max\{| \pi_{f}^{a,b}(X_t) | ,\,| \pi_{q}^{a,b}(X_t) | \}\le C$, where $C>0$ is finite by the regularity assumptions on $f,q, \mu$ and $\sigma$.
We now apply the BDG inequality similar to the control of $\mathcal{I}_{1,2}$ above, and obtain
\begin{align*}
& \mathbb{E}^M_\lambda \left(\left[\frac{h^d}{\Delta }e_a^\top \mathsf{Z}_2 e_b\right]\right)  \\
\le \,& 
\frac{4C^3}{h^{d}}  \mathbb{E}^M_\lambda\left(\sum_{k=0}^{n-1} \frac{1}{h^{d}}\mathcal{K}\left(\frac{{\mathcal{D}}_x(X_{k \Delta}) }{h} \right)^4\int_{k\Delta}^{(k+1)\Delta} 
  \mathbb{E}^M_\lambda \big( M_{b,k}(s^* - k\Delta)^2 \big| \mathcal{F}_k\big) ds \right) \\
 \le \,& \frac{4C^3 C_1   \Delta }{h^{d}} \mathbb{E}^M_\lambda\left(\sum_{k=0}^{n-1}\frac{1}{h^{d}}\mathcal{K}\left(\frac{{\mathcal{D}}_x(X_{k \Delta}) }{h} \right)^4\int_{k\Delta}^{(k+1)\Delta} \pi_{q}^{b,b}(X_u) du \right)  \,,
\end{align*}
where the term inside the expectation is the same as $[\Theta^{q^b}(T)]$, where $\Theta^{q^b}(T)$ is defined in \eqref{theta_t_definition} with $\mathcal{K}^2$. By Lemma \ref{Theta_T-analysis-lemma}, $[\Theta^{q^b}(T)]=O(\Upsilon(T))$.
Therefore, we have
$\mathbb{E}^M_\lambda \left[\frac{h^d}{\Delta }e_a^\top \mathsf{Z}_2 e_b\right]= O\left(\frac{\Delta   \Upsilon(T)}{h^d}\right)$, and hence 
\begin{align}
\frac{h^d}{\Delta }e_a^\top \mathsf{Z}_2 e_b  &= O_p\left( \frac{\sqrt{\Delta\Upsilon(T)}}{h^{d/2}}\right)  \label{gamma_2_sigma_Z_2}\,.
\end{align}
Next, we control $e_a^\top \mathsf{Z}_1 e_b$ by approximation at uniform discrete times:
\begin{align*}
 e_a^\top \mathsf{Z}_1e_b 
= \,& \underbrace{2\sum_{k=0}^{n-1} \mathcal{K}_{h}(X_{k\Delta})^2 \Delta \pi_{f}^{a,a}(X_{k\Delta}) \int_{k\Delta}^{(k+1)\Delta} \left( \int_{k\Delta}^t {\sigma}_{g^b,\gamma}(X_u) dW^\gamma_u \right){\sigma}_{g^b, \beta}(X_t) d{W}^\beta_t}_{e_a^\top \hat{\mathsf{Z}}_1e_b}\\
&\,+ 2\sum_{k=0}^{n-1} \mathcal{K}_{h}(X_{k\Delta})^2\int_{k\Delta}^{(k+1)\Delta}
\left(\pi_{f}^{a,a}(X_s)-\pi_{f}^{a,a}(X_{k\Delta}) \right) ds  \\
&\qquad \underbrace{\qquad\times \int_{k\Delta}^{(k+1)\Delta} \left( \int_{k\Delta}^t {\sigma}_{g^b, \gamma}(X_u) dW^\gamma_u \right){\sigma}_{g^b, \beta}(Y_t) dW^\beta_t}_{e_a^\top (\mathsf{Z}_1 -  \hat{\mathsf{Z}}_1)e_b} \,.
\end{align*}
Just as for the computations leading to \eqref{gamma_2_sigma_Z_2}, we can compute $\frac{h^d}{\Delta }e_a^\top \hat{\mathsf{Z}}_1e_b = O_p\big( \frac{\sqrt{\Delta \Upsilon(T)}}{h^{d/2}}\big)$.
The analysis of $e_a^\top (\mathsf{Z}_1 -  \hat{\mathsf{Z}}_1)e_b$ is similar to that of $e_a^\top( {\Gamma}_{1,1}^{f,\sigma} -\hat{\Gamma}_{1,1}^{f,\sigma})e_b$ above in \eqref{Gamma11-approx-eqn}, except the quadratic variation of the former contains one less $\Delta$ factor. Hence, we achieve $e_a^\top (\mathsf{Z}_1 -  \hat{\mathsf{Z}}_1)e_b=O_p(\Delta^{1/2}  \Upsilon(T))$. We thus have
\begin{align}
   \frac{h^d}{\Delta} e_a^\top \mathsf{Z}e_b &= O_p\left( \frac{\sqrt{\Delta \Upsilon(T)}}{h^{d/2}}\right) + O_p(\Delta^{1/2}\Upsilon(T))\,.  \label{equation final Z order in Gammaq sigma_2}
\end{align}
We proceed to examine $\mathsf{A}$. 
By again approximating the process at uniform discrete times, 
\begin{align*}
 \pi_{q}^{b,b}(X_t)  &
\pi_{f}^{a,a}(X_s) 
  =   \left( \pi_{q}^{b,b}(X_t) -  \pi_{q}^{b,b}(X_{k \Delta}) \right)\pi_{f}^{a,a}(X_s)  \\
  &+  
\pi_{q}^{b,b}(X_{k\Delta})\left(\pi_{f}^{a,a}(X_s) - \pi_{f}^{a,a}(X_{k \Delta})\right) + 
  \pi_{q}^{b,b}(X_{k\Delta}) \pi_{f}^{a,a}(X_{k \Delta})\,,
\end{align*}
$e_a^\top \mathsf{A}e_b$ is decomposed into
\begin{align*}
& \underbrace{\sum_{k=0}^{n-1} \mathcal{K}_{h}(X_{k \Delta})^2 \int_{k\Delta}^{(k+1)\Delta} 
  \int_{k\Delta}^s \left(\pi_{q}^{b,b}(X_t) -  \pi_{q}^{b,b}(X_{k \Delta})\right)\pi_{f}^{a,a}(X_s)\,  dt ds}_{:=\mathsf{A}_1} \\
  +\,& \underbrace{\sum_{k=0}^{n-1} \mathcal{K}_{h}(X_{k \Delta})^2 \int_{k\Delta}^{(k+1)\Delta} 
  \int_{k\Delta}^s  \pi_{q}^{b,b}(X_{k\Delta})\left(\pi_{f}^{a,a}(X_s) - \pi_{f}^{a,a}(X_{k \Delta})\right)\,  dt ds}_{:=\mathsf{A}_2}\\
  +\,& \underbrace{\sum_{k=0}^{n-1} \mathcal{K}_{h}(X_{k \Delta})^2 \int_{k\Delta}^{(k+1)\Delta} 
  \int_{k\Delta}^s  \pi_{q}^{b,b}(X_{k\Delta}) \pi_{f}^{a,a}(X_{k \Delta})\,  dt ds}_{:=\mathsf{A}^q_3}\,.
\end{align*}
By the same application of It\^o's Lemma (e.g., the treatment of $[Z_t] - \widehat{[Z_t]}$ in Lemma  \ref{discrete-approximation-martingale-cross-quad-var}), we obtain 
\begin{align*}
\frac{h^d}{\Delta \Upsilon(T)}\mathsf{A}_1&\,= O_p(\Delta^{1/2} ) + O_p\Big( \sqrt{\frac{\Delta}{h^d \Upsilon(T)}} \Big)\,, \\
\frac{h^d}{\Delta \Upsilon(T)}\mathsf{A}_2&\,= O_p(\Delta^{1/2}  ) + O_p\Big( \sqrt{\frac{\Delta}{h^d \Upsilon(T)}} \Big)\,.
\end{align*}
We thus have
\begin{align*}
   \frac{h^d}{\Delta \Upsilon(T)} e_a^\top (\mathsf{A}_1+\mathsf{A}_2)e_b &= O_p(\Delta^{1/2}  ) + O_p\Big( \sqrt{\frac{\Delta}{h^d \Upsilon(T)}} \Big)
\end{align*}
by assumption. 
The nontrivial term in the analysis of $\hat{\Pi}_2$ is $\mathsf{A}^q_3$. Note that when we analyze $\Gamma^{f,\sigma}_2$, there exists a similar term, denoted likewise as $\mathsf{A}^f_3$, which is also nontrivial. We have
\begin{align*}
   &\frac{h^d}{\Delta \Upsilon(T)}(\mathsf{A}^q_3+\mathsf{A}^f_3) \\
   =&\,  \frac{\Delta}{ h^d \Upsilon(T)}\sum_{k=0}^{n-1} \mathcal{K}\left(\frac{{\mathcal{D}}_x(X_{k \Delta}) }{h} \right)^2 \frac{1}{2} (\pi_{q}^{b,b}(X_{k\Delta}) \pi_{f}^{a,a}(X_{k \Delta})+\pi_{q}^{b,b}(X_{k\Delta}) \pi_{q}^{a,a}(X_{k \Delta})) \\
  \xrightarrow[]{\hspace{0.1cm}d\hspace{0.1cm}} &\,\frac{1}{2}   \kappa_{2,0}   g_\alpha(1)p_X(x) (\pi_{q}^{b,b}(x) \pi_{f}^{a,a}(x)+\pi_{q}^{b,b}(x) \pi_{q}^{a,a}(x)) \,,
\end{align*}
where the convergence comes from Lemmas \ref{A1-analog} and \ref{A2-analog}.
With \eqref{equation final Z order in Gammaq sigma_2},
\begin{align*}
 \,& \frac{  \big[\frac{h^{d/2}}{\Delta^{1/2}} e_a^\top(\Gamma_2^{q,\sigma}+\Gamma_2^{f,\sigma}) e_b\big]}{\Upsilon(T)} =\frac{h^d}{\Delta \Upsilon(T)}(\mathsf{A}^q_3+\mathsf{A}^f_3)+ o_p(1)\\
\,&\qquad\qquad  \xrightarrow[]{\hspace{0.1cm}d\hspace{0.1cm}}    \kappa_{2,0}   g_\alpha(1)p_X(x) \frac{1}{2}(\pi_{q}^{b,b}(x) \pi_{f}^{a,a}(x)+\pi_{q}^{b,b}(x) \pi_{q}^{a,a}(x))\,.
\end{align*}

To obtain the desired weak convergence to the normal distribution, we need to evaluate cross quadratic variation. We sketch the key steps here without details since they are the same as the above quadratic variation terms. The cross quadratic variation of $\big[e_a^\top{\Gamma}_2^{q,\sigma} e_b,\, e_c^\top{\Gamma}_2^{q,\sigma} e_d\big]$ is
\begin{align*}
 \sum_{k=0}^{n-1} \mathcal{K}_{h}(X_{k \Delta})^2\int_{k\Delta}^{(k+1)\Delta} 
 M_{b,k}(t- k \Delta)M_{d,k}(t- k \Delta) \pi_{f}^{a,c}(x)(X_t) \, dt\,.
\end{align*}
With the polarization $M_{b,k}M_{d,k}=\frac{1}{4}((M_{b,k}+M_{d,k})^2-(M_{b,k}-M_{d,k})^2)$, we can apply the Doob-Meyer decomposition. The associated $\mathsf{A}$ part of the term involving $(M_{b,k}+M_{d,k})^2$ is
\[
 \sum_{k=0}^{n-1} \mathcal{K}_{h}(X_{k \Delta})^2\int_{k\Delta}^{(k+1)\Delta} \left( \int_{k\Delta}^s  (\pi_{q}^{b,b}+\pi_{q}^{d,d}+2\pi_{q}^{b,d})(X_t) dt\right)\pi_{f}^{a,c}(X_s) ds
\]
and hence the associated $\frac{h^d}{\Delta \Upsilon(T)}(\mathsf{A}^q_3+\mathsf{A}^f_3)$ part of the term involving $M_{b,k}M_{d,k}$ is
\begin{align*}
 &\frac{\Delta}{ h^d \Upsilon(T)}\sum_{k=0}^{n-1} \mathcal{K}\left(\frac{{\mathcal{D}}_x(X_{k \Delta}) }{h} \right)^2 \frac{1}{2} (\pi_{q}^{b,d}(X_{k\Delta}) \pi_{f}^{a,c}(X_{k \Delta})+\pi_{f}^{b,d}(X_{k\Delta}) \pi_{q}^{a,c}(X_{k \Delta})) \\
\,&\qquad\qquad  \xrightarrow[]{\hspace{0.1cm}d\hspace{0.1cm}}   \kappa_{2,0}   g_\alpha(1)p_X(x) \frac{1}{2}(\pi_{q}^{b,d}(x) \pi_{f}^{a,c}(x)+\pi_{f}^{b,d}(x) \pi_{q}^{a,c}(x))\,.
\end{align*}
Recall the index routine that for a random matrix $M\in \mathbb{R}^{p\times p}$ with $\texttt{cov}(\text{vec}(M))=A\otimes B\in \mathbb{R}^{p^2\times p^2}$, where $A$ and $B$ are non-negative definite, the $(i,j)$-th entry of $\texttt{cov}(\text{vec}(M))$ is $A_{a,b}B_{c,d}$, where $i=(b-1)p+a$ and $j=(d-1)+c$. With these preparation, by Proposition \ref{thm3.16_limit_theorems_null} and the same argument as \eqref{Remark 4.26 and Chapter 7 limit_theorems_null} in Theorem \ref{generalized-drift-est} with $\Theta^f$ replaced by $\frac{h^d}{\Delta } e_a^\top\Gamma_2^{q,\sigma} e_b$ and $\mathcal{K}$ replaced by $\mathcal{K}^2$, we have
\[
\frac{\frac{h^{d/2}}{\Delta^{1/2}}\text{vec}(\Gamma^{q,\sigma}_2+\Gamma^{f,\sigma}_2)}{\sqrt{{ \hat{L}}(x)}} \xrightarrow[]{\hspace{0.1cm}d\hspace{0.1cm}}  N\left(\boldsymbol{0},{\frac{\kappa_{2,0}}{2}(\pi_{f,f}(x) \otimes \pi_{q,q}(x)+\pi_{q,q}(x) \otimes \pi_{f,f}(x))}\right)\,.
\]
With \eqref{normalized-order-gamma2mu} and Slutsky's theorem, we obtain
    \begin{align*}
&\sqrt{\frac{h^d{ \hat{L}}(x)}{\Delta}} \hat{\Pi}_2(x) =  \frac{h^{d/2}}{\sqrt{{ \hat{L}}(x) \Delta}} (\Gamma^{q,\sigma}_2+\Gamma^{f,\sigma}_2)  + \frac{h^{d/2}}{\sqrt{{ \hat{L}}(x) \Delta}} (\Gamma^{q,\mu}_2+\Gamma^{f,\mu}_2) \\ 
=\,& \frac{h^{d/2}}{\sqrt{{ \hat{L}}(x) \Delta}} (\Gamma^{q,\sigma}_2+\Gamma^{f,\sigma}_2)  + o_p(1) \xrightarrow[]{\hspace{0.1cm}d\hspace{0.1cm}}  N\left(\boldsymbol{0}, \frac{\kappa_{2,0}}{2}(\pi_{f,f}(x) \otimes \pi_{q,q}(x)+\pi_{q,q}(x) \otimes \pi_{f,f}(x))\right) \,.
\end{align*}

By putting all the above together and introducing $\text{vech}$, Slutsky's theorem gives the desired convergence:
\begin{align*}
&\sqrt{\frac{h^d   { \hat{L}}(x)}{\Delta}} \left( \hat{\pi}_{f,g}(x) - \pi_{f,g}(x)   - h^2{B}^{\pi}_{f,g}(x)\right) \\
=\,&  \sqrt{\frac{h^d   { \hat{L}}(x)}{\Delta}} \left[ \hat{\Pi}_1(x) + \hat{\Pi}_2(x) +(\hat{\Pi}_3(x) -  {\pi}_{f,g}(x)   - h^2 {B}^{\pi}_{f,g}(x))\right] \\
=\,& \sqrt{\frac{h^d   { \hat{L}}(x)}{\Delta}} \hat{\Pi}_2(x) + o_p(1)  \xrightarrow[]{\hspace{0.1cm}d\hspace{0.1cm}} N\left(\boldsymbol{0}, \kappa_{2,0}\Xi_{f,g}(x)\right) \,.
\end{align*}
\end{proof}

To analyze the drift estimators, we require the following corollary, which will also be used in establishing tangent space estimation. Note that the relationship among $h$, $\Upsilon(T)$, and $\Delta$ differs from that in Theorem \ref{generalized-diff-est} by a factor of $\Delta$.

\begin{corollary}
\label{diffusion-estimator-at-drift-scale}
Suppose the assumptions in Theorem \ref{generalized-diff-est} hold, but replace assumptions $\frac{h^d  \Upsilon(T)}{\Delta} \xrightarrow{}\infty$,  $\frac{h^{d+4} \Upsilon(T)}{\Delta} \xrightarrow[]{} C$, $C>0$,  and $h^{d-4}  \Upsilon(T) \Delta \xrightarrow[]{} 0$ by  ${h^d \Upsilon(T)}\xrightarrow{}\infty$,  ${h^{d+4}   \Upsilon(T)}\xrightarrow[]{} C$ for $C>0$, and $h^{d-4} \Upsilon(T) {\Delta^2} \xrightarrow[]{} 0$. Then, 
\begin{align*}
\sqrt{{h^d   \hat{L}^{\texttt{(o)}}(x)}} \text{vech } \left( \hat{{\pi}}_{f,g}(x) - {\pi}_{f,g}(x)   - h^2 {B}^{\pi}_{f,g}(x)\right)
\xrightarrow[]{\hspace{0.1cm}p\hspace{0.1cm}}  \mathbf{0} \in \mathbb{R}^p\,.
\end{align*}
\end{corollary}

\begin{proof}  
Since the assumptions ${h^d  \Upsilon(T)}\xrightarrow{}\infty$,  ${h^{d+4}  \Upsilon(T)}\xrightarrow[]{} C,$ and $h^{d-4} \Upsilon(T) {\Delta^2} \xrightarrow[]{} 0$ differ from those in Theorem \ref{generalized-diff-est} by a factor of ${\Delta}$, we scaled up $\hat{{\pi}}_{f,g}(x) - {\pi}_{f,g}(x)   - h^2 {B}^{\pi}_{f,g}(x)$ by a factor of $\sqrt{\Delta}$. The error analysis is identical to that in the proof of Theorem \ref{generalized-diff-est}, while we now have the dominant term $\sqrt{h^d\hat{L}^{\texttt{(o)}}(x)} \hat{\Pi}_2(x) = o_p(1)$ since $\Delta = o(1)$. Thus, we obtain the convergence to $\mathbf{0}$ in probability.
\end{proof}

\section{Proofs of Main Theorems}

\subsection{Proof of Theorem \ref{euclidean-diffusion-estimate-total} about diffusion estimator}\label{section proof obs diffusion proof}

The proof follows immediately by an application of Theorem \ref{generalized-diff-est} to the map $f=q:=\iota:M \rightarrow \mathbb{R}^p$. By definition, $\hat{\pi}^{(\texttt{o})}(x)=\hat{\pi}_{\iota,\iota}(x)$. 
Next, note that $\sigma_{\iota,l}(x)$ defined in \eqref{ito-integral} is a directional derivative of $\iota$ so that $\sigma_{\iota,l}=\sigma_l\iota=\sigma^k_l\partial_k\iota^je_j$, where $e_j$ is the standard orthonormal basis of $\mathbb{R}^q$, and the pushforwarded vector satisfies ${\iota}_*\sigma_l=\sigma_l^k\partial_k\iota^je_j$, so we have $\sigma_{\iota,l}(x)=\iota_*\sigma_l(\iota(x))$. This leads to $ {\pi}_{\iota,\iota}(x) \otimes {\pi}_{\iota,\iota}(x)  = \pi^{(\texttt{o})}(x) \otimes \pi^{(\texttt{o})}(x) $ and hence ${\Xi}_{\iota,\iota}(x) =  \Xi^{(\texttt{o})}(x)$. Therefore, we have ${\pi}^{(\texttt{o})}(x)={\pi}_{\iota,\iota}(x)$. 
Also, by definition in \eqref{B_pi_varphi}, $b^{(\texttt{o})}_{\pi}(x)=B^{\pi}_{f,f}(x)$. 
With all the above, Theorem \ref{generalized-diff-est} gives the desired result.

\subsection{Proof of Theorem \ref{diff-est-tangent-space-proj} about tangent space estimation}
\label{Tan-space-est-and-proj}

Fix $x\in M$. We will apply Corollary \ref{diffusion-estimator-at-drift-scale}  to the map $f=q:=\iota:M \rightarrow \mathbb{R}^p$. In this setup, as discussed in Section \ref{section proof obs diffusion proof}, $\hat{\pi}^{(\texttt{o})}(x)=\hat{\pi}_{\iota,\iota}(x)$ and ${\pi}^{(\texttt{o})}(x)={\pi}_{\iota,\iota}(x)$.

Choose ambient coordinates around $\iota(x)$ so that $\pi^{(\texttt{o})}(x)$ is nonzero only on the left upper $d\times d$ matrix. Denote the eigenvalue decomposition $\pi^{(\texttt{o})}(x)=UDU^\top$, where $U\in O(p)$, $D=\texttt{diag}(\lambda_1,\ldots,\lambda_d,0,\ldots,0)\in \mathbb{R}^{p\times p}$ and $\lambda_1\geq \ldots \geq\lambda_d$. By construction, the first $d$ eigenvectors of $\pi^{(\texttt{o})}(x)$, denoted as, form an orthonormal basic of $\iota_*T_xM$.
By Corollary \ref{diffusion-estimator-at-drift-scale}, we have  
\begin{align*}
\hat\pi^{(\texttt{o})}(x)= \pi^{(\texttt{o})}(x) + h^2 b_\pi^{(\texttt{o})}(x) \in \mathbb{R}^{p \times p}\,,
\end{align*}
where $b_\pi^{(\texttt{o})}(x)\in \mathbb{R}^{p \times p}$ is a symmetric matrix corresponding to the bias of the diffusion estimator at $x$, with $b_\pi^{(\texttt{o})}(x) = O(1) + o_p\Big( \frac{1}{\sqrt{h^{d+4}\hat{L}^{(\texttt{o})}(x)}}\Big)$. Note that  by the assumption $h^{d+4} \Upsilon(x) \xrightarrow[]{\hspace{0.1cm}\hspace{0.1cm}} C$ and \eqref{reciprocal-lemma-nu(T)-hatL}, asymptotically $h^{d+4}\hat{L}^{(\texttt{o})}(x)$ converges in distribution to a random variable when $T\to \infty$. Thus,  $o_p\Big( \frac{1}{\sqrt{h^{d+4}\hat{L}^{(\texttt{o})}(x)}}\Big)=o_p(1)$ by Slutsky's theorem. Also, we can replace $o_p\Big( \frac{1}{\sqrt{h^{d+4}\hat{L}^{(\texttt{o})}(x)}}\Big)$ by $o_p\Big( \frac{1}{\sqrt{h^{d+4}\Upsilon(T)}}\Big)$. Clearly, $b_\pi^{(\texttt{o})}(x) = O_p(1)$.

Denote the eigendecomposition of  $\hat\pi^{(\texttt{o})}(x)= \hat U\hat D\hat U^\top$,
where $ \hat U\in O(p)$ (the space of $p \times p$ orthogonal matrices) and $ \hat D$ are diagonal with eigenvalues ordered non-increasingly. 
Denote $\hat{U}_{d}$ and $U_d$ to be the dominant $d$ columns of $\hat U$ and $U$.
Using the perturbation technique \cite{van2007computation}, asymptotically when $h\to 0$, we have 
\begin{align*}
\hat{P}_x := \hat U_{d}\hat U_{d}^\top = U_dU_d^\top + h^2b^{(\texttt{tan})}(x) = P_x + h^2b^{(\texttt{tan})}(x)
\end{align*}
for some ${p\times p}$ symmetric matrix $b^{(\texttt{tan})}(x)$, where $b^{(\texttt{tan})}(x)= O(1)+o_p\Big( \frac{1}{\sqrt{h^{d+4}\Upsilon(T)}}\Big)$. We thus conclude the theorem.

\subsection{Proof of Theorem \ref{main theorem drift} about drift estimator}\label{section proof of theorem of obs drift}

\begin{proof}
Analogous to proving Theorem \ref{euclidean-diffusion-estimate-total} by applying Theorem \ref{generalized-diff-est}, we prove Theorem \ref{main theorem drift} by applying Theorem \ref{generalized-drift-est}, with the additional consideration of the projection $\hat{P}_x$ of \eqref{definition Py Euclidean}.
With $f=\iota : M \rightarrow \mathbb{R}^p$, $\mathcal{D}_x(x'):=\|\iota(x)-\iota(x')\|_{\mathbb{R}^p}$, { we} obtain
\begin{align}
    \sqrt{{h^d   \hat{L}^{(\texttt{o})}(x)}}  \left( \hat{\mu}_\iota (x) - \mu_\iota(x) - h^2B^{\mu,\texttt{o}}_\iota(x) \right)
\xrightarrow[]{\hspace{0.1cm}d\hspace{0.1cm}}  N(\mathbf{0}, \,  \kappa_{2,0} \pi^{(\texttt{o})}(x)) \label{observed-drift-estimate-statement}
\end{align}
by Theorem \ref{generalized-drift-est}, where $B^{\mu,\texttt{o}}_\iota$ is defined in \eqref{drift-bias}. 
Following the same argument as that in proving Theorem \ref{euclidean-diffusion-estimate-total}, we have $\mu_\iota(x):=\iota_*\mu(\iota(x))$, and hence ${P}_x\mu_\iota(x)=\mu^{(\texttt{o})}(x)$. Recall that $\hat{\mu}^{(\texttt{o})}(x)=\hat{P}_x\hat{\mu}_\iota (x)$.  However, the targeting bias term $b_\mu^{(\texttt{o})}(x)\neq P_xB^{\mu,\texttt{o}}_\iota(x)$, but rather $b_\mu^{(\texttt{o})}(x) = P_x B^{\mu,\texttt{o}}_\iota(x) - \overline{b} ^{(\texttt{tan})}(x) \mu_\iota(x)$  \eqref{mu-varphi-bias}, as the drift estimator bias also contains the bias term induced by the tangent space estimation.

By Theorem \ref{diff-est-tangent-space-proj}, we have $\hat{P}_x = P_x + h^2b^{(\texttt{tan})}(x)$, where $b^{(\texttt{tan})}(x) = \overline{b} ^{(\texttt{tan})}(x) +  \epsilon^{(\texttt{tan})}(x)$ with $\overline{b} ^{(\texttt{tan})}(x)=O(1)$ and $\epsilon^{(\texttt{tan})}(x) = o_p\Big( \frac{1}{\sqrt{h^{d+4} \Upsilon(T)}}\Big)$. Putting all together, write 
\begin{align}
&   \hat{\mu}^{(\texttt{o})}(x) - \mu^{(\texttt{o})}(x) - h^2b^{(\texttt{o})}_{\mu}(x) \nonumber  \\  
= \,&   P_x(\hat{\mu}_\iota (x)-\mu_\iota(x) - h^2{B}_\iota^{\mu}(x)) + (P_x - \hat{P}_x)\hat{\mu}_\iota(x) + h^2 \overline{b} ^{(\texttt{tan})}(x)\mu_\iota(x)  \nonumber \\
   =\, &  P_x (\mu_\iota(x) - \hat{\mu}_\iota (x) - h^2{B}_\iota^{\mu}(x) ) \label{decomp-of-ob-drift-I}\\
   &- h^2b^{(\texttt{tan})}(x)(\hat{\mu}_\iota(x)-\mu_\iota(x)) -  h^2 \epsilon^{(\texttt{tan})}\mu_\iota(x) \label{decomp-of-ob-drift-II}\,.
\end{align} 
By \eqref{observed-drift-estimate-statement} and the continuous mapping theorem, \eqref{decomp-of-ob-drift-I} after normalization becomes
\begin{align}
     \sqrt{{h^d   \hat{L}^{(\texttt{o})}(x)}} P_x (\mu_\iota(x) - \hat{\mu}_\iota (x) - h^2B_\iota^{\mu}(x) )  \xrightarrow[]{\hspace{0.1cm}d\hspace{0.1cm}} N( \mathbf{0}, \,\kappa_{2,0}\pi^{(\texttt{o})}(x))\,, \nonumber 
\end{align}
where we use the fact that $P_x\pi^{(\texttt{o})}(x)P_x^\top=\pi^{(\texttt{o})}(x)$.
The first term of \eqref{decomp-of-ob-drift-II} after normalization becomes
\begin{align}
    \sqrt{h^{d+4}  \hat{L}^{(\texttt{o})}(x)} b^{(\texttt{tan})}(x)(\hat{\mu}_\iota(x)-\mu_\iota(x))=o_p(1)\nonumber
\end{align}
by using \eqref{observed-drift-estimate-statement} and the fact that $\overline{b} ^{(\texttt{tan})}(x)=O(1)$.
The second term of \eqref{decomp-of-ob-drift-II} after normalization becomes
\begin{align*}
    \sqrt{h^{d+4}  \hat{L}^{(\texttt{o})}(x)}\epsilon^{(\texttt{tan})}(x){\mu}_\iota (x)  = o_p\left(\sqrt{\frac{\hat{L}^{(\texttt{o})}(x)}{\Upsilon(T)}}\right)=o_p(1)\,,
\end{align*} 
where the last equality comes from \eqref{reciprocal-lemma-nu(T)-hatL}. By combining all the above controls, we obtain the desired conclusion.
 
\end{proof}

\section{More details on numerical simulation}\label{section more numerical simulations}

\begin{figure}[hbt!]
    \centering
\includegraphics[trim=10 80 0 80, clip, width=0.9\textwidth]{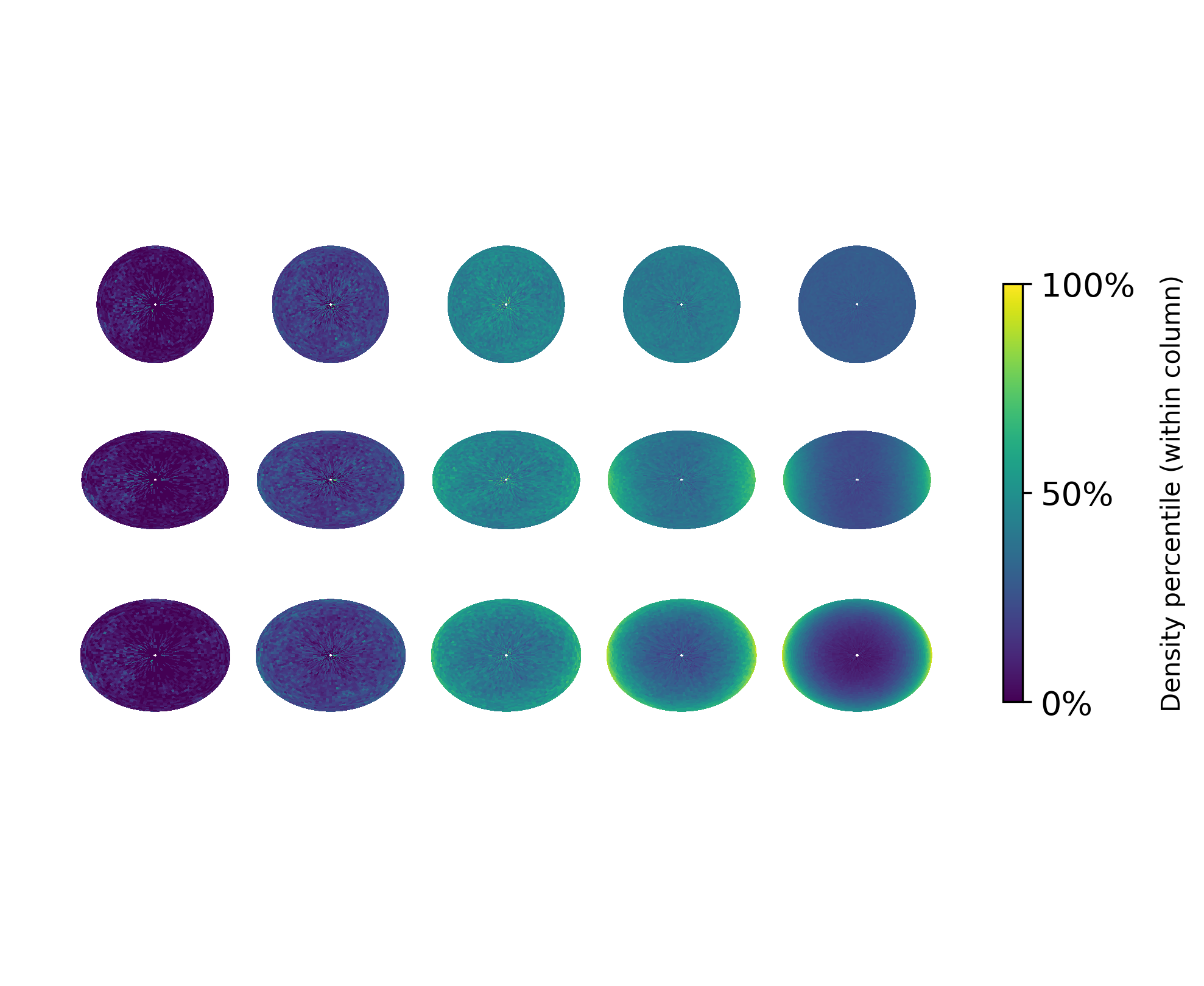}
    \caption{Top to bottom row: $\hat{L}^{(\texttt{o})}$ on ellipsoids with eccentricities $(1,1,1)$, $(1.5,1,1)$, and $(2,1.5,1)$ evaluated from a single long trajectory with different $n_i$ and $\Delta = 10^{-2}$, where $\log_{10}(n_i) \in \{4, 5,6,7,8\}$, from left to right column. Bins are colored by the percentile of their density across all eccentricities (within each column) for a fixed $T$. This enables meaningful visual comparison across plots sharing the same $T$ value.}
    \label{fig:ellipsoid_density}
\end{figure}

\begin{figure}[hbt!]
    \centering
    \includegraphics[width=0.85\textwidth]{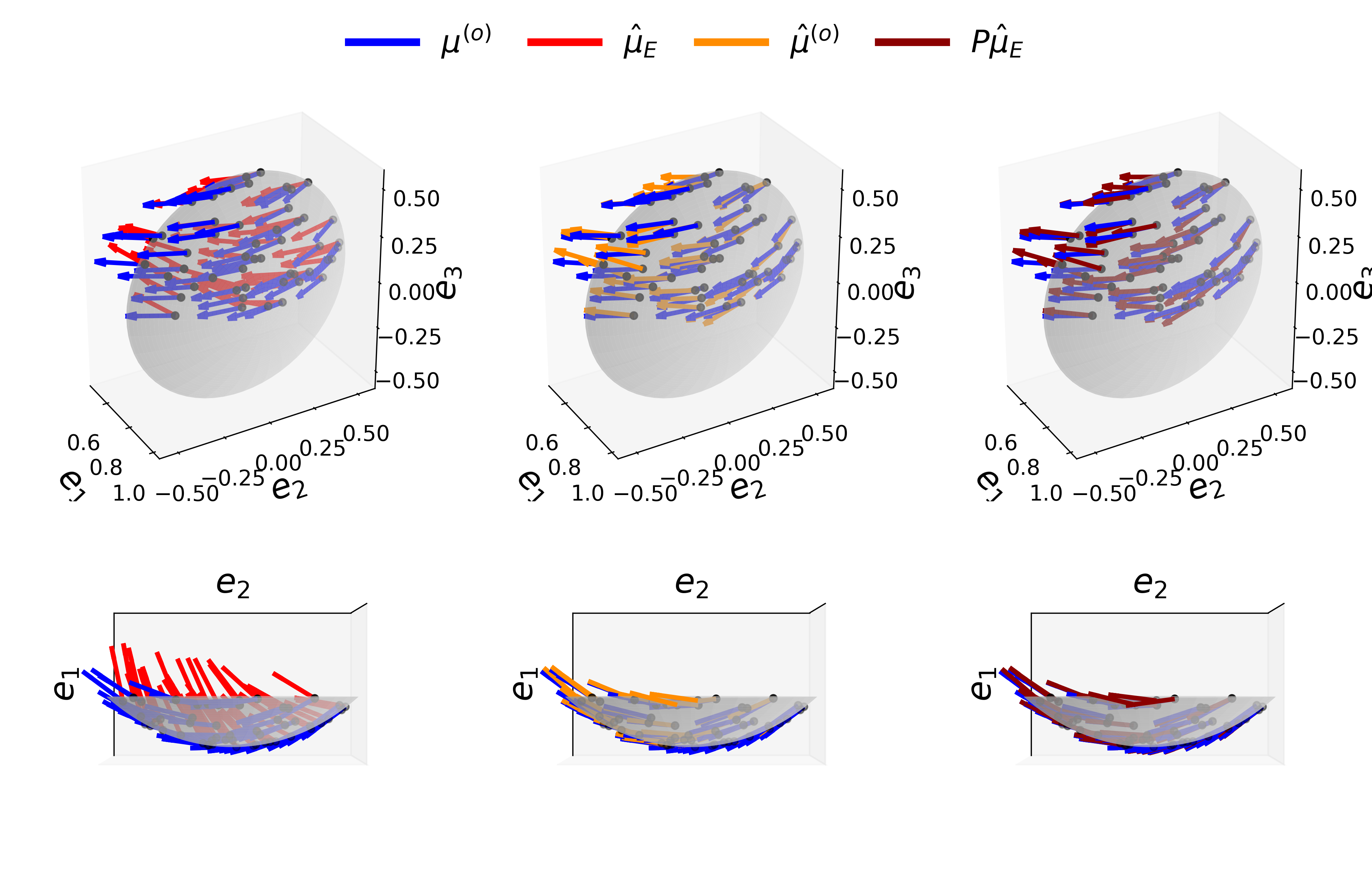}
    \caption{From left to right: visualizations of $\hat{\mu}_E(x), \hat{\mu}^{(\texttt{o})}(x),$ and $P_x \hat{\mu}_E(x)$, where $P_x$ is the projection operator onto the tangent space $T_xM$, for base-point samples $x$ drawn uniformly from a spherical cap centered at $(1,0,0)^\top$ and observed on ellipsoids with eccentricity $(1,1,1)$, shown from two viewing angles. The ground-truth drift vector is superimposed as blue arrows.}
    \label{fig:IB1-11}
\end{figure}

\begin{figure}[hbt!]
    \centering
    \includegraphics[width=0.85\textwidth]{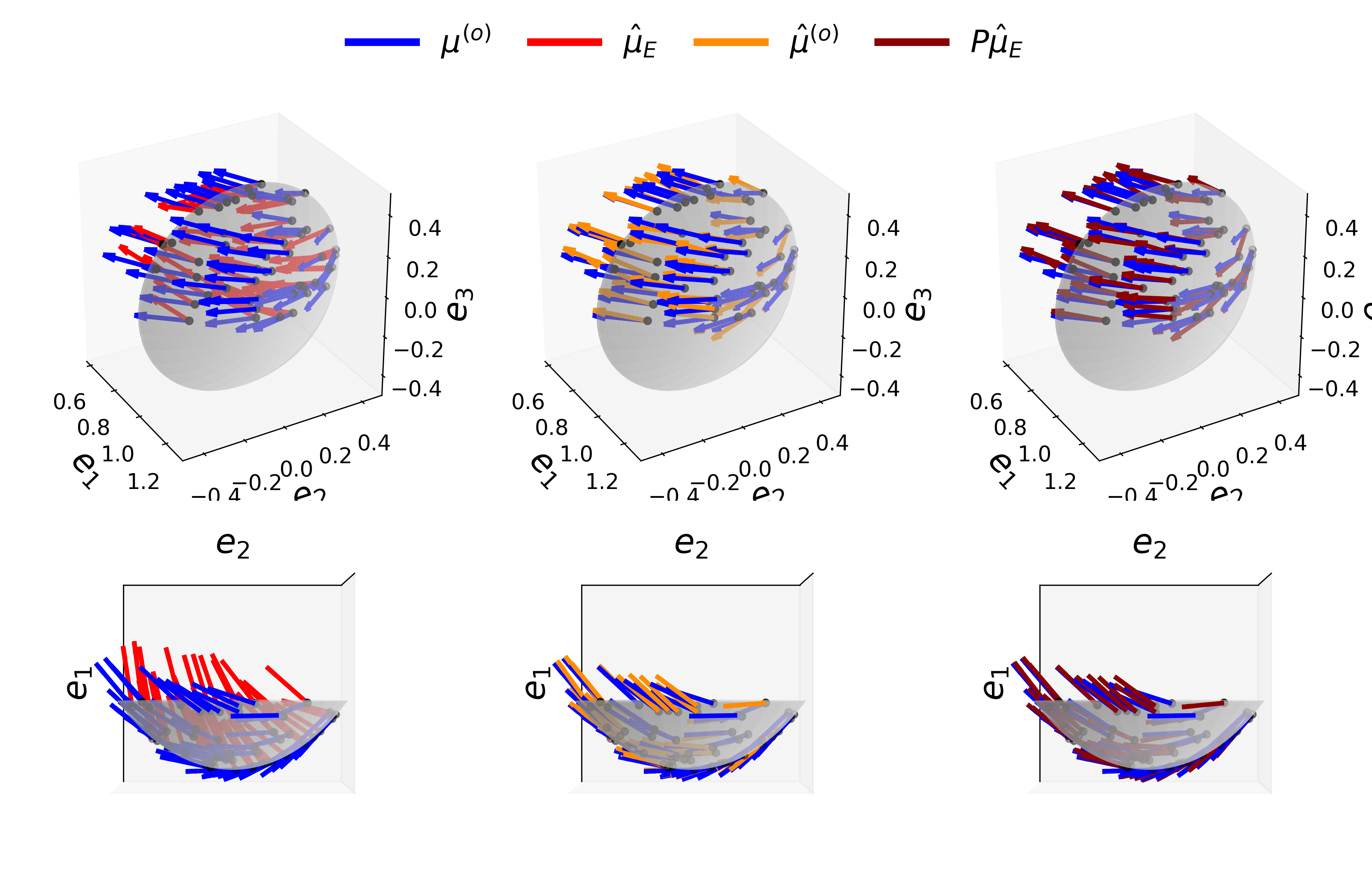}
    \caption{From left to right: visualizations of $\hat{\mu}_E(x), \hat{\mu}^{(\texttt{o})}(x),$ and $P_x \hat{\mu}_E(x)$, where $P_x$ is the projection operator onto the tangent space $T_xM$, for base-point samples $x$ drawn uniformly from a spherical cap centered at $(1,0,0)^\top$ and observed on ellipsoids with eccentricity $(1.5,1,1)$, shown from two viewing angles. The ground-truth drift vector is superimposed as blue arrows.}
    \label{fig:IB1-12}
\end{figure}

\begin{figure}[hbt!]
    \centering
    \includegraphics[width=0.9\textwidth]{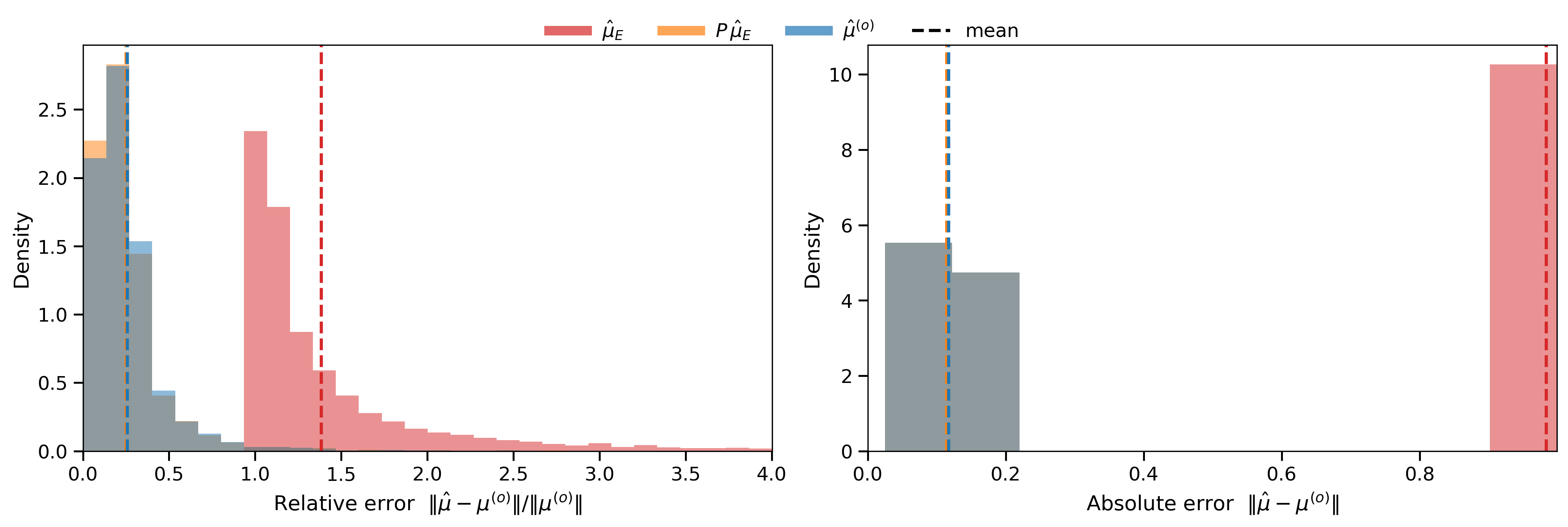}
      \includegraphics[width=0.9\textwidth]{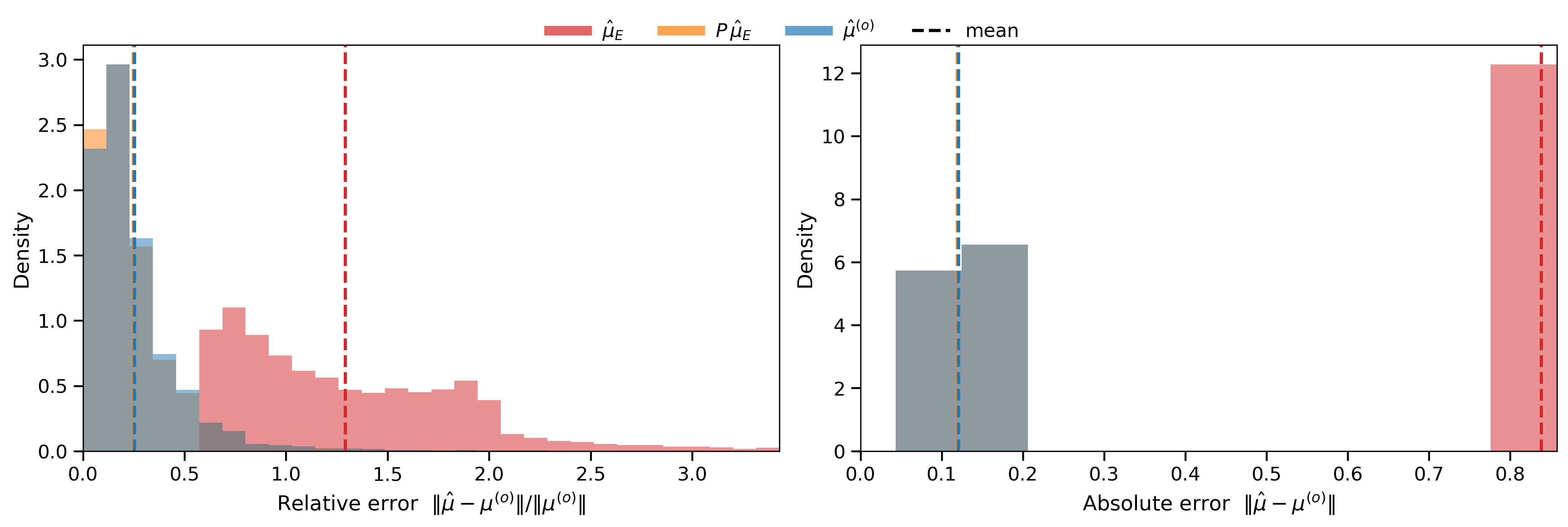}
         \includegraphics[width=0.9\textwidth]{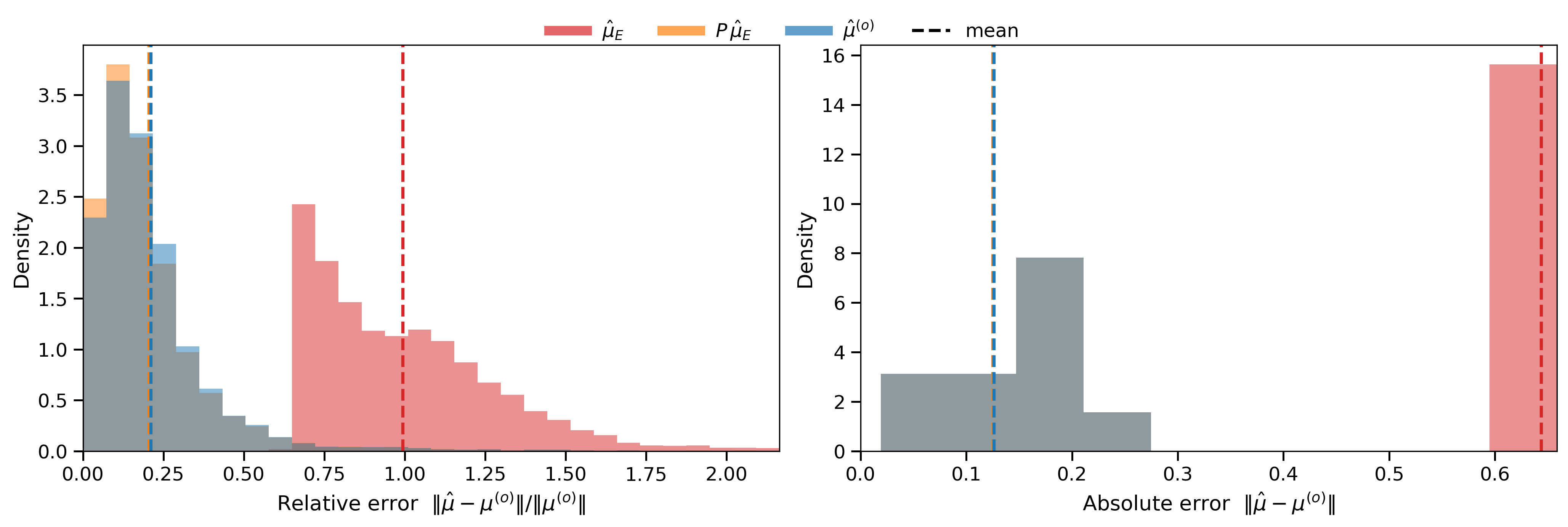}
   \caption{Histograms associated with Table \ref{tab:mag-and-error-plots-obs-ell-table}. Left: means and standard deviations of NRMSE of various drift estimators at $x$ where {$\|\mu^{(\texttt{o})}(x)\| / \|\mu^{(\texttt{o})}\|_\infty> 0.05$}. Right: means and standard deviations of RMSE of various estimators when  $x$ where {$\|\mu^{(\texttt{o})}(x)\| / \|\mu^{(\texttt{o})}\|_\infty \le 0.05$}.}
    \label{fig:obs-ell-hist}
\end{figure}

\begin{figure}[hbt!]
    \centering
            \includegraphics[width=0.75\textwidth]{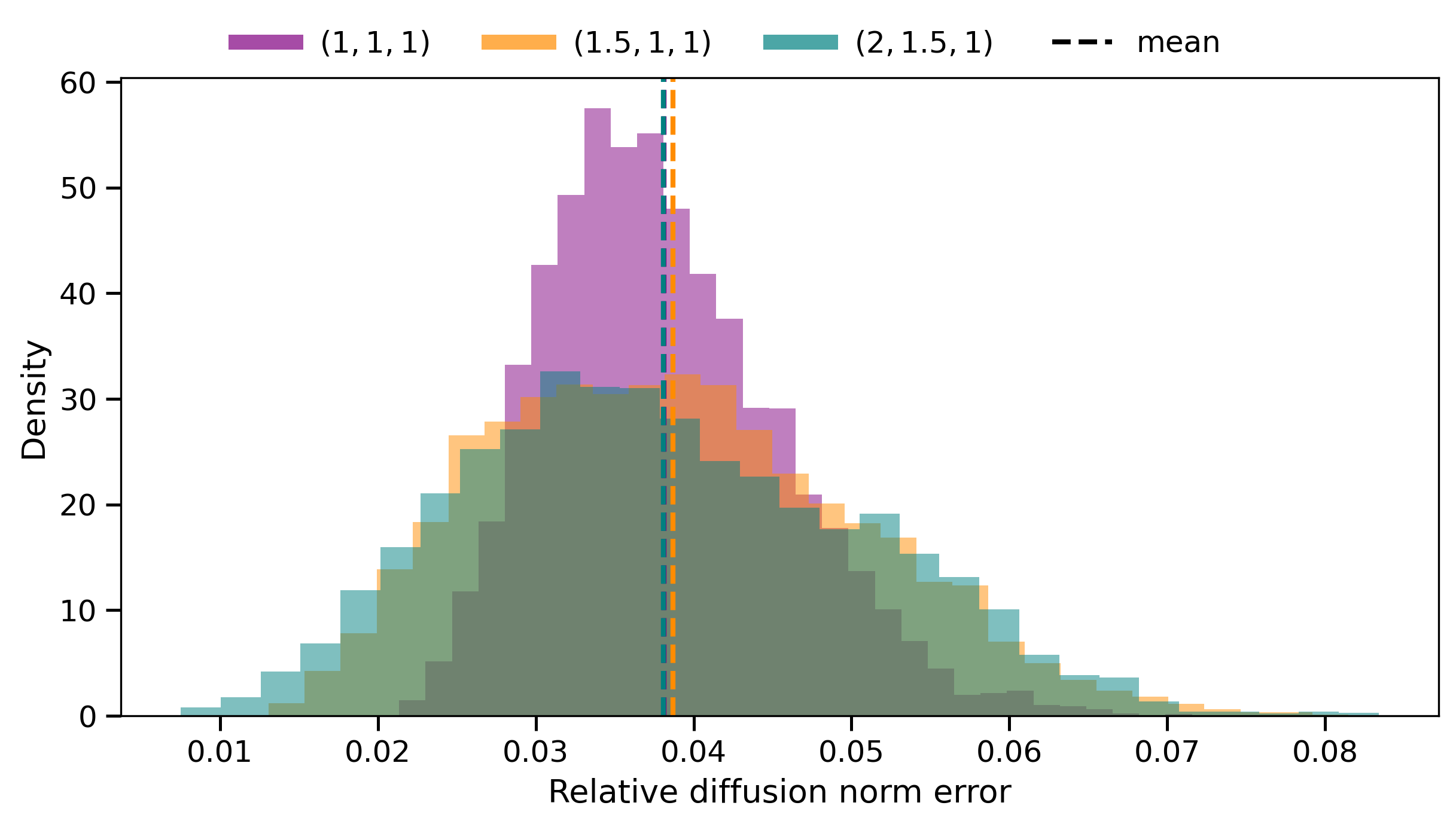}
   \caption{ \raggedright Histograms associated with the error of $\pi^{(\texttt{o})}$ in Table \ref{tab:mag-and-error-plots-obs-ell-table}.}
    \label{fig:obs-ell-hist-2}
\end{figure}

\begin{figure}[hbt!]
    \centering
    \includegraphics[width=0.75\textwidth]{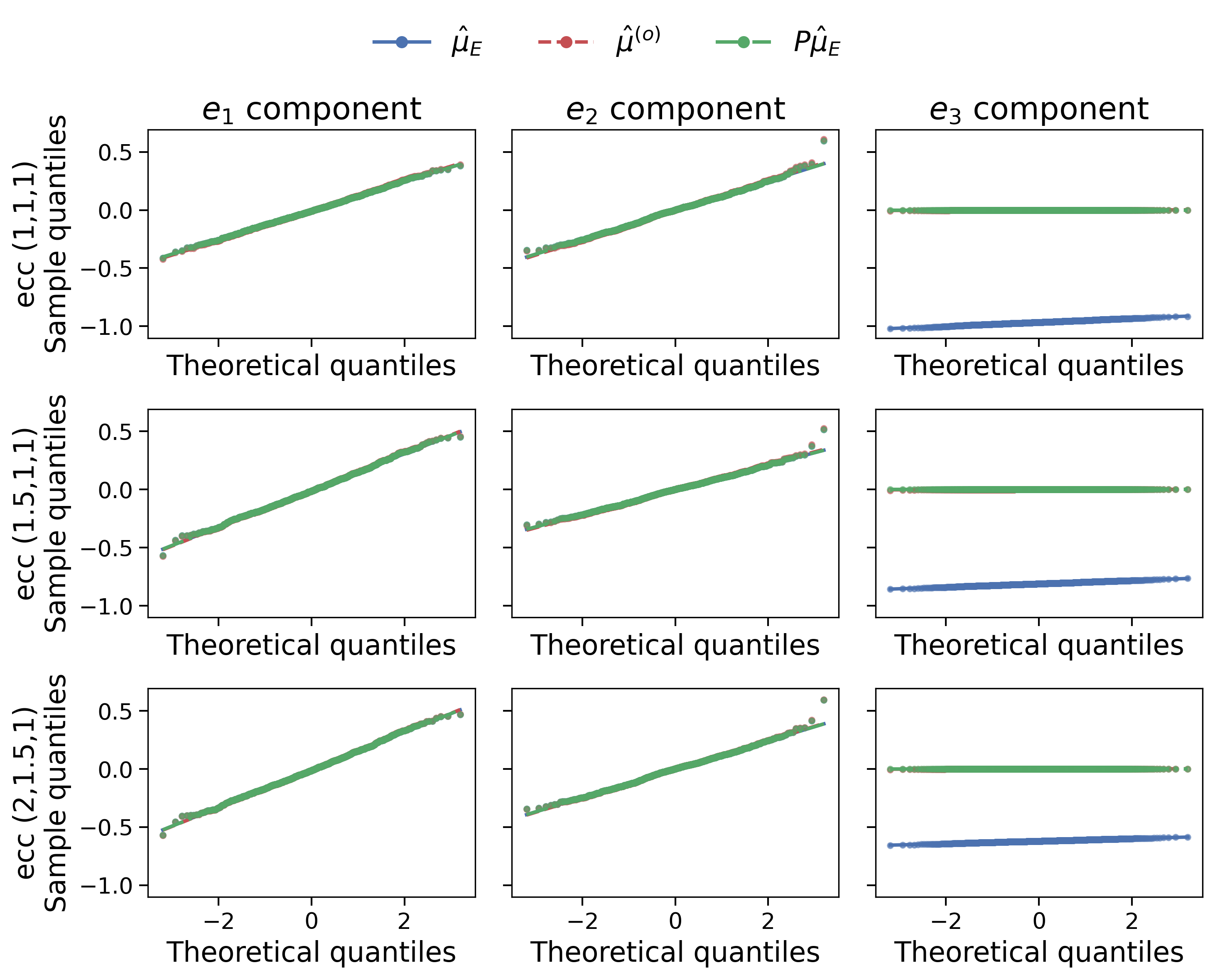}
    \caption{QQ plots corresponding to the drift estimation errors shown in Figure~\ref{fig:IA3a}, assessing agreement with the theoretical Gaussian reference; the normal component after estimated projection to tangent space.}
    \label{fig:IA2}
\end{figure}

\begin{figure}[hbt!]
    \centering
    \includegraphics[width=0.75\textwidth]{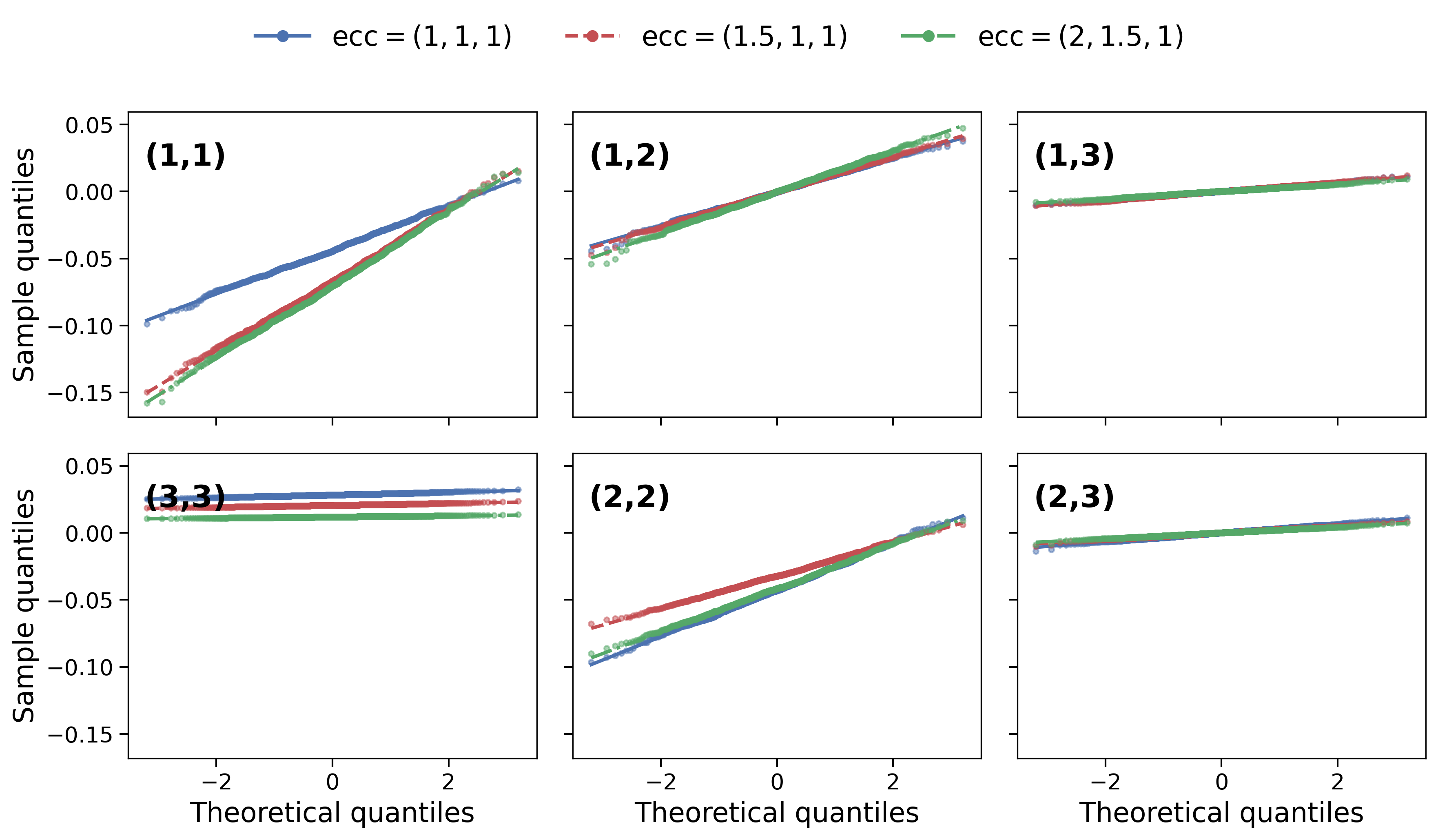}
    \caption{QQ plots corresponding to the vector field coordinate $(i,j)$ errors for $\hat{\pi}^{(\texttt{o})}(x)$ of Figure~\ref{fig:IA3b}.}
    \label{fig:IA2-22}
\end{figure}

\begin{table}[t]
\renewcommand{\arraystretch}{1.15}
\begin{tabular}{lcccc}
\toprule
Error type & $\hat{\mu}_E$ & $\hat{\mu}^{(\texttt{o})}$ & $P_x\hat{\mu}_E$ & $\hat{\pi}^{(\texttt{o})}$ \\
\midrule
$\frac{\|\hat\mu(x)-\mu^{(\texttt{o})}(x)\|}{\|\mu^{(\texttt{o})}(x)\|}$ ($\frac{\|\mu^{(\texttt{o})}(x)\|}{\|\mu^{(\texttt{o})}\|_\infty}\ge 0.05$)       & $1.530 \pm 1.120$  & $0.259 \pm 0.247$ & $0.252 \pm 0.244$ & --- \\
$\frac{|\|\hat\mu(x)\|-\|\mu^{(\texttt{o})}(x)\||}{\|\mu^{(\texttt{o})}(x)\|}$ ($\frac{\|\mu^{(\texttt{o})}(x)\|}{\|\mu^{(\texttt{o})}\|_\infty}\ge 0.05$) & $0.856 \pm 1.050$ & $0.162 \pm 0.184$ & $0.157 \pm 0.179$ & --- \\
$\Theta(\hat\mu(x),\mu^{(\texttt{o})}(x))$ ($\frac{\|\mu^{(\texttt{o})}(x)\|}{\|\mu^{(\texttt{o})}\|_\infty}\ge 0.05$)               & $0.936 \pm 0.191$& $0.181 \pm 0.251$ & $0.177 \pm 0.252$ & --- \\
$\|\hat\mu(x)-\mu^{(\texttt{o})}(x)\|$   ($\frac{\|\mu^{(\texttt{o})}(x)\|}{\|\mu^{(\texttt{o})}\|_\infty}< 0.05$)                    & $0.984 \pm 0.014$ & $0.117 \pm 0.053$ & $0.114 \pm 0.051$ & --- \\
$\frac{\|\hat\pi(x)-\pi^{(\texttt{o})}(x)\|_F}{\|\pi\|_F}$          & --- & --- & --- & $0.038 \pm 0.007$ \\
  $\|\sin\Theta_2\|_F$        & --- & --- & --- & $0.017 \pm 0.006$ \\

\midrule

$\frac{\|\hat\mu(x)-\mu^{(\texttt{o})}(x)\|}{\|\mu^{(\texttt{o})}(x)\|}$ ($\frac{\|\mu^{(\texttt{o})}(x)\|}{\|\mu^{(\texttt{o})}\|_\infty}\ge 0.05$)       & $1.410 \pm 0.976$ & $0.257 \pm 0.267$ & $0.250 \pm 0.264$ & --- \\
$\frac{|\|\hat\mu(x)\|-\|\mu^{(\texttt{o})}(x)\||}{\|\mu^{(\texttt{o})}(x)\|}$ ($\frac{\|\mu^{(\texttt{o})}(x)\|}{\|\mu^{(\texttt{o})}\|_\infty}\ge 0.05$) & $0.764 \pm 0.883$& $0.174 \pm 0.209$ & $0.170 \pm 0.205$ & --- \\
$\Theta(\hat\mu(x),\mu^{(\texttt{o})}(x))$ ($\frac{\|\mu^{(\texttt{o})}(x)\|}{\|\mu^{(\texttt{o})}\|_\infty}\ge 0.05$) & $0.887 \pm 0.25$ & $0.165 \pm 0.248$ & $0.161 \pm 0.249$ &  \\
$\|\hat\mu(x)-\mu^{(\texttt{o})}(x)\|$   ($\frac{\|\mu^{(\texttt{o})}(x)\|}{\|\mu^{(\texttt{o})}\|_\infty}< 0.05$)                    & $0.838 \pm 0.015$ & $0.121 \pm 0.049$ & $0.119 \pm 0.049$ & --- \\
 $\frac{\|\hat\pi(x)-\pi^{(\texttt{o})}(x)\|_F}{\|\pi\|_F}$          & --- & --- & --- & $0.039 \pm 0.012$ \\
$\|\sin\Theta_2\|_F$       & --- & --- & --- & $0.016 \pm 0.009$ \\
\bottomrule
\end{tabular}
\caption{\raggedright Same as Table \ref{tab:mag-and-error-plots-obs-ell-table} about summary of various evaluation metrics, but with eccentricities $(1,1,1)$ (top block) and $(1.5,1,1)$ (bottom block). $\hat{\mu}(x)$ is the estimator of $\mu^{(\texttt{o})}(x)$, which can be $\hat{\mu}_E$, $\hat{\mu}^{(\texttt{o})}$, or $P_x\hat{\mu}_E$, where $P_x$ is the projection to $T_xM$, listed in the top. $\hat\pi(x)$ is the estimator of $\pi^{(\texttt{o})}(x)$, which is $\hat{\pi}^{(\texttt{o})}(x)$. $\Theta(\mu^{(\texttt{o})},\hat\mu)$ is the angle between $\mu^{(\texttt{o})}$ and $\hat\mu$ with the unit radian.  $\|\sin\Theta\|_F$ is the subspace distance between the dominant 2D eigenspaces of $\hat\pi^{(\texttt{o})}$ and $\pi^{(\texttt{o})}$.
}
\label{tab:mag-and-error-plots-obs-ell-table-2}
\end{table}

\begin{figure}[hbt!]
    \centering
\includegraphics[width=1\textwidth]{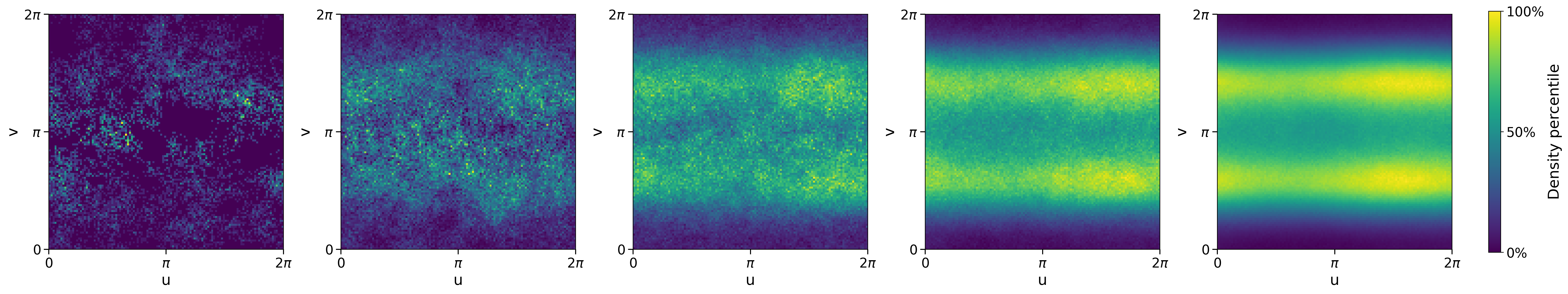}
     \caption{A single long trajectory of length $n = 10^8$ with time-step $\Delta = 10^{-2}$ is run on the parameter space $N = [0,2\pi) \times [0, 2\pi)$ and mapped to the Klein Bottle in $\mathbb{R}^4$, where the occupation density $\hat{L}^{(\texttt{o})}$ is plotted in $N$, where $\log_{10}(n_i) \in \{4, 5,6,7,8\}$, from left to right column. In each image, bins are colored according to the percentile of their density values within that plot.}
    \label{fig:IC1-2}
\end{figure}

\begin{figure}[hbt!]
    \centering
    \includegraphics[width=0.45\textwidth]{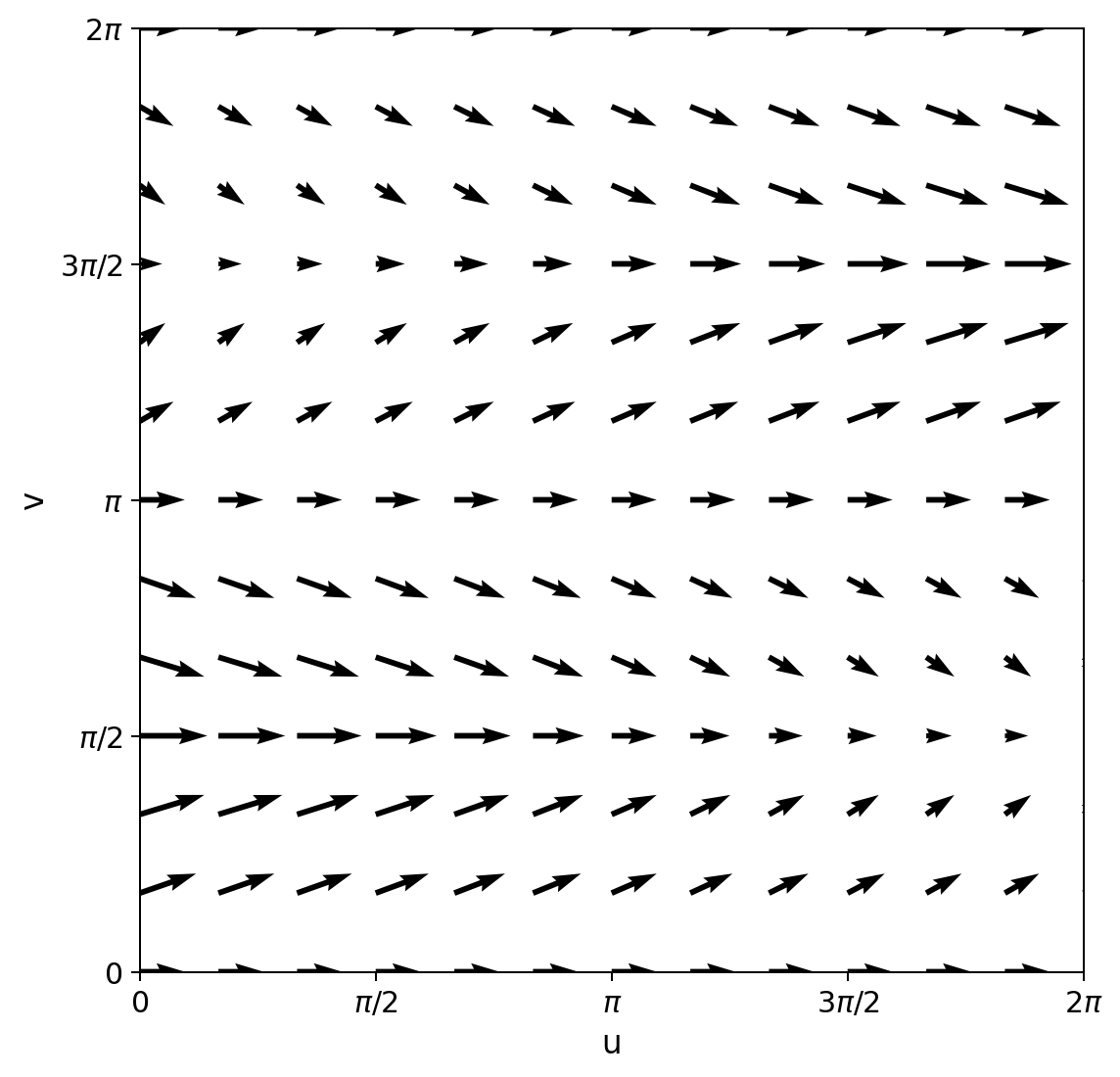}
     \caption{The drift vector field $\mu(u,v) = (1 + \frac{1}{2}\cos\left(\frac{u}{2} \right)\sin(v), \frac{1}{2}\sin\left(2v\right))^\top$ on $N = [0,2\pi) \times [0, 2\pi)$ used to generate the dynamics on the Klein bottle.}
    \label{fig:KB_drift}
\end{figure}

\begin{figure}[hbt!]
    \centering
    \includegraphics[width=0.6\textwidth]{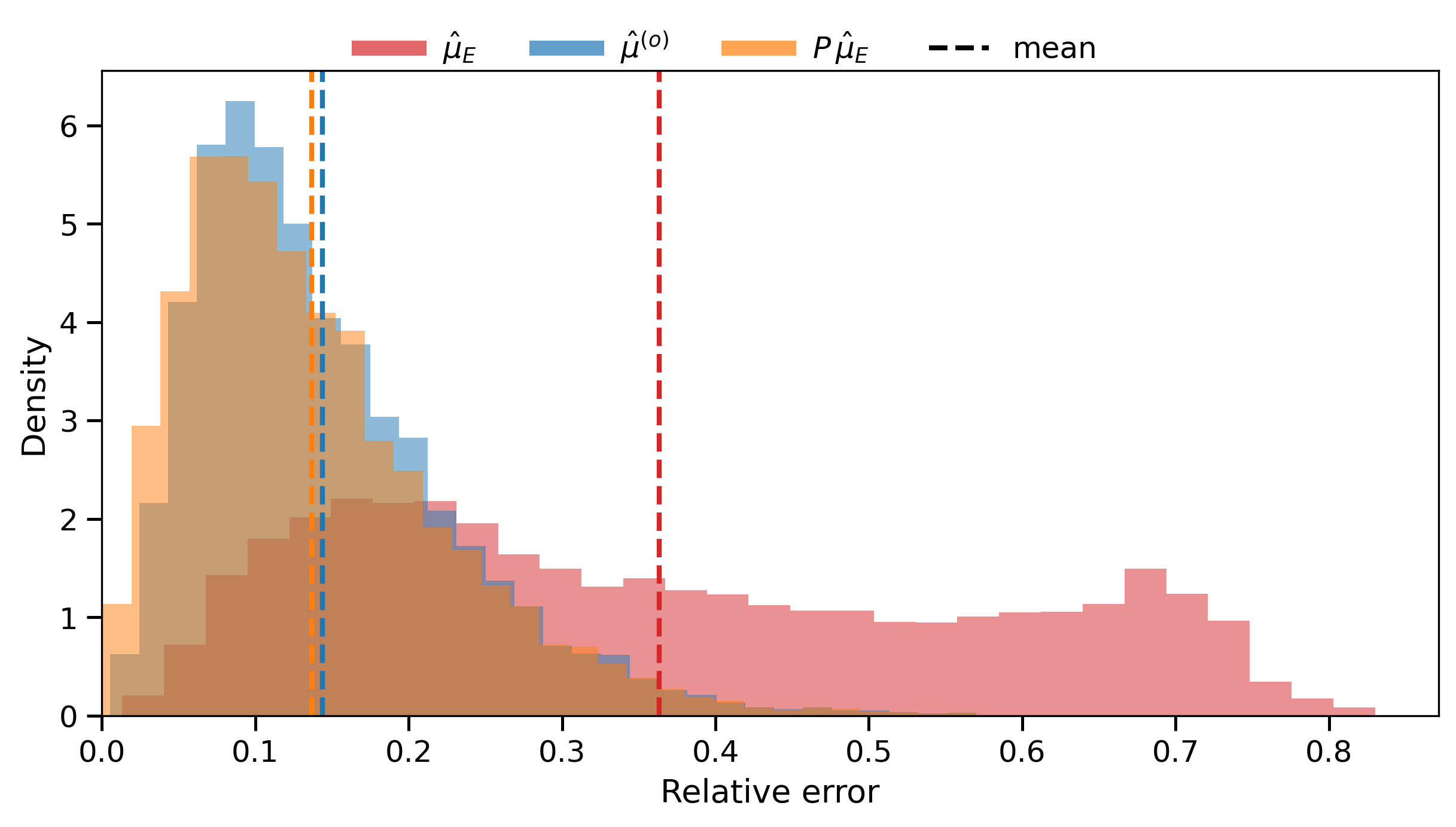}
 \includegraphics[width=0.6\textwidth]{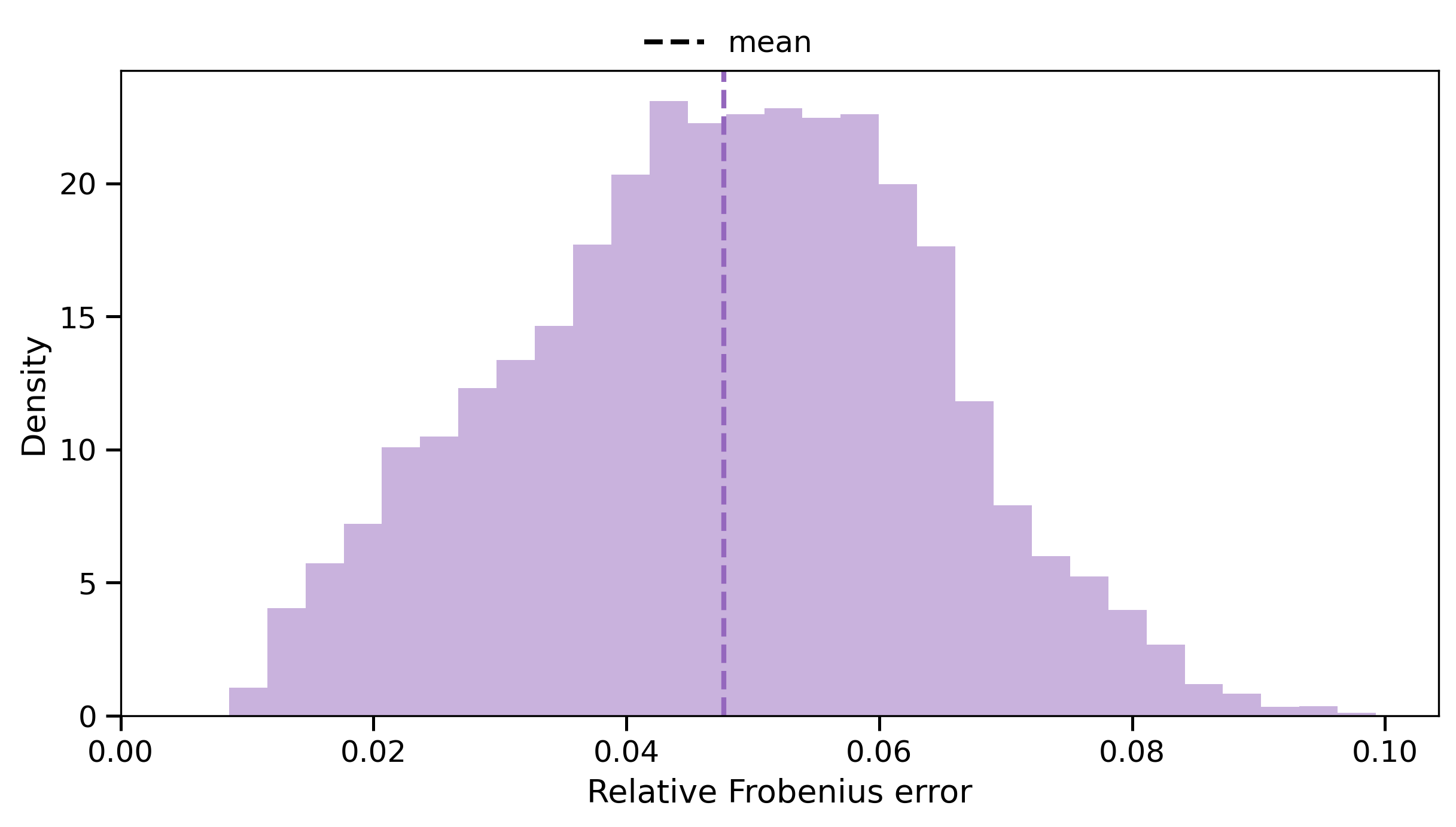}
     \caption{Top: Histogram of NRMSE of various estimators on the Klein Bottle. Bottom: Histogram of NRMSE of diffusion estimator.}
    \label{fig:IC2-2}
\end{figure}

\clearpage

\begin{table}
\begin{center}
\begin{longtable}{ |c|p{12cm}| } 
 \hline
$b_{\mu}^{(\texttt{o})}$ & bias of observed drift vector estimator \\
$b_\pi^{(\texttt{o})}$ & bias of observed diffusion matrix estimator \\
$C_X$ & process specific constant \\
$D(\mathbb{R}_+,\mathbb{R})$ & the Skorohod space with Borel $\sigma$-algebra and canonical filtration \\
$d$ & dimension of the manifold \\
$\Delta$ & sampling period \\
$\mathcal{D}_x$ & a distance-like function at $x\in M$ \\
$e(X)$ & explosion time of manifold-valued diffusion \\
$\mathbb{E}_x$ & expectation conditional on $X_0 = x$ \\
$\Gamma$ & Gamma function \\
$h$ & bandwidth \\
$I, I_{p \times p}$ & the identity matrix (of size $p \times p$) \\
$\iota$ & an embedding of $M$ into $\mathbb{R}^p$ \\
$\mathbf{1}_A$ & the indicator function of the set $A$ \\
$\mathcal{K}$ & a kernel function \\
$\kappa_{p,q}$ & kernel-dependent constant for $p \in \mathbb{N}$, $q \in \{0\}\cup \mathbb{N}$ \\
$\mu^{\texttt{(o)}}$ & observed drift vector \\
$n$ & number of process observations \\
$p$ & ambient space dimension\\
$\pi^{\texttt{(o)}}$ & observed diffusion matrix \\
$R_{m}$ & $m$-th regeneration time of generalized life-cycle decomposition \\
$S_m$ & $m$-th subsequent regeneration time of generalized life-cycle decomposition \\
$\sigma_\alpha$ & $\alpha$-th component diffusion vector field \\
$\sim$ & asymptotically equivalent (ratio $\to 1$) \\
$T$ & sampling period  \\
$\theta_s$ & shift operator for time $s \ge 0$ \\
$\tilde{g}_{s,t}(a, b)$ & $\tilde{\phi}_{s,t}(a,b) - \tilde{\phi}_{s}(a)\tilde{\phi}_{t}(b)$, a measure of path dependence \\
$\Upsilon$ & the scaling factor of the process $X_t$ \\
$\mathcal{W}(M)$ & $C([0, \infty), M)$ \\
$\mathcal{W}^r_0$ & $\{\omega \in C([0, \infty), \mathbb{R}^r): \omega(0) = 0\}$ \\
$\mathcal{B}_t(\mathcal{W}(M))$ & the $\sigma$-algebra generated by cylinder sets of $\mathcal{W}(M)$ up to time $t>0$ \\
$\mathcal{B}_t(\mathcal{W}^r_0)$ & the $\sigma$-algebra generated by cylinder sets of $\mathcal{W}^r_0$ up to time $t>0$ \\
$W_t, W^\alpha_t$ & standard $r$-dimensional Brownian motion and its components \\
$\mathcal{X}$ & high-dimensional time series \\
$\mathsf A$ & recurrent atom \\
$\mathsf T_m$ & exponential jump-times \\
 \hline
\end{longtable}
\end{center}
\caption{Commonly used notation.\label{Table: more symbols}}
\end{table}

\fi

\end{document}